\documentclass[journal,onecolumn]{IEEEtran}
\usepackage{graphicx,amsthm,amsmath,amsfonts,latexsym,amssymb}
\usepackage{color,array,verbatim,algpseudocode,bm}
\definecolor{menucolor}{rgb}{0.1,0.52,0.47}
\definecolor{urlcolor}{rgb}{0.85,0.37,0.01}
\definecolor{runcolor}{rgb}{0.46,0.44,0.701}
\definecolor{filecolor}{rgb}{0.2,0.5,0.01}
\definecolor{linkcolor}{rgb}{0.12,0.47,0.70}
\definecolor{citecolor}{rgb}{0.55,0.36,0.01}
\definecolor{anchorcolor}{rgb}{0.4,0.4,0.4}
\usepackage[colorlinks=true, urlcolor=urlcolor, linkcolor=linkcolor,citecolor=citecolor,filecolor=filecolor, anchorcolor=anchorcolor,menucolor=menucolor]{hyperref}
\usepackage{nameref}
\usepackage{zref-xr}
\usepackage{cite}
\zxrsetup{
  tozreflabel=false,
  toltxlabel=true,  
}
\usepackage[shortlabels]{enumitem}
\ifCLASSOPTIONcompsoc
  \usepackage[caption=false,font=normalsize,labelfont=sf,textfont=sf]{subfig}
\else
\usepackage[caption=false,font=footnotesize]{subfig}
\fi
\usepackage[mathscr]{euscript}
\usepackage{amsmath,amssymb}
\usepackage{enumitem} 
\usepackage[linesnumbered,ruled]{algorithm2e}
\renewcommand{\hat}{\widehat}
\newcommand{\Var}{{\mathbb V}\mbox{ar}}
\newcommand{\mX}{\mathcal X}
\newcommand{\mC}{\mathcal C}
\newcommand{\mD}{\mathcal D}
\newcommand{\mY}{\mathcal Y}
\newcommand{\mS}{\mathcal S}
\newcommand{\mB}{\mathcal B}
\newcommand{\mG}{\mathcal G}
\newcommand{\mP}{\mathcal P}
\newcommand{\mZ}{\mathcal Z}
\newcommand{\mM}{\mathcal M}
\newcommand{\mL}{\mathcal L}
\newcommand{\mE}{\mathcal E}
\newcommand{\mV}{\mathcal V}
\newcommand{\mH}{\mathcal H}

\newcommand{\mW}{\mathcal W}
\newcommand{\mU}{\mathcal U}
\newcommand{\mO}{\mathcal O}

\newcommand{\mhB}{\mathcal{ \hat B}}
\newcommand{\mhZ}{\mathcal{ \hat Z}}
\newcommand{\mhV}{\mathcal{ \hat V}}
\newcommand{\mbX}{\mathcal{ \bar X}}
\newcommand{\mbY}{\mathcal{ \bar Y}}

\newcommand{\bmbY}{\boldsymbol{\mathcal{ \bar Y}}}
\newcommand{\bX}{\boldsymbol{X}}
\newcommand{\bY}{\boldsymbol{Y}}
\newcommand{\bw}{\boldsymbol{w}}
\newcommand{\bbm}{\boldsymbol{m}}

\newcommand{\bs}{\boldsymbol{s}}
\newcommand{\be}{\boldsymbol{e}}
\newcommand{\bj}{\boldsymbol{j}}
\newcommand{\by}{\boldsymbol{y}}
\newcommand{\bx}{\boldsymbol{x}}
\newcommand{\bi}{\boldsymbol{i}}
\newcommand{\bmX}{\boldsymbol{\mathcal X}}
\newcommand{\bmY}{\boldsymbol{\mathcal Y}}
\newcommand{\bmZ}{\boldsymbol{\mathcal Z}}
\newcommand{\bmE}{\boldsymbol{\mathcal E}}
\newcommand{\bmS}{\boldsymbol{\mathcal S}}

\DeclareMathOperator{\tr}{tr}

\DeclareMathOperator{\vecc}{vec}
\DeclareMathOperator{\vech}{vech}
\DeclareMathOperator{\N}{\mathcal N}
\DeclareMathOperator{\I}{I}
\DeclareMathOperator*{\circs}{\circ}
\DeclareMathOperator*{\argmax}{arg\,max}
\DeclareMathOperator*{\Ast}{\ast}

\newcommand{\hM}{ \hat M}
\newcommand{\hL}{ \hat L}
\newcommand{\hSigma}{ \hat \Sigma}
\newcommand{\hsigma}{ \hat \sigma}
\newcommand{\hlambda}{ \hat \lambda}
\newcommand{\tuckL}{ [\![}
\newcommand{\tuckR}{ ]\!]}

\newcommand{\gn}{\~{n}}
\newcommand{\eleven}{11}

\newcommand\sdots{\!\hbox to 1em{.\hss.\hss.}}

\newtheorem{defn}{Definition}[section] 
\newtheorem{thm}{Theorem}[section] 
\newtheorem{cor}{Corollary}[section] 
\newtheorem{lemma}{Lemma}[section] 
\newtheorem{property}{Property}[section] 

\newcommand{\citep}{\cite}
\newcommand{\citet}{\cite}
\begin{document}
\title{Reduced-Rank  Tensor-on-Tensor Regression and Tensor-variate Analysis of Variance}
\author{Carlos~Llosa-Vite~and~Ranjan~Maitra
  \thanks{C. Llosa-Vite and R. Maitra are with the Department of Statistics
at Iowa State University, Ames, Iowa 50011, USA. e-mail:
\{cllosa,maitra\}@iastate.edu.}
}




\maketitle
\begin{abstract}
     Fitting regression models with many multivariate responses and  covariates can be challenging, but such responses and covariates sometimes have tensor-variate structure. We  extend the classical multivariate regression model to exploit such structure  in two ways: first, we impose four types of low-rank tensor formats on the  regression coefficients. Second, we model the errors using the tensor-variate normal distribution that imposes a Kronecker separable format on the covariance matrix. We obtain maximum likelihood estimators via block-relaxation algorithms and derive their computational complexity and asymptotic distributions. Our regression framework enables us to formulate  tensor-variate analysis of variance (TANOVA) methodology. This methodology, when applied in a one-way TANOVA layout, enables us to identify  cerebral regions significantly associated with the interaction of  suicide attempters or non-attemptor ideators and positive-, negative- or death-connoting words in a functional Magnetic Resonance Imaging study. Another application uses three-way TANOVA on the Labeled Faces in the Wild image dataset to distinguish facial characteristics related to ethnic origin, age group and gender.

   \end{abstract}

\maketitle

\section{Introduction}\label{sec:introduction}
The classical simple linear regression (SLR) model (without intercept) relates the response variable $y_i$ to the explanatory variable $x_{i}$  as
$
y_i = \beta x_{i} +e_i$ with $\Var(e_i) = \sigma^2$
for $i =1,2,\sdots,n$, 
where $\beta$ is the regression coefficient parameter and $\sigma^2$
is the variance parameter. A natural extension of SLR
for vector-valued responses and explanatory variables is the
multivariate multiple linear regression (MVMLR) model
\begin{equation}\label{eq:MLR}
\boldsymbol{y}_i = B \boldsymbol{x}_{i} +\boldsymbol{e}_i
,\quad
\Var (\boldsymbol{e}_i) = \Sigma,
\end{equation}
where 
$(\boldsymbol{y}_1,\boldsymbol{x}_1),(\boldsymbol{y}_2,\boldsymbol{x}_2),\sdots,(\boldsymbol{y}_n,\boldsymbol{x}_n)$
are vector-valued responses and covariates and $(B,\Sigma)$ are parameters. The number of
parameters relative to the sample size in \eqref{eq:MLR} is greater in
the MVMLR model that in its SLR
counterpart because the parameters $(B,\Sigma)$ are
matrix-valued~\citep{johnsonandwichern08}. Several methodologies, for
example, the lasso and graphical lasso
\citep{tibshirani96,friedmanetal08}, envelope models
\citep{cooketal10} and reduced rank regression
\citep{mukherjeeandzhu11}, have been proposed to alleviate issues
arising from the large
number of parameters in \eqref{eq:MLR}. However, these methodologies
are only for vector-valued observations and do not exploit their underlying structure that may further reduce the number of
necessary parameters, in some cases making computation
feasible. Here we consider \textit{tensor}- or array-structured
responses and covariates that arise in many applications, such as the 
,
two motivating examples introduced next.

\subsection{Illustrative Examples}
\label{sec:examples}
\subsubsection{Cerebral activity in subjects at risk of suicide}\label{subseq:suicide} The United States Center for Disease 
Control and Prevention reports that 47173 Americans died  by
suicide in 2017, accounting for about two-thirds of all homicides in that year. 
Accurate suicide risk assessment is challenging, even for trained
mental health professionals, as 78\% of patients who commit
suicide deny such ideation even in their last communication with professionals~\citep{buschetal03}. 
Understanding how subjects at risk of suicide respond to different stimuli is important to guide treatment and therapy.
\citet{justetal17} provided data from a functional
Magnetic Resonance Imaging (fMRI) study of nine suicide attempters and eight suicide
non-attempter ideators (henceforth, ideators) upon
exposing them to ten words each with positive, negative or death-related
connotations. 
Our interest is in understanding brain regions associated with the
interaction of a subject's attempter/ideator status and word type
to inform diagnosis and treatment.  

Traditional approaches fit separate regression models at each voxel
without regard to spatial context that is only addressed  
post hoc at the time of inference. 
A more holistic strategy would use \eqref{eq:MLR} with the response
vector $\boldsymbol{y}_i$ of size $30\!\times\!43\!\times\!
56\!\times\! 20\!=\!1444800$, which contains thirty $43\!\times\!56\!\times\! 20$ image volumes, for each of the $i=1,2,\sdots,17$ subjects. The
explanatory variable here is a 2D 
vector that indicates a subject's status as a suicide attempter or
ideator. Under this framework,  $B$ and $\Sigma$ have over 2.8 million and 1 trillion unconstrained parameters, making estimation with
only $17$ subjects impractical. Incorporating a 3D spatial autoregressive (AR)
structure into the image volume, and another correlation structure
between the words can allow estimation of the variance, but still needs additional methodology to accomodate the large sixth-order tensor-structured regression parameter $B$. We develop such methodology in this paper, and return to
this dataset in  Section~\ref{application:suicide}. 
\subsubsection{Distinguishing characteristics of  faces}\label{subseq:lfw}
Distinguishing the visual characteristics of faces is important for
biometrics. The Labeled Faces in the Wild (LFW)
database \citep{Huangetal07} is used for developing
and testing facial recognition methods, and contains over
13000  $250\times250$ color 
images of faces of different individuals along with their
classification  into  ethnic origin,  age group and gender~\citep{AfifiandAbdelhamed19,Kumaretal09}.
We use a subset, totalling 605 images, for which the three attributes of
ethnic origin (African, European or Asian, as specified in the database),
cohort (child, youth, middle-aged or senior) and gender (male or
female) are unambiguous. The color at each pixel is a 3D RGB vector 
so each image (response) is a $250\!\times\! 250\!\times\! 3$
array. The three image attributes can each be 
represented by an indicator vector, so  the covariates (attributes)
have a three-way tensor-variate structure.  Our objective
is to train a linear model to help us distinguish the visual characteristics of
different attributes. Transforming the 3D tensors into 
vectors and fitting~\eqref{eq:MLR} requires a $B$ of 
$250\!\times\!250\!\times\!3\!\times\!2\!\times\!3\!\times\!4$ or 4.5 million
unconstrained parameters and an error covariance matrix $\Sigma$ of
over 17 billion similar parameters, making  accurate inference (from
only 605 observed images) impractical. Methodology that incorporates
the reductions  afforded by the tensor-variate structures of the
responses and the covariates can  redeem the situation. We revisit
this dataset in Section~\ref{application:LFW}. 

\subsection{Related work and overview of our contributions}
The previous examples show the inadequacy of training
\eqref{eq:MLR} on tensor-valued data without additional accommodation
for structure, as the sizes of $B$ and $\Sigma$ in unconstrained
vector-variate regression grow with the dimensions of the
tensor-valued responses and explanatory 
variables. Several regression frameworks~\citep{bietal21,panagakisetal21}
that efficiently allow for tensor
responses~\citep{hoff14,sunandli16,liandzhang17} or
covariates~\citep{guoetal12,fuetal14,zhouetal13,xiaoshanetal18,liandzhang21,rabusseauandkadri16,guhaniyogietal17,ahmedetal20,kongetal20,zhouetal21,dengetal21,poythressetal21,papadogeorgouetal21} have recently been considered.
Tensor-on-tensor regression (ToTR) refers to the case where both the
response and covariates are tensors. In this context, \citet{hoff14}
proposed an outer matrix product (OP) factorization of $B$,
\citet{lock17} suggested {\em canonical polyadic} or CANDECOMP/PARAFAC
(CP) decomposition~\citep{carrollandchang70,
  harshman70}, while \citet{liuetal20} and \citet{llosa18}
factorized $B$ using a tensor ring (TR)~\citep{zhaoetal16b, oseledets11} format. 
Finally, \citet{brandianddimatteo21} considered a Tucker
(TK)~\citep{tucker66,koldaandbader09,bahrietal017} framework but
failed to include necessary constraints in their estimation algorithm.
The CP, TR and OP formats on $B$ allow for quantum dimension reduction
without affecting prediction or discrimination ability of the
regression model. However, these methodologies do not account for
dependence  within tensor observations,  the sampling distribution of
their estimated coefficients and the natural connection that
exists between ToTR and the related analysis of variance
(ANOVA). Here, we propose a general ToTR framework that renders four
low-rank tensor formats on the coefficient $B$: CP, TK,
TR  and the OP, while simultaneously allowing the errors to follow a
\textit{tensor-variate normal (TVN)  distribution}
\citep{hoff11,akdemirandgupta11,ohlsonetal13,manceuranddutilleul13}
that posits a Kronecker structure on the $\Sigma$ of~\eqref{eq:MLR}. 
Assuming TVN-distributed errors allows us to obtain the sampling
distributions of the estimated coefficients under their assumed
low-rank format. Indeed, Section \ref{application:suicide} uses our
derived sampling distributions to produce statistical parametric maps
to help detect significant neurological interactions between
death-related words and suicide attempter/ideator status.
The TVN assumption on the errors also allows us to consider dependence
within the tensor-valued observations. The Kronecker structured
$\Sigma$ in the TVN model renders a different covariance matrix for
each tensor dimension, allowing us to 
simultaneously study multiple dependence contexts within the same
framework. Here we also introduce the notion of tensor-variate ANOVA
(TANOVA) under the ToTR framework, which is analogous to ANOVA and
multivariate ANOVA (MANOVA) being instances of
multiple linear regression (MLR) and MVMLR.

The rest of the paper is structured as follows. 
Section~\ref{sec:methodology} 
first presents notations and network diagrams, low-rank tensor 
formats, the TVN distribution and our preliminary results that we develop
for use in this paper. We formulate  ToTR and TANOVA
methodology with low-rank tensor formats on the covariates, and TVN
errors. We provide algorithms for finding  maximum likelihood (ML)
estimators and study their properties.
Section~\ref{sec:simulation} evaluates performance of our methods in
two simulation scenarios while Section~\ref{sec:data_application}
applies our methodology to the motivating applications of
Sections~\ref{subseq:suicide} and  \ref{subseq:lfw}. We conclude with
some discussion. A supplementary appendix with sections,
theorems, lemmas, figures and equations prefixed with
``S'' is also available.

\section{Theory and Methods}\label{sec:methodology}
This section introduces a regression model with TVN errors and
tensor-valued responses and covariates.  We provide notations and
definitions, then introduce our models and develop algorithms for ML
estimation under the TK, CP, OP and TR low-rank formats. A special
case leads us to TANOVA. We also derive asymptotic properties of our
estimators and computational complexity of our algorithms.
\subsection{Background and preliminary results} \label{sec:preliminaries}
We provide a unified treatment of tensor reshapings
and contractions by integrating the work of \citet{kolda06},
\citet{koldaandbader09}, and \citet{cichockietal17} with our own
results that we use later. 
We use $\tr(.)$,  $(.)'$, and $(.)^-$ to denote the
trace, transpose, and pseudo-inverse, $I_n$ for  the $n\!\times\! n$
identity matrix and $\otimes$ for the Kronecker product
(Section~\ref{app:multilinear_statistics} has additional definitions
and illustrations). We define tensors as multi-dimensional arrays of numbers. The total number of {\em modes} or sides of a tensor is
called its {\em order}. We use  lower-case letters (\textit{i.e.} $x$)
to specify scalars, bold lower-case italics (\textit{i.e.}
$\boldsymbol{x}$) for vectors, upper-case  italics (\textit{i.e.}
$X$) for matrices, and calligraphic scripts
(\textit{i.e.} $\mX$) for higher-order tensors. Random matrices or
vectors are denoted using $\bX$ and random tensors by $\bmX$. We denote the
$(i_1,\sdots,i_p)th$ element of a $p$th order tensor $\mX$ using
$\mX(i_1,\sdots,i_p)$ or $\mX(\bi)$ where  $\bi =
[i_1,\sdots,i_p]'$. The vector outer product, with notation $\circ$, of $p$ vectors generates the $p$th-order tensor $\mX = \circs_{j=1}^p \boldsymbol{x}_j$ with $
(i_1,\sdots,i_p)$th element 
$\mX(\bi) = \prod_{j=1}^p \boldsymbol{x}_j(i_j).$
Any $p$th-order tensor $\mX \in \mathbb{R}^{m_1\times \sdots \times m_p}$ (or $\mathbb{R}^{\times_{j=1}^pm_j}$) can be expressed using the vector outer product as
\begin{equation}\label{tens}
\mX= \sum_{i_1=1}^{m_1} \sdots \sum_{i_p=1}^{m_p} \mX(i_1,\sdots,i_p)
\big( \circs_{q=1}^p \boldsymbol{e}_{i_q}^{m_q} \big),
\end{equation}
where $\boldsymbol{e}_{i}^{m}\in \mathbb{R}^{m}$ is a unit basis
vector with 1 as the $i$th element and 0 everywhere else. Equation \eqref{tens}
 allows us to reshape a tensor by only manipulating the vector
outer product. A $p$th-order diagonal tensor $\mathbb{I}^p_r
\in \mathbb{R}^{\times_1^p r}$ has ones where the indices in each mode
coincide, and zeroes  elsewhere, that is,
\begin{equation}\label{eq:diag}
\mathbb{I}^p_r
=
 \sum\limits_{i=1}^r (\circs\limits_{q=1}^p\boldsymbol{e}_i^r).
\end{equation}
Tensor structures are conveniently represented by tensor network
diagrams that are a recent adaptation~\citep{cichockietal17} from
quantum physics where they were originally introduced to visually
describe many-body problems. Each node in a tensor network diagram
corresponds to a tensor and each edge coming from the node
represents a mode. A node with no edges is a scalar, a node with one
\begin{figure}[!h]
\vspace*{-0.2cm}
\centering
  \mbox{
    \subfloat[$x\in\mathbb
    R$]{\includegraphics[width=.22\columnwidth]{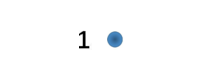}\label{subfig:network-a}}
    \subfloat[$\boldsymbol{x}\in{\mathbb
      R}^2$]{\includegraphics[width=.22\columnwidth]{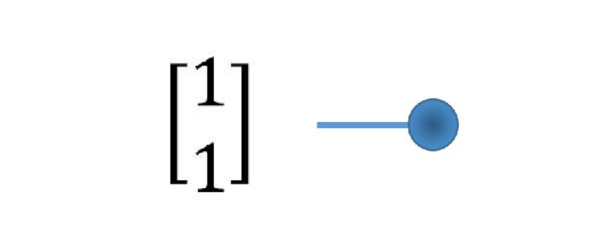}\label{subfig:network-b}}
    \hfill
     \subfloat[$X\in{\mathbb R}^{2\times 2}$]{\includegraphics[width=.28\columnwidth]{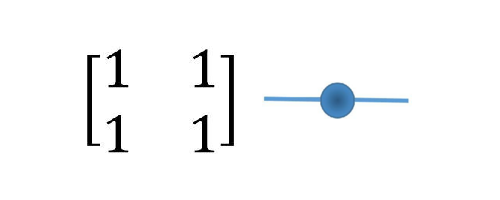}\label{subfig:network-c}}
     \hfill
     \subfloat[$\mX\in{\mathbb R}^{2\times 2\times 2}$]{\includegraphics[width=.28\columnwidth]{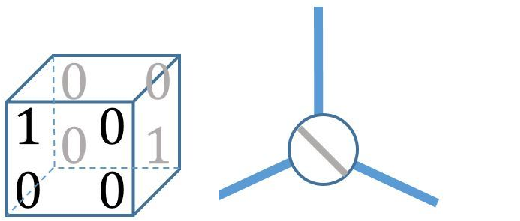}\label{subfig:network-d}}}
     \caption{Tensor network diagrams of a (a) scalar, (b) vector, (c)
       matrix and (d) third-order diagonal tensor.}
     \label{fig:networks}
\end{figure}
edge is a vector and a node with two edges is a matrix. More generally, a node with
$p$ edges is a $p$th-order tensor (Figs.~\ref{subfig:network-a}--\ref{subfig:network-d}). (The angle
between edges has no meaning beyond aesthetics.) Diagonal
tensors as in \eqref{eq:diag} are represented by putting a diagonal across
the node, as in Fig.~\ref{subfig:network-d}.
\subsubsection{Tensor reshapings, contractions, low-rank formats}
The matricization of a tensor is a matrix with its elements arranged differently. The following definition is from \citet{kolda06}
\begin{defn}\label{def:matdef}
Let $\mathscr{S}  = \{ r_1,\sdots,r_L\}$ and $\mathscr{T} =
\{m_1,\sdots,m_M\}$ be ordered sets that partition the set of
modes $\mathscr{M} = \{1,\sdots,p\}$ of $\mX \in
\mathbb{R}^{\times_{j=1}^p m_j}$. Here $L+M=p$. Then if $\bi = [i_1,\sdots,i_p]'$, the matricization $\mX_{(\mathscr{S}\times \mathscr{T})}$ is a matrix of size $(\prod_{q \in \mathscr{S}} m_q) \times (\prod_{q \in \mathscr{T}} m_q)$ defined as
\begin{equation}\label{eq:mat}
\mX_{(\mathscr{S}\times \mathscr{T})}
= \sum\limits_{i_1=1}^{m_1} \sdots \sum\limits_{i_p=1}^{m_p} \mX(\bi)
\big( \bigotimes\limits_{q \in \mathscr{S}} \boldsymbol{e}_{i_q}^{m_q} \big)
\big( \bigotimes\limits_{q \in \mathscr{T}} \boldsymbol{e}_{i_{q}}^{m_{q}} \big)'.
\end{equation}
\end{defn}
We define matricizations in \textit{reverse lexicographic} order to be
consistent with the traditional matrix vectorization. This means that
the $q$ modes in the multiple Kronecker product \eqref{eq:mat} are
selected in reverse order. \begin{table}[h]
 	\caption{\label{table:1}Tensor reshapings defined by
          specifying partitions $(\mathscr{S},\mathscr{T})$ of
          $\mathscr{M}$ in  \eqref{eq:mat}.}
        \centering
\begin{tabular}{ c | c | c | c }
    \hline
   Reshaping & Notation & $\mathscr{S}$ & $\mathscr{T}$ \\[1ex]  \hline
   \begin{tabular}{@{}c@{}}$k$th mode \\matricization \end{tabular}
       &
        $\mX_{(k)}$ & $\{k\}$ & 
        \begin{tabular}{@{}c@{}}
        $\{1,\sdots,k-1,k+1,$\\$\sdots,p\}$
        \end{tabular} \\ [1ex] 
   \begin{tabular}{@{}c@{}}$k$th canonical \\matricization \end{tabular} & $\mX_{<k>}$ & $\{1, \sdots,k\}$ & $\{k+1,\sdots,p\}$ \\ [1ex] 
    vectorization & $\vecc(\mX)$ & $\{1, \sdots,p\}$ & $\varnothing$ \\ \hline
  \end{tabular}
\end{table}
Table \ref{table:1} defines several
reshapings by selecting different partitions
$(\mathscr{S},\mathscr{T})$ of $\mathscr{M}$. These definitions are
clarified in \eqref{defvec}, \eqref{kmode} and
\eqref{canonical}.

Tensor contractions~\citep{cichockietal17} generalize the matrix product to
higher-ordered tensors. We use
$\mX\times_{k_1,\sdots,k_a}^{l_1,\sdots,l_a} \mY$ to denote
the mode-${{l_1,\sdots,l_a}\choose {k_1,\sdots,k_a}}$ product
or contraction between the $(k_1,\sdots,k_a)$ modes of $\mX \in
\mathbb{R}^{\times_{j=1}^p m_j}$ and the
$(l_1,\sdots,l_a)$ modes of $\mY \in \mathbb{R}^{\times_{j=1}^q n_j}$, where $m_{k_1} = n_{l_1}, \ldots, m_{k_a} = n_{l_a}$. This contraction results in a
tensor of order $p+q-2  a$ where the $a$ pairs of modes $(l_j,k_j)$
get \textit{contracted}. A simple contraction between the $k$th mode
of $\mX$ and the $l$th mode of $ \mY$ has 
$(i_1,\sdots,i_{k-1},i_{k+1},\sdots,i_p,j_1,\sdots,j_{l-1},j_{l+1},\sdots,j_q)$th element 
\begin{align}\label{def:modecontraction}
\sum\limits_{t=1}^{m_l}
\mX(i_1,\sdots,i_{k-1},t,i_{k+1},\sdots,i_p)
\mY(j_1,\sdots,j_{l-1},t,j_{l+1},\sdots,j_q).
\end{align}
Similarly, multiple contractions sum over multiple products of the
elements of $\mX$ and $\mY$. Table~\ref{table:2} defines some
contractions using this notation. An important special case is 
{\em partial contraction} that contracts all the  $p<q$ modes of $\mX \in
\mathbb{R}^{\times_{j=1}^p m_j}$ with the first $p$ modes of
$\mY\in \mathbb{R}^{\times_{j=1}^q m_j}$ 
resulting in a tensor  $\langle \mX |\mY \rangle =\mX
\times_{1,\sdots,p}^{1,\sdots,p} \mY$  of size 
 $ \mathbb{R}^{\times_{j=p+1}^q m_j} .$
The partial contraction helps  define ToTR,
and can also be written as a matrix-vector multiplication using
Lemma~\ref{lemma1}~\ref{lemma1:v} (below).
\begin{table}
  \caption{\label{table:2}Tensor contractions, where the contraction
    along one mode is defined as per \eqref{def:modecontraction}. Here $\mX \in \mathbb{R}^{\times_{j=1}^p m_j}$, $\mY \in \mathbb{R}^{\times_{j=1}^q n_j}$, and $X$ and $Y$ are the cases where $p=2$ and $q=2$ respectively.}
  \centering
\begin{tabular}{ c | c | c | c }
    \hline
   Contraction & Notation & Definition & Conditions \\ [1ex] \hline
    matrix product & $XY$ & $X\times_2^1 Y$ & $p=q=2$ \\ [1ex] 
\begin{tabular}{@{}c@{}}$k$th mode \\matrix product \end{tabular}      &$\mX \times_k Y$&$\mX \times_k^2 Y$& $q =2$\\[1ex] 
\begin{tabular}{@{}c@{}}$k$th mode \\vector product \end{tabular}     &$\mX \bar{\times}_k \boldsymbol{y}$&$\mX \times_k^1  \boldsymbol{y}$&$q=1$\\[1ex] 
    inner product & $\langle \mX, \mY \rangle $ & $\mX\times_{1,\sdots,p}^{1,\sdots,q} \mY$ & $p=q$ \\ [1ex] 
    partial contraction & $\langle \mX |\mY \rangle
$&$\mX \times_{1,\sdots,p}^{1,\sdots,p} \mY$&$p<q$\\[1ex] 
\begin{tabular}{@{}c@{}}last mode with \\first mode \end{tabular}    
      & $\mX \times^1 \mY$&$\mX \times_p^1\mY$& --- \\
    \hline
  \end{tabular}
  \vspace*{-0.5cm}
\end{table}
The tensor trace is a self-contraction between the two outer-most modes of a tensor. If $m_1 = m_p$, then
\begin{equation}
\tr (\mX)
= \sum_{i=1}^{m_1} \mX(i,:,:\sdots,:,i),
\end{equation}
whence $\tr (\mX)\in \mathbb{R}^{\times_{j=2}^{p+1} m_j}$. 
The contraction between two distinct modes (from possibly the same
tensor) is represented in tensor network diagrams by joining the
corresponding edges (see Fig.~\ref{fig:contractions} for examples).
\begin{figure}[h]
\vspace*{-0.5cm}
     \centering
     \subfloat[$AB \!\in\! \mathbb{R}^{p\times r}$]{\includegraphics[width=0.24\columnwidth]{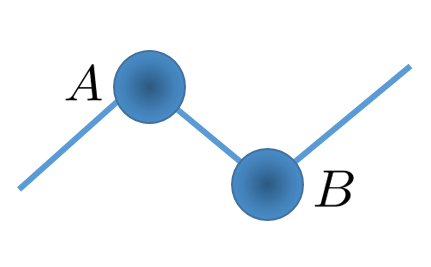}\label{subfig:contractions-a}}
     \hspace*{0.01cm}
     \subfloat[$\langle A | \mC \rangle \!\in\! \mathbb{R}^{r\times s}$]{\includegraphics[width=0.27\columnwidth]{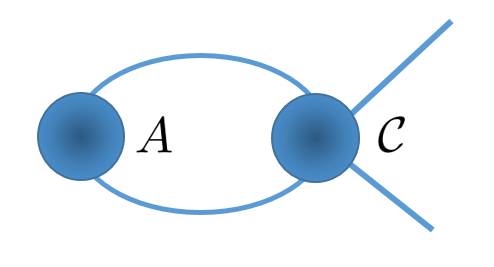}\label{subfig:contractions-b}}
     \hspace*{0.01cm}
     \subfloat[$\langle \mD , \mE \rangle \!\in\! \mathbb{R}$]{\includegraphics[width=0.22\columnwidth]{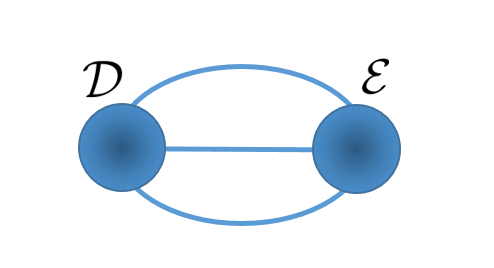}\label{subfig:contractions-c}}
     \hspace*{0.01cm}
     \subfloat[$\tr(\mC) \!\in\! \mathbb{R}^{q\times r}$]{\includegraphics[trim=0 2.5 0 2.5, clip,width=.26\columnwidth]{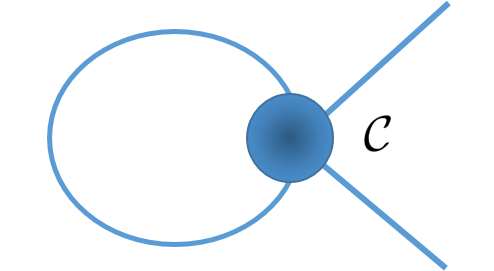}\label{subfig:contractions-d}}
     \caption{Tensor network diagrams of (a) matrix product,
       (b) partial contraction, (c) inner product
       and (d) trace. Here $A\in\mathbb{R}^{p\times q}$, $B\in\mathbb{R}^{q\times r}$,  $\mD,\mE \in \mathbb{R}^{p\times q\times r}$ and $\mC \in \mathbb{R}^{p\times q\times r \times p}$.}
     \label{fig:contractions}
\end{figure}
Also, applying the $k$th mode matricization to every mode of $\mV \in
\mathbb{R}^{\times_{j=1}^p g_j}$ with respect to (WRT) $A_i \in \mathbb{R}^{m_i\times g_i}, i=1,2,\sdots,p$ results in the TK product $\mB =  [\![ \mV; A_1, \sdots,A_p  ]\!]$ defined as
  \begin{align}\label{Tuckerform}
 \mB = \sum\limits_{i_1=1}^{m_1} \sdots \sum\limits_{i_p=1}^{m_p} \mV{(i_1 \sdots i_p)} \big( \circs\limits_{q=1}^p A_q(:,i_q) \big).
  \end{align}
A tensor $\mB$ that can be written as the product in
\eqref{Tuckerform} is said to have a {\em TK format} with {\em TK rank}
$(g_1,\sdots,g_p)$. In this case, $\mV$ is called the \textit{core
  tensor} and $A_1,\sdots,A_p$ are  the \textit{factor
  matrices}. When $g_i<<m_i$ for $i=1,\ldots,p$, the TK format
substantially reduces the number of unconstrained elements in a
tensor and 
its complexity. TK formats are  often fit by higher-order singular
value  (HOSVD)~\citep{DeLathauweretal00} or  LQ (HOLQ) decomposition~\citep{gerardandhoff16}. For a diagonal core tensor $\mV$,
as in \eqref{eq:diag}, the TK format reduces to the CP format of rank $r$. This reduction is equivalent to setting the
tensor $\mB$ as the sum of $r$ vector outer products, where the
vectors correspond to the columns of matrix factors $A_i \in
\mathbb{R}^{m_i\times r}, i= 1,\ldots,p$, 
\begin{equation}\label{CPform}
 \mB =[\![  \boldsymbol{\lambda} ; A_1, \sdots,A_p  ]\!]
=\sum_{i=1}^{r} \boldsymbol{\lambda}(i) 
\big( \circs\limits_{q=1}^p A_q(:,i) \big).
\end{equation}
The vector $\boldsymbol{\lambda}\in \mathbb{R}^r$ contains the
diagonal entries of the core tensor, and is often set to the
proportionality constants that make the matrix factors have unit
column norms. When $\boldsymbol{\lambda}$ is ignored in the
specification of \eqref{CPform}, we assume that $\boldsymbol{\lambda} = [1,1,\sdots,1]'$. 
A tensor $\mB$ is said to have an OP format if it can be written as $\mB = \circ [\![ A_1,A_2 \sdots,A_p  ]\!] $, or
\begin{align}\label{OPform}
\mB = 
\sum\limits_{\substack{i_1,\sdots i_p\\j_1\sdots,j_p}}
(\prod\limits_{q=1}^p A_q[i_q,j_q])
\Big\{
\big( \circs\limits_{q=1}^p \boldsymbol{e}_{j_q}^{h_q} \big) \circ
\big( \circs\limits_{q=1}^p \boldsymbol{e}_{i_q}^{m_q} \big)
\Big\},
\end{align}
where for all $q=1,\sdots,p$, we have  $A_q \in \mathbb{R}^{m_q\times
  h_q}$ and the summation over $i_q$  is
from 1 through $m_q$, and that over $j_q$ is from 1 through $h_q$.
 Our novel OP format is essentially the outer product of multiple matrices, and is useful for expressing the TK product of \eqref{Tuckerform} as a partial contraction between $\mV$ and $\circ [\![ A_1,A_2 \sdots,A_p  ]\!]$, as we shortly state and prove in Theorem~\ref{thm:mat_outer_product}\ref{thm:mat_outer_product:y}.
The derivation needs some properties of tensor products and
reshapings that we prove first in Lemma~\ref{lemma1}, along with
several other properties that are useful for tensor 
manipulations. (Lemma~\ref{lemma1}\ref{lemma1:z},\ref{lemma1:x} and
\ref{lemma1:w} have been stated without proof in \citet{kolda06},
\citet{koldaandbader09}, and \citet{cichockietal17} but we provide
proofs here for completeness.)
\begin{lemma}\label{lemma1} 
Let $\mathcal{X} \in \mathbb{R}^{\times_{j=1}^p m_j}$. Then using the notation of Tables \ref{table:1} and \ref{table:2}, where $k = 1,\ldots,p$,
\begin{enumerate}[label=(\alph*)]
\item \label{lemma1:z} $\mathcal{X}_{<p-1>} = \mathcal{X}_{(p)}'$.
\item \label{lemma1:y} $\vecc(\mathcal{X})
= \vecc(\mathcal{X}_{(1)}) 
= \vecc(\mathcal{X}_{<1>}) 
= \ldots 
=\vecc{(\mathcal{X}_{<p>}}).
$
\item  \label{lemma1:x} $\langle \mathcal{X} , \mathcal{Y} \rangle 
\!=\!  (\vecc{\mathcal{X})}' (\vecc{\mathcal{Y})}
\!=\! \tr{(\mathcal{X}_{(k)}\mathcal{Y}_{(k)}')}$,
 $\mathcal{Y} \in \mathbb{R}^{\times_{j=1}^p m_j}$.
 \item \label{lemma1:w}
$
\vecc[\![ \mathcal{X}; A_1, \sdots,A_p  ]\!] 
=  \big(\bigotimes_{i=p}^1 A_i\big) \vecc( \mathcal{X})
$,\quad where $A_i \in R^{n_i \times m_i}$ for any $n_i \in \mathbb{N} $.

\item \label{lemma1:v}
$
\vecc\langle \mathcal{X} | \mathcal{B} \rangle
= \mathcal{B}_{<p>}' \vecc{ \mathcal{X}}
,
\mathcal{B}\in \mathbb{R}^{(\times_{j=1}^p m_j)\times(\times_{j=1}^q h_j)}.
 $
\item \label{lemma1:u} $
\vecc\mathcal{X}_{(k)} = K_{(k)}
\vecc{(\mathcal{X})}
,\quad where
K_{(k)} \!=\! \big( I_{\prod_{i=k+1}^{p}m_{i}} \otimes K_{\prod_{i=1}^{k-1}m_i,m_k}\big).
$

\end{enumerate}
\end{lemma}
\begin{proof}
See Section \ref{proof:lemma1}. 
\end{proof}  
We now use Lemma \ref{lemma1} to state and prove the following
\begin{thm}\label{thm:mat_outer_product}
 Consider a $p$th-order tensor $\mX$ and matrices $M_1,\ldots,M_p$ such that the TK product with $\mX$ can be formed. Then 
\begin{enumerate}[label=(\alph*)]
\item\label{thm:mat_outer_product:z}$\circ [\![ M_1, \ldots,M_p  ]\!]_{<p>} = \bigotimes_{q=p}^1 M_q'$.
\item\label{thm:mat_outer_product:y}
$
 \langle \mX | \circ [\![ M_1, \ldots,M_p  ]\!] \rangle 
 = 
 [\![ \mX; M_1, \sdots,M_p  ]\!] .
$

\item\label{thm:mat_outer_product:x} For any $k=1,\sdots p$, let  $\mathscr{S} = \{ k,p+k\}$. Then
\vspace*{-0.3cm}
\begin{equation}\label{eq:mat_outer_product}
\begin{split}
\circ [\![ &M_1, \ldots,M_p  ]\!]_{(\mathscr{S}\times \mathscr{S}^\mathsf{c})} 
\\&=(\vecc M_k')(\vecc \circ [\![ M_1, \ldots,M_{k-1},M_{k+1},\ldots,M_p  ]\!])'.
\end{split}
\end{equation}

\end{enumerate}
\end{thm}
\begin{proof}
See Section \ref{proof:mat_outer_product}.
\end{proof}

\underline{Remark}: If $\mB = \circ [\![X',X'  ]\!]$ for some matrix
$X$, then
Theorem \ref{thm:mat_outer_product}\ref{thm:mat_outer_product:z}
implies that $\mB_{<2>} = X \otimes X $ 
while $\mB_{(1,3)\times (2,4)} = (\vecc X)(\vecc X)'$ by
Theorem~\ref{thm:mat_outer_product}\ref{thm:mat_outer_product:x}. In other words, $X \otimes X$ and $(\vecc X)(\vecc X)'$ are different
matricizations of the same OP tensor $\mB$, which
formalizes our intuition
because both $X \otimes X$ and $(\vecc X)(\vecc X)'$ have the same number of 
elements, and it motivates naming the format \textit{outer product}. 

Finally, a tensor $\mB$ is said to have a Tensor Ring (TR) format with TR rank $(g_1,\sdots,g_p)$ if it can be expressed as
\begin{equation}\label{TRform}
\mB
= \tr (
\mG_1\times^1\sdots \times^1\mG_p
),
\end{equation}
where $\mG_{i}\in \mathbb{R}^{g_{i-1} \times m_i \times g_i}$ for
$i=1,\sdots,p$ and $g_0 = g_{p}$. The TR format is called the
Matrix Product State (MPS) with closed boundary conditions in many-body
physics \citep{orus14}. When exactly one of the TR ranks is 1, the TR
format is the same as the Tensor Train (TT) format \citep{oseledets11}
and is the MPS with open boundary conditions.

We conclude our discussion of low-rank tensor formats by using tensor
network diagrams to illustrate in Fig.~\ref{fig:formats}, a fourth-order
tensor in the TK, CP and TR formats.
\begin{figure}[h]
     \centering
    \subfloat[{\scriptsize $\tuckL \mV;A_1,A_2,A_3,A_4 \tuckR$}]{\includegraphics[trim=5 0 5 0, clip,width=.32\columnwidth]{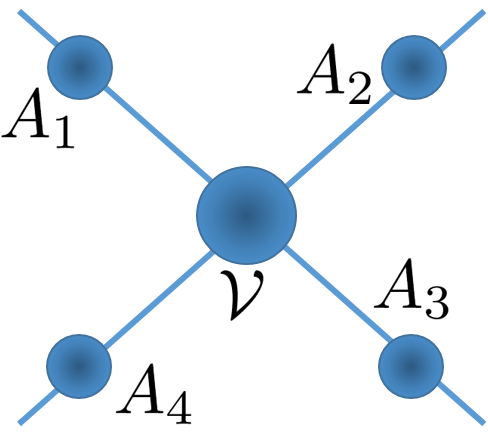}\label{subfig:format-a}}
    \hspace*{0.1cm}
    \subfloat[{\scriptsize $\tuckL B_1,B_2,B_3,B_4\tuckR$}]{\includegraphics[width=.30\columnwidth]{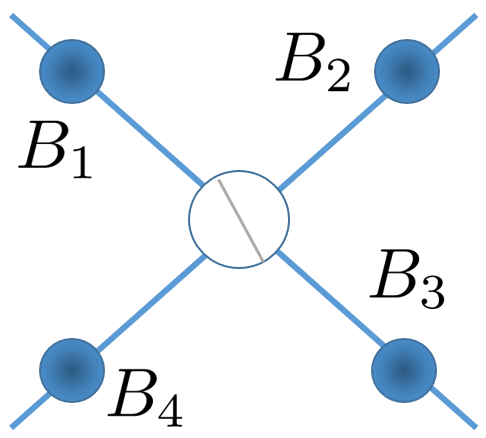}\label{subfig:format-b}}
    \hspace*{0.1cm}
     \subfloat[{\scriptsize $\tr(\mC_1\!\times^1\!\mC_2\!\times^1\!\mC_3\!\times^1\!\mC_4)$}]
{\includegraphics[width=.32\columnwidth]{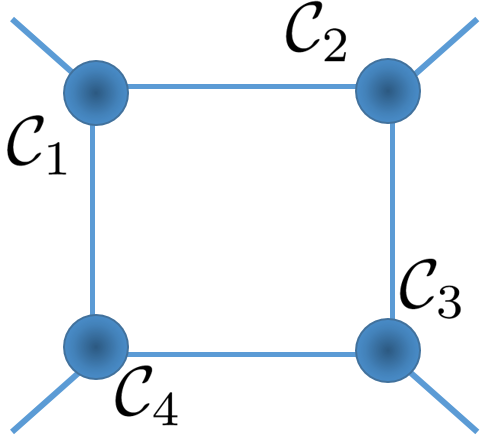}\label{subfig:format-d}}
     \caption{Tensor network diagrams of example fourth-order tensor of
       (a)  TK format \eqref{Tuckerform}, (b) CP format
       \eqref{CPform}, and (c) TR format \eqref{TRform}.} 
     \label{fig:formats}
\end{figure}
\subsubsection{The TVN Distribution}
The matrix-variate normal distribution, abbreviated in
\citet{thompsonetal20} as MxVN to distinguish it from the
vector-variate multinormal distribution (MVN), was studied extensively
in  \citet{guptaandnagar99}. A random matrix $\bX\in
\mathbb{R}^{m_1\times m_2}$  follows a MxVN distribution if
$\vecc(\bX)$ is MVN with covariance matrix $\Sigma_2\otimes \Sigma_1$,
where $\Sigma_k \in \mathbb{R}^{m_k\times m_k}$ for $ k=1,2$. Our TVN
distribution, formulated in Definition \ref{def:tvn}, extends this
idea to the case of higher-order random tensors. For simplicity, we
define the following notation for use in the rest of the paper:
\begin{equation*}
m=\prod_{i=1}^p m_i, \hspace{.1cm}
m_{-k}=m/m_k,\hspace{.1cm}
\Sigma = \bigotimes_{i=p}^1 \Sigma_i,\hspace{.1cm}
\Sigma_{-k}=\bigotimes_{\substack{i=p\\i\neq k}}^1 \Sigma_i.
\end{equation*}
 \begin{defn}\label{def:tvn}
A random tensor $\bmX \in \mathbb{R}^{\times_{j=1}^p m_j}$ follows a
$p$th-order TVN distribution with mean $\mM \in
\mathbb{R}^{\times_{j=1}^p m_j}$ and non-negative definite scale
matrices $\Sigma_i \in \mathbb{R}^{m_i \times m_i}$ for $i =
1,2,\ldots,p$ ({\em i.e.}, $\bmX \sim \N_{\bbm}\big(\mM,\Sigma_1,\Sigma_2\ldots,\Sigma_p\big)$ where $\bbm = [m_1,m_2,\ldots,m_p]'$ ) if $\vecc(\bmX) \sim \N_{m} \big(\vecc(\mM), \Sigma \big). $
\end{defn}

The Kronecker product in Definition \ref{def:tvn} is in reverse order
because we have defined vectorization in reverse lexicographic order. 
Definition \ref{def:tvn} defines the TVN distribution in terms of a
vectorization. We state and prove the distribution of other tensor
reshapings in Theorem \ref{thm:tensnorm_reshap}. These results are
essential in the development of ToTR models
with TVN errors, as they allow us to model the vectorized tensor
errors in terms of the MVN distribution. 
\begin{thm}\label{thm:tensnorm_reshap}
The following statements are equivalent:
\begin{enumerate}[label=(\alph*)]
\item \label{tensnorm_reshap:z}$\bmY \sim \N_{\bbm}(\mM,\Sigma_1,\Sigma_2,\sdots,\Sigma_p)$
\item \label{tensnorm_reshap:y}$\vecc(\bmY) \sim \N_m ( \vecc(\mM), \Sigma)$
\item \label{tensnorm_reshap:x}
$
\bmY_{(k)} \sim \N_{[m_k,m_{-k}]'} ( \mM_{(k)}, \Sigma_k,  \Sigma_{-k}),\quad k=1,2,\sdots,p$
\item For $k=1,2,\ldots,p$ and $\bbm_k = [\prod_{i=1}^k m_i,\prod_{i=k+1}^p m_i ]'$,\label{tensnorm_reshap:w}

$
\bmY_{<k>} \sim \N_{\bbm_k} ( \mM_{<k>},  \bigotimes_{i=k}^1 \Sigma_i, \bigotimes_{i=p}^{k+1} \Sigma_i)
$
\end{enumerate}

\end{thm}
\begin{proof} 
  \ref{tensnorm_reshap:z} and \ref{tensnorm_reshap:y} are equivalent,
  following Definition~\ref{def:tvn} while \ref{tensnorm_reshap:y} and
  \ref{tensnorm_reshap:x} are so from
  Lemma~\ref{lemma1}\ref{lemma1:u} with $K_{(k)}\Sigma K_{(k)}' = 
  \Sigma_{-k}\otimes \Sigma_k$. Further, \ref{tensnorm_reshap:y} and
  \ref{tensnorm_reshap:w} are  equivalent because of
  Lemma~\ref{lemma1}\ref{lemma1:y}.  
\end{proof}
The density of $\bmY \sim \N_{\bbm}(\mM,\Sigma_1,\Sigma_2,\sdots,\Sigma_p)$ is
$
f(\mY;\mM,\Sigma)  
= |2\pi\Sigma|^{-1/2}\exp\Big\{-\frac{1}{2} D_\Sigma^2(\mY,\mM)\Big\},
$
where $D_\Sigma^2(\mY,\mM)$ is the squared Mahalanobis distance
between $\mY$ and $\mM$, and
has the equivalent representation
\begin{align}\label{eq:mahalanobis1} 
\begin{split}
D_\Sigma^2(\mY,\mM) 
&=  \vecc(\mY-\mM)'\Sigma^{-1}\vecc(\mY-\mM)\\
&\equiv  \langle(\mY-\mM),[\![(\mY-\mM); \Sigma_1^{-1},\Sigma_2^{-1}, \sdots,\Sigma_p^{-1}  ]\!]\rangle,
\end{split}
\end{align}
by Lemmas
\ref{lemma1}\ref{lemma1:x} and \ref{lemma1}\ref{lemma1:w}.
Property~\ref{prop:Kron} provides similar alternative expressions for the
determinant $\det{(\Sigma)}$ of $\Sigma$.
\subsection{Tensor-variate Linear Models with TVN Errors}
\subsubsection{Tensor-on-Tensor Regression}
We formulate the ToTR model as
\begin{equation}\label{eq:model}
\bmY_i =\Upsilon + \langle  \mX_i|  \mB \rangle +\bmE_i
,\quad
i = 1,2\sdots,n,
\end{equation}
where the response $\bmY_i\!\in\!\mathbb{R}^{\times_{j=1}^p m_j}$
and the covariate $\mX_i\in \mathbb{R}^{\times_{j=1}^l h_j}$ are both
tensor-valued, $\bmE_i \overset{iid}{\sim} \N_{\bbm} (0, \sigma^2
\Sigma_1,\Sigma_2,\ldots,\Sigma_p)$ is the TVN-distributed error, $\mB\in
\mathbb{R}^{ (\times_{j=1}^l h_j)\times (\times_{j=1}^p m_j)}$ is the (tensor-valued) regression
parameter and $\Upsilon$ 
is the (tensor-valued) intercept. This model is essentially the classical
MVMLR model but exploits the
tensor-variate structure of the  covariates and responses to reduce the
total number of parameters. To see this, we apply
Lemma~\ref{lemma1}\ref{lemma1:v} and
Theorem~\ref{thm:tensnorm_reshap}\ref{tensnorm_reshap:y} to
vectorize \eqref{eq:model} as  
$
\vecc(\bmY_i)  = \vecc(\Upsilon) + \mB_{<l>}'\vecc(\mX_i)+ \boldsymbol{e}_i,
$
where $\boldsymbol{e}_i\overset{iid}{\sim}\N_m(0,\sigma^2\Sigma)$ is the error.
 This formulation leads to a MVMLR model with intercept, which can be incorporated into the
covariates as $ [\vecc(\Upsilon)\mB_{<l>}'][1,\vecc(\mX_i)']'$. But the
covariate $[1,\vecc(\mX_i)']'$ is then no longer a vectorized tensor and 
we can not exploit the tensor structure of $\mX_i$. To obviate this
possibility,~\eqref{eq:model} includes a separate intercept term $\Upsilon$.

Imposing a low-rank format on $\mB$ proffers several
advantages. First, it makes the regression model practical to use, as
accurate estimation of an unstructured version of $\mB$ may otherwise
be prohibitive when dimensionality is high relative to sample
size. Second, $\mB$ can be interpreted, in spite of its high
dimensions, as being explainable through a few lower-dimensional
tensor factor matrices (Fig.~\ref{fig:regressors}).  
\begin{figure}[h]
\vspace*{-0.5cm}
     \centering
     \subfloat[][]{\includegraphics[trim=0 0 20 0, clip,width=0.25\linewidth]{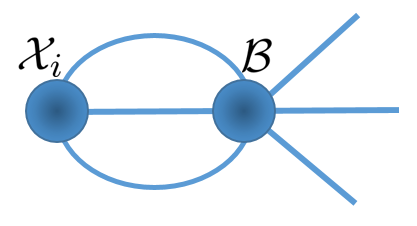}\label{subfig:regress-a}}
     \subfloat[][]{\includegraphics[width=0.25\linewidth]{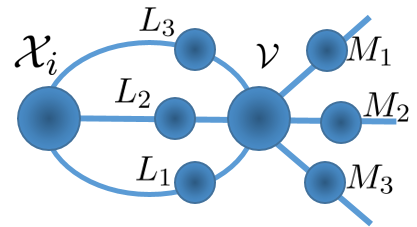}\label{subfig:regress-b}}
     \subfloat[][]{\includegraphics[width=0.25\linewidth]{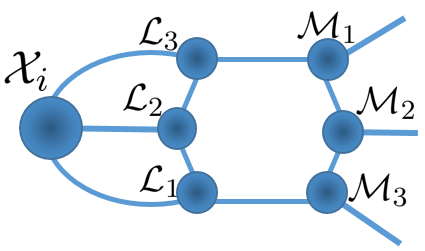}\label{subfig:regress-c}}
     \subfloat[][]{\includegraphics[width=0.25\linewidth]{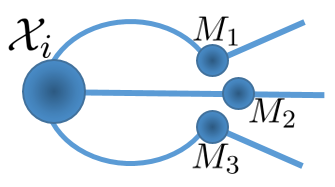}\label{subfig:regress-d}}
     \caption{Tensor network diagrams of tensor-on-tensor regression
       when both the response and the covariates are third-order
       tensors and for (a) $\mY_i = \langle \mX_i | \mB\rangle$ and (b-d)
       the special cases when the $\mB$ in $\langle \mX_i |
       \mB\rangle$ are of (b) Tucker, (c) TR and (d) OP formats (illustrated in Fig. \ref{fig:formats}). }
     \label{fig:regressors}
\end{figure}
These explanations mirror the many-body problem in physics, where
weakly-coupled degrees of freedom are often embedded in
ultra-high-dimensional Hilbert spaces
\citep{orus14,wolfetal08}. 
We now turn to the problem of learning $\mB$ in the given setup. 

The dispersion matrix $\Sigma$ in the TVN distribution specification
of the $\bmE_i$s in~\eqref{eq:model} 
is Kronecker-separable and leads to 
the number of unconstrained parameters from $(\prod_{i=1}^p m_i)\times
(\prod_{i=1}^p m_i +1)/2 $ to $\sum_{i=1}^pm_i (m_i +1)/2$. 
Further, this Kronecker-separable structure is intuitive because it
assigns a covariance matrix to each tensor-response dimension. This
allows us to separately incorporate different dependence contexts that
exist within a tensor. For example, a tensor may have temporal and
spatial contexts in its modes. Kronecker separability allows us to
assign a covariance matrix to each of these different
contexts. However, Kronecker separability results in unidentifiable
scale matrices $(\Sigma_i)$, as $cA \otimes 
B = A \otimes cB$ for any matrices $A, B$ and $c\neq 0$, so we 
 constrain the scale matrices to each have
$\Sigma_i(1,1) = 1$, and introduce a parameter $\sigma^2$ to
capture an overall proportional scalar variance. This
approach further reduces the number of parameters by $p-1$ and imposes
a curved exponential family distribution on the errors~\citep{srivastavaetal08}.  
\subsubsection{The TANOVA model}\label{subseq:TANOVA}
Suppose that  we observe  independent tensors
$\bmY_{j_1,\sdots,j_l,i}\in \mathbb{R}^{\times_{j=1}^p m_j}$ $(i =1,\ldots, n_{j_1,\ldots,j_l})$,  where $j_k$
 $(k=1,\sdots,l)$, indexes the $k$th categorical (factor) variable
 of $h_k$ levels, that is, $j_k\in\{1,\sdots,h_k\}$. 
 The tensor-valued parameter $\mB$ encoding the dimensions of $\bmY_{i,j_1,\sdots,j_l}$ and all the possible factor classes is of size $h_1\times\sdots\times
h_l\times m_1\times\sdots\times m_p$. To see this, consider
the $l$th ordered single-entry tensor $\mX_{j_1,\sdots,j_l} =
\circs_{q=1}^l \boldsymbol{e}_{i_q}^{h_q}$ that is unity at
$(j_1,\sdots,j_l)$ and zero everywhere else. Then $\langle \mX_{j_1,\sdots,j_l} | \mB \rangle \in \mathbb{R}^{\times_{j=1}^p m_j}$ and
\begin{equation}\label{eq:TANOVA}
\langle \mX_{j_1,\sdots,j_l} | \mB \rangle = \mB[j_1,\sdots,j_l,:,:,\sdots,:].
\end{equation}
Therefore, modeling $\mathbb{E}(\bmY_{i,j_1,\sdots,j_l})$ as $\langle \mX_{j_1,\sdots,j_l} | \mB \rangle$ results in each factor combination $(j_1,\sdots,j_p)$ getting assigned its own mean parameter as a sub-tensor of $\mB$. This high-dimensional parameter $\mB$ is identical to  cell-means MANOVA if we vectorize \eqref{eq:TANOVA} using Lemma \ref{lemma1}\ref{lemma1:v} as
$
\mB'_{<l>}\vecc(\mX_{j_1,\sdots,j_l}) = \vecc(\mB[j_1,\sdots,j_p,:,:,\sdots,:]),
$
in which case $\vecc(\mX_{j_1,\sdots,j_p}) = \otimes_{q=l}^1 \boldsymbol{e}_{i_q}^{h_q}$ corresponds to a row in the MANOVA design matrix. 
Although this formulation is fundamentally the same as \eqref{eq:TANOVA}, the latter helps us visualize and formulate a low-rank format on $\mB$.
Based on the format, we refer to tensor-valued response regression
models with the mean in (\ref{eq:TANOVA}) as TANOVA$(l,p)$, where $l$
is the number of different factors and $p$ is the order of the tensor-valued
response. In this way, because scalar and vector variables are tensors of
order 0 and 1, ANOVA and MANOVA correspond to TANOVA$(1,0)$
and TANOVA$(1,1)$, respectively. In general, TANOVA$(l,p)$ with TVN errors can be expressed as 
\begin{equation}\label{eq:model_noM}
\bmY_i =\langle  \mX_i|  \mB \rangle +\bmE_i
,\quad
\bmE_i \overset{iid}{\sim} \N_{\bbm}
(0, 
\sigma^2
\Sigma_1,\Sigma_2,\sdots,\Sigma_p),
\end{equation}
where $\mX_i$ is the single-entry tensor that contains all the
assigned factors of $\bmY_i$, for $i = 1,\ldots,n$. Model
\eqref{eq:model_noM} is a ToTR model as in \eqref{eq:model} with no
intercept. This is analogous to ANOVA and MANOVA being special cases
of univariate and multivariate multiple linear regressions,
respectively. Further, the log-likelihood function
of~\eqref{eq:model_noM} is  
\begin{equation}\label{eq:likel_noM}
\ell =
-\dfrac{n}{2} \log|2\pi\sigma^2\Sigma|
-\dfrac{1}{2\sigma^2}\sum_{i=1}^n
D_\Sigma^2(\bmY_i,\langle  \mX_i|  \mB
\rangle).
\end{equation}

\subsection{Parameter Estimation}\label{subsec:parest}
We obtain estimators of \eqref{eq:model}  before deriving their properties.
\subsubsection{Profiling the intercept}\label{sec:est_intercept}
We first show that the intercept in \eqref{eq:model} can be profiled
out by centering the covariates and the responses. To see this, we express the loglikelihood in terms of $\Upsilon$ as
\begin{equation}\label{eq:tensnorm1}
\ell=
-\dfrac{1}{2\sigma^2}\sum_{i=1}^n
D_\Sigma^2(\bmY_i,\Upsilon+\langle  \mX_i|  \mB
\rangle).
\end{equation}
We define the tensor differential of an inner product WRT $\Upsilon$ using the matrix differential $\langle\partial \Upsilon,\mS\rangle = \tr(\partial\Upsilon_{(1)}\mS_{(1)}')$.  Applying it to \eqref{eq:tensnorm1} yields
$$
\partial \ell(\Upsilon) = \dfrac1{\sigma^2}\langle \partial \Upsilon, [\![ 
\bmS ;
\Sigma_1^{-1},\sdots,\Sigma_p^{-1} ]\!] \rangle,
$$
where $n^{-1}\bmS\!=\!\bmbY\!-\!\langle\mbX|\mB\rangle\!-\!\Upsilon$ and $\bmbY,\mbX$ are the averaged responses
and covariates. Now, $\partial \ell(\Upsilon)\!=\!0$ if $\bmS\!=\!0$,
and after profiling on $\mB$, the ML estimator (MLE) of $\Upsilon$ is
$
\hat{\Upsilon}(\mB) =\bmbY - \langle \mbX| \mB\rangle.
$
Setting  $\Upsilon$ in~\eqref{eq:tensnorm1} to be
$\hat{\Upsilon}(\mB)$  yields~\eqref{eq:likel_noM} with centered responses and covariates, so we
assume without loss of generality (WLOG) that~\eqref{eq:model_noM} has
no intercept, and estimate the other parameters. 
Our estimation uses block-relaxation~\citep{deleeuw94} to 
optimize~\eqref{eq:likel_noM}: we partition the
parameter space into blocks 
and serially optimize the parameters in each 
block while holding fixed the other parameters.
\subsubsection{Estimation of $\mB,\Sigma_1,\Sigma_2,\ldots,\Sigma_p$  given $\sigma^2$}
Our estimates simplify as per the format of $\mB$ so we
consider each case individually, before providing an overview.
\paragraph{\underline{TK format}}\label{sec:estTuck}
Let $\mB$ have TK format of rank
$(c_1,\sdots,c_l,d_1,\sdots,d_p)$:
\begin{equation}\label{eq:reg_Tuck}
\mB_{TK} =[\![  \mV ; L_1, \sdots,L_l,M_1,\sdots,M_p  ]\!],
\end{equation}
where $M_k'\Sigma_k^{-1}M_k = I_{d_{k}}$ for $k=1,\sdots,p$. Then the
number of parameters to be estimated goes down from the unconstrained 
$\prod_{i=1}^l h_i \prod_{i=1}^p m_i$  to $\prod_{i=1}^l c_i
\prod_{i=1}^p d_i + \sum_{i=1}^l c_ih_i+\sum_{i=1}^p d_im_i$. The
constraint $M_k'\Sigma_k^{-1}M_k = I_{d_{k}}$ greatly simplifies estimation
and inference. Using \eqref{eq:reg_Tuck}, we 
vectorize \eqref{eq:model_noM} for any $k =1,\sdots,l$ as
\begin{equation}\label{eq:Tucker_L}
\vecc (\mY_i) = \mH_{ik<2>}^{TK} \vecc (L_k) + \boldsymbol{e}_i
\end{equation}
where 
$\mH_{ik<2>}^{TK}$ is the 2-canonical matricization of the tensor
\begin{align}\label{eq:H_tucker}
\begin{split}
\mH_{ik}^{TK}=\mX_i\times_{1,\sdots,k-1,k+1,\sdots,l}^{1,\sdots,k-1,k+1,\sdots,l} [\![& \mV; L_1, \sdots ,L_{k-1},
\\&I_{h_k},L_{k+1},\sdots,L_{l},M_1,\sdots, M_{p}  ]\!]
\end{split}
\end{align}
and $ \boldsymbol{e}_i$s are i.i.d $\N_{m}(0,\sigma^2\Sigma)$.
Optimizing~\eqref{eq:likel_noM} WRT $L_k$ forms its own block, which
corresponds to a MVMLR model where, for $S_k^{TK}
= \sum_{i=1}^n  
 (\mH_{ik<2>}^{TK} \Sigma^{-1}  {\mH_{ik<2>}^{TK}}')$, 
 \begin{align}\label{eq:Tuckerest_L}
\vecc(\widehat{L_k}) = 
(S_k^{TK})^{-1}
\big(\sum\limits_{i=1}^n 
\mH_{ik<2>}^{TK}\Sigma^{-1} \vecc{(\mY_i )} \big).
 \end{align}
The computation of $S_k^{TK}$ is greatly simplified based on the constraint that $M_k'\Sigma_k^{-1}M_k=I_{d_k}$ for all $k=1,2,\dots,p$. 
 For fixed $L_1,\sdots,L_p,\Sigma_1,\sdots,\Sigma_p$, we estimate
 $M_1,\sdots,M_p,\mV$. We first show that 
$\mV$ can be profiled from the loglikelihood for fixed $M_1,\sdots,M_p$. To see this, we write an alternative vectorized form of \eqref{eq:model_noM} as
$$
\vecc (\mY_i)= \big( \otimes_{i=p}^1 M_i \big) \mV_{<l>}'\bw_i
   + \boldsymbol{e}_i, 
$$
 where $\bw_i = \vecc[\![  \mX_i ; L_1',L_2', \sdots,L_l']\!]$. Letting $Z\! =\! Y\!-\!M\mV_{<l>}'W$, for $M=\otimes_{k=p}^1 M_k$, $Y =
[\vecc \mY_1 \ldots \vecc \mY_n]$, and $ W\! =\! [\bw_1 \ldots
\bw_n]$ simplifies ~\eqref{eq:likel_noM} to
\begin{equation}\label{eq:Tucker_Likel1}
\ell = 
-\frac{n}{2}\log|\Sigma| 
-\frac{1}{2\sigma^2} \Big\{ \tr(Z'\Sigma^{-1}Z) \Big\}.
\end{equation}
Optimizing \eqref{eq:Tucker_Likel1} for fixed $M_1,\sdots,M_p$ yields the profiled MLE
\begin{equation}\label{eq:Tuckerest_V}
\hat{\mV}_{<l>}(M_1,\Sigma_1,\sdots,M_p,\Sigma_p) =W^{-'}Y'(\bigotimes_{k=p}^1\Sigma_k^{-1}M_k ),
\end{equation}
where $W^-$ is the right inverse of $W$. Therefore, given values of
all $M_k$s, we obtain $\mhV$ by simply inserting them in
\eqref{eq:Tuckerest_V}. To estimate $M_k$ we profile $\mhV$ out of the loglikelihood by replacing 
\eqref{eq:Tuckerest_V} into \eqref{eq:Tucker_Likel1}, and expressing it up to a constant as
\begin{equation} \label{eq:Tucker_VLikel}
\ell(M_k,\Sigma_k) = 
\frac{1}{2\sigma^2}  || M_k'\Sigma_k^{-1}Q_k||^2_2,
\end{equation}
where $Q_k\!=\![\![  \mY_{T} ; M_1'\Sigma_1^{-1},\sdots,M_{k-1}'\Sigma_{k-1}^{-1},I_{m_k},
M_{k+1}'\Sigma_{k+1}^{-1},\sdots,$ $M_p'\Sigma_p^{-1}, W^-W]\!]_{(k)}$
and $\mY_{T}\in \mathbb{R}^{(\times_{j=1}^p m_j)\times
  n}$ is such that $\mY_{T}(:,\sdots,:,i) $ $= \mY_i$ for
$i=1,\sdots,n$. From \eqref{eq:Tucker_VLikel}, the MLE of $M_k$ is obtained via generalized SVD of $Q_k$ \citep{Abdi07}:
\begin{equation}\label{eq:Tuckerest_M}
\hat{M}_k(\Sigma_k) 
= \argmax_{M_k: M_k'\Sigma_k^{-1}M_k = I_{d_k}} || M_k'\Sigma_k^{-1}Q_k||^2_2 
= \Sigma_k^{1/2}U,
\end{equation}
with the leading $d_k$ left singular vectors of $\Sigma_k^{-1/2}Q_k$
as the columns of $U$. To estimate $\Sigma_k$ at fixed $\mB_{TK}$, we
write \eqref{eq:Tucker_Likel1}  as
\begin{equation}\label{eq:Tucker_MLikel}
\ell(\Sigma_k) = 
-\frac{nm_{-k}}{2}\log|\Sigma_k| 
-\frac{1}{2\sigma^2} \tr (\Sigma_k^{-1}S_k),
\end{equation}
where $S_k = \sum_{i=1}^n \mZ_{i(k)}\Sigma_{-k}^{-1} \mZ_{i(k)}'$ and $\mZ_{i(k)} = \mY_i-  \langle\mX_i|  \mB \rangle$. 
The MLE of $\Sigma_k$ under the (TK format) constraint of \eqref{eq:reg_Tuck} is
\begin{equation}\label{eq:Tuckerest_S}
\hat{\Sigma}_k
= \argmax_{\Sigma_k: \Sigma_k[1,1] = 1} \ell(\Sigma_k) = ADJUST(nm_{-k},\sigma^2,S_k),
\end{equation}
where the \textit{ADJUST} procedure is as introduced in
\citet{glanzandcarvalho18} and that was shown to satisfy the Karush-Kuhn-TK (KKT)
conditions. Without this constraint, the MLE is
$S_k/(nm_{-k}\sigma^2)$.  Further reductions can be obtained by
imposing additional parameterized structures on $\Sigma_k$s, as needed.
Our block relaxation algorithm, detailed in Algorithm~\ref{alg:1}, is
initialized by methods discussed in
Section~\ref{subsec:computation_normal}.
\begin{algorithm}\label{alg:1}
    \SetKwInOut{Input}{Initial values} 
    \Input{$k=0$, $\hsigma^{2(0)}, \hL_1^{(0)},\sdots,\hL_l^{(0)},$ $(\hM_1^{(0)},\hSigma_1^{(0)}),\sdots,(\hM_p^{(0)},\hSigma_p^{(0)})$}
    Center the data while saving the means $\mbX$, $\mbY$.\\
    \While{convergence criteria is not met}{
    	\For{$j=1,2,\sdots,p$}{$\hM_j^{(k+1)}$ as per  \eqref{eq:Tuckerest_M}} 
	$\mhV^{(k+1)}$ as per \eqref{eq:Tuckerest_V}\\
	\For{$j=1,2,\sdots,l$}{$\hL_j^{(k+1)}$ as per \eqref{eq:Tuckerest_L}} 
	\For{$j=1,2,\sdots,p$}{
	$\hSigma_j^{(k+1)}$ as per \eqref{eq:Tuckerest_S}\\
	$\hsigma^{2(k+1)}$ as per \eqref{eq:sig2}
	} 
	$k \leftarrow k + 1$
    }
    $\hat{\Upsilon} = \mbY - \langle \mbX | \mhB_{TK} \rangle$
    \caption{Block-relaxation algorithm for ToTR~\eqref{eq:model} with TK formatted $\mB$  }
\end{algorithm}
\paragraph{\underline{CP  format}}\label{sec:estCP}
We now optimize \eqref{eq:likel_noM} when $\mB$ is 
\begin{equation}\label{eq:reg_CP}
\mB_{CP} =[\![  \boldsymbol{\lambda} ; L_1,L_2, \sdots,L_l,M_1,M_2,\sdots,M_p  ]\!],
\end{equation}
that is, of CP format of rank $r$.
Then, with $\Sigma_k=I_{m_k}$ for $k=1,2,\sdots,p$, 
\eqref{eq:model_noM} reduces to the framework of
\citet{lock17}. The CP format reduces the number of parameters
in $\mB$  from $\prod_{i=1}^p m_i \prod_{i=1}^l h_i$ to
$r(\sum_{i=1}^p m_i + \sum_{i=1}^l h_i)$.  Here also, we optimize
\eqref{eq:likel_noM} 
via a  block-relaxation algorithm. The $k$th block corresponds to
$(M_k,\Sigma_k)$ for $k=1,2,\sdots,p$ and can be
estimated in a MVMLR framework by  applying Theorem
\ref{thm:tensnorm_reshap}\ref{tensnorm_reshap:x} on the $k$th mode matricized form of \eqref{eq:model_noM}: 
\begin{equation}\label{eq:reg_CPkth}
{\mY_i}_{(k)} =  M_kG_{ik}^{CP}+E_i
, \quad E_i \overset{iid}{\sim} 
\N_{[m_k, m_{-k}]'}(0,\sigma^2\Sigma_k,\Sigma_{-k}),
\end{equation}
where $G_{ik}^{CP}\equiv\mG_{ik(k)}^{CP}$ is the $k$th mode matricization of 
$
\mG_{ik}^{CP} \!=\!
[\![\langle  [\![\mX_i;L_1',\sdots,L_l']\!]| \mathcal{I}_r^{p+l} \rangle;M_1,\sdots,M_{k-1},I_{m_k},
M_{k+1},\sdots,M_p]\!].
$
Additional simplifications of $G_{ik}^{CP}$ are possible, for example,
using (Section ~\ref{supp:CP})  the Khatri-Rao product
($\odot$)~\citep{koldaandbader09}. When all parameters except $(M_k,\Sigma_k)$ are held fixed, \eqref{eq:reg_CPkth} matches a MVMLR model with
loglikelihood  
\begin{equation}\label{eq:reg_MLECP}
\ell(\Sigma_k,M_k) = \dfrac{nm_{-k}}{2}\log|\Sigma_k^{-1}|-\dfrac{1}{2\sigma^2}\tr(\Sigma_k^{-1} S_k),
\end{equation}
with $S_k  = \sum_{i=1}^n Z_{ik}\Sigma_{-k}^{-1}Z_{ik}'$, where $Z_{ik} =
{\mY_i}_{(k)} -  M_kG_{ik}^{CP}$. Then 
the MLEs are 
\begin{align}\label{eq:CPM}
\begin{split}
&\hat{M}_k = \sum\limits_{i=1}^n {\mY_i}_{(k)} \Sigma_{-k}^{-1}
 {G_{ik}^{CP}}'\Big[\sum\limits_{i=1}^n G_{ik}^{CP}
\Sigma_{-k}^{-1}   {G_{ik}^{CP}}' \Big]^{-1},
\\
&\hat{\Sigma}_k(M_k) = ADJUST(nm_{-k},\sigma^2,S_k).
\end{split}
\end{align}
The matrices $\sum_{i=1}^n {\mY_i}_{(k)} \Sigma_{-k}^{-1}
 {G_{ik}^{CP}}'$, $\sum_{i=1}^n G_{ik}^{CP}
\Sigma_{-k}^{-1}{G_{ik}^{CP}}'$ are substantially simplified in Section~\ref{supp:CP}. We estimate $\Sigma_k$ 
by directly optimizing~\eqref{eq:reg_MLECP}. The other $l$
blocks in the block-relaxation algorithm correspond to
$L_1,\sdots,L_l$ and are each MVMLR models obtained by
vectorizing \eqref{eq:model_noM} as:
 \begin{equation}\label{eq:reg_CPvec}
 \vecc (\mY_i) =   H_{ik}^{CP} \vecc (L_k) + \boldsymbol{e}_i
, \quad \boldsymbol{e}_i\overset{iid}{\sim} 
\N_m(0,\sigma^2\Sigma),
 \end{equation}
 where $H_{ik}^{CP}=\mH_{ik<2>}^{CP}$ and $\mH_{ik}^{CP}$ is identical
 to the $\mH_{ik}^{TK}$ of \eqref{eq:H_tucker}, but for the fact
 that $\mV$ is the diagonal tensor
 $\mathcal{I}_r^{p+l}$. (Fig.~\ref{fig:tot_overview} displays
 $\mH_{ik}^{CP}$ for when $p\!=\!l\!=\!3$.)
 For $k=1,\sdots,l$, holding all parameters except $L_k$ fixed makes
 \eqref{eq:reg_CPvec} a MVMLR model with the MLE of $L_k$ obtained as 
 \begin{equation}\label{eq:CPL}
\vecc(\hat{L}_k) \!=\! 
\Big(
\sum\limits_{i=1}^n 
 H_{ik}^{CP} \Sigma^{-1}  {H_{ik}^{CP}}'
\Big)^{-1}
\Big(\sum\limits_{i=1}^n 
H_{ik}^{CP}\Sigma^{-1} \vecc{(\mY_i )} \Big).
 \end{equation}
 The matrices $\sum_{i=1}^n 
H_{ik}^{CP}\Sigma^{-1} \vecc{(\mY_i )}$,  $\sum_{i=1}^n 
 H_{ik}^{CP} \Sigma^{-1}  {H_{ik}^{CP}}'$ are substantially
 simplified in Section \ref{supp:CP}.
As summarized in Fig.~\ref{fig:tot_overview}, the tensors
$\mH_{ik}^{CP}$ and $\mG_{ik}^{CP}$ play a critical role  in the
estimation of $M_k$ and $L_k$  through \eqref{eq:reg_CPkth} and
\eqref{eq:reg_CPvec}, permitting the use of standard MVMLR and MLR
estimation methods. From~\eqref{CPform}, we deduce that the 
$j$th columns of all the factor matrices in the CP
decomposition~\eqref{eq:reg_CP} are identifiable only up to a
constant. So 
we constrain the columns to have unit norm. 
The MLEs of our parameters are obtained using  a block-relaxation
algorithm, as 
outlined in Algorithm \ref{alg:CP}.
\begin{algorithm}\label{alg:CP}  
    \SetKwInOut{Input}{Initial Values} 
    \Input{$k=0, \hsigma^{2(0)}, \hL_2^{(0)},\hL_3^{(0)},\sdots,\hL_l^{(0)},$ $\hM_1^{(0)},\hM_2^{(0)},\sdots,\hM_p^{(0)}$}
    Center the data while saving the means $\mbX$, $\mbY$.\\
    \While{convergence criteria is not met}{
    \For{$j=1,2,\sdots,l$}{$\hL_j^{(k+1)}$ as per \eqref{eq:CPL} and normalize its columns} 
    \For{$j=1,2,\sdots,p-1$}{$(\hM_k^{(k+1)},\hSigma_k^{(k+1)})$ as per \eqref{eq:CPM} and normalize the columns of $\hM_k^{(k+1)}$} 
	    $(\hM_p^{(k+1)},\hSigma_p^{(k+1)})$ as per \eqref{eq:CPM}\\
	     $\hsigma^{2(k+1)}$ as per \eqref{eq:sig2}\\
	    Normalize the columns of $\hM_p^{(k+1)}$ while setting $\boldsymbol{\hlambda}^{(k+1)}$ to those norms \\	     
	        $k=k+1$
    }
    $\hat{\Upsilon} = \mbY - \langle \mbX | \mhB_{CP} \rangle$
    \caption{Block-relaxation algorithm for ToTR~\eqref{eq:model} with CP formatted $\mB$}
\end{algorithm}

\paragraph{\underline{OP format}}\label{sec:estOP} For an OP-formatted $\mB$, {\em i.e.},
\begin{equation}\label{eq:reg_OP}
\mB_{OP} = \circ [\![ M_1,\sdots,M_p  ]\!], 
\end{equation}
 we use Theorem \ref{thm:mat_outer_product}\ref{thm:mat_outer_product:y} to express \eqref{eq:model_noM} as
\begin{equation}\label{Tucker}
\mY_i = [\![ \mX_i; M_1, \sdots,M_p  ]\!]	+\mE_i.
\end{equation}
We estimate the parameters in~\eqref{Tucker} by applying the $k$th mode matricization for each $k =1,\sdots,p$ on both sides as
$
{\mY_i}_{(k)} =  M_kG_{ik}^{OP}+E_i,
$
where $G_{ik}^{OP} = \mX_{i(k)}\Big(\bigotimes_{j=p,j\neq k}^1
M_j'\Big)$ and $E_i \overset{iid}{\sim} 
\N_{[m_k, m_{-k}]'}(0,\sigma^2\Sigma_k,\Sigma_{-k})$. Given the similarities between this formulation and \eqref{eq:reg_CPkth}, the MLEs of the factor matrices are as in
\eqref{eq:CPL} but with $G_{ik}^{OP}$ instead of $G_{ik}^{CP}$. The
optimization procedure is similar to Algorithm 2, with the difference
again that $M_1,\sdots,M_{p-1}$ are normalized to have unit Frobenius
norm. We conclude by noting that~\eqref{Tucker} is the 
multilinear tensor regression setup of~\citet{hoff14}, and for $p=2$ is the
matrix-variate regression framework of \citet{viroli12} and
\citet{dingandcook16}.  So the OP format frames existing
methodology within the ToTR framework.

\paragraph{\underline{TR format}}\label{sec:estTR}
\label{eq:reg_TC}
Let $\mB$ have TR format~\eqref{TRform}, {\em i.e.}
\begin{equation}\label{eq:reg_TR}
\mB_{TR} = \tr (
\mL_1   \times^1\!
\mL_2   \times^1\!
\sdots  \times^1\!
\mL_l  \times^1\!
\mM_1   \times^1\!
\mM_2   \times^1 \!
\sdots  \times^1\!
\mM_p
),
\end{equation}
of TR rank 
$(s_1,\sdots,s_l,g_1,\sdots,g_p)$,
where $\mL_j$ and $\mM_k$ are third order tensor of sizes $(s_{j-1}\times h_j\times s_j)$ and $(g_{k-1}\times m_k\times g_k)$ respectively, for all $j=1,\sdots,l$ and $k=1,\sdots,p$, and where $g_0\! =\! s_l$ and $s_0\! =\! g_p$. The TR format reduces the
number of unconstrained parameters in $\mB$ from $\prod_{i=1}^l h_i
\prod_{i=1}^p m_i$ to $\sum_{j=1}^l s_{j-1}h_js_j+\sum_{k=1}^p
g_{k-1}m_kg_k$.
To estimate
parameters, we apply the $k$th mode matricization for $k =1,\sdots,p$  
on both sides of~\eqref{eq:model}, yielding
\begin{equation}\label{eq:reg_TRkth}
{\mY_i}_{(k)} =  \mM_{k(2)} G_{ik}^{TR}+E_i
,  \quad E_i \overset{iid}{\sim} 
\N_{[m_k, m_{-k}]'}(0,\sigma^2\Sigma_k,\Sigma_{-k}),
\end{equation}
and the vectorization for $k =1,\sdots,l$, which gives us
 \begin{equation}\label{eq:reg_TRvec}
 \vecc (\mY_i) =   H_{ik}^{TR} \vecc (\mL_k) + \boldsymbol{e}_i
,  \quad \boldsymbol{e}_i\overset{iid}{\sim} 
\N_{m}(0,\sigma^2\Sigma),
 \end{equation}
where $G_{ik}^{TR}$ and $H_{ik}^{TR}$ are matrices as defined  in
Section~\ref{supp:TR}.
Fig. \ref{fig:tot_overview} represents tensor-variate versions of
$H_{ik}^{TR}$ and $G_{ik}^{TR}$ for when  $p\!=\!l\!=\!3$. 
Because \eqref{eq:reg_CPkth} and \eqref{eq:reg_CPvec}  are similar to
\eqref{eq:reg_TRkth} and \eqref{eq:reg_TRvec}, our ML estimators 
mirror the CP format case but by replacing
$(H_{ik}^{CP},G_{ik}^{CP})$ with $(H_{ik}^{TR},G_{ik}^{TR})$. In this
case, estimating $(M_k,\vecc(L_k))$ corresponds to estimating
$(\mM_{k(2)},\vecc(\mL_k))$. The optimization procedure is
similar to Algorithm \ref{alg:CP}, with the difference in this case
being that each factor tensor, other than $\mM_p$, is scaled to have
unit Frobenius norm.
We end here by noting that the special TR case of TT
format has been used for ToTR in \citep{liuetal20}, with $\Sigma=I$,
and hence $\Sigma_k=I$ for all $k$.

\subsubsection*{\underline{Concluding Remarks}} Fig.~\ref{fig:tot_overview}
summarizes our estimation methods for $\mB$ of different formats.
\begin{figure}[h]
\includegraphics[width=1.0\linewidth]{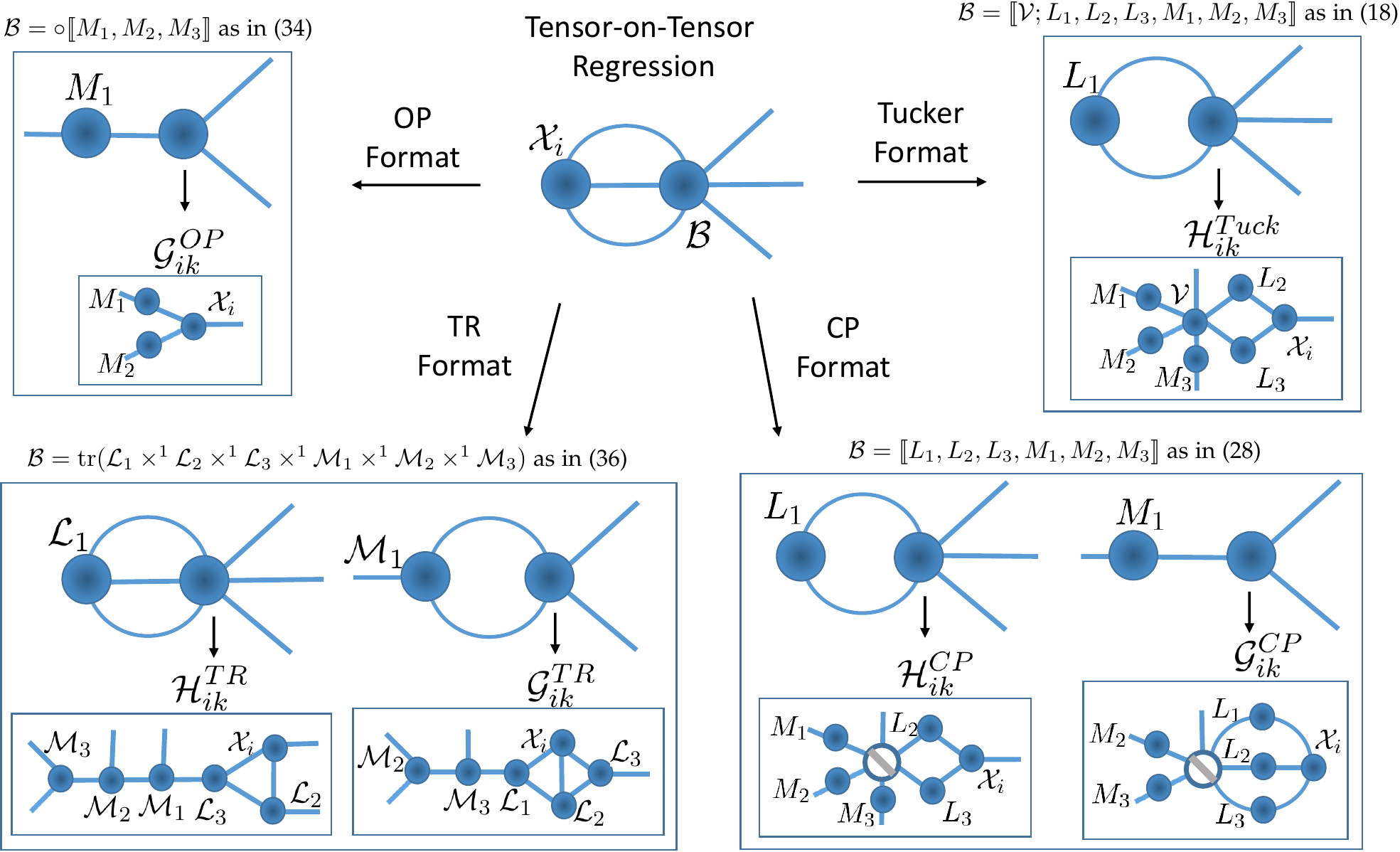}
  \centering
  \caption{Equivalent tensor-network diagrams for $\langle \mX_i, \mB \rangle$ when $p=l=3$, which  can be expressed in multiple ways depending on the tensor factor to be estimated and the type of low-rank on $\mB$ (which are illustrated in Fig. \ref{fig:regressors}). By choosing the tensors $\mG_{ik}$ and $\mH_{ik}$,  
  the tensors $M_1$ and $\mM_1$ can be estimated using multivariate multiple linear regression, and the tensors $L_1$ and $\mL_1$ can be estimated using multiple linear regression, respectively. In these cases, matricized versions of $\mG_{ik}$ and $\mH_{ik}$ are part of the design matrix.}
  \label{fig:tot_overview}
\end{figure}
We see  that in many cases, an algorithmic block can be made to
correspond to a linear model by appropriate choice of $\mG_{ik}$ or
$\mH_{ik}$. Then, fitting a tensor-response linear model involves 
sequentially fitting smaller-dimensional linear models (one for each
tensor factor) until convergence. This intuition behind the estimation
of $\mB$ is not restricted to the TK, CP, OP and TR formats, but can
also help guide estimation algorithms for other formats such as the hierarchical Tucker and the tensor tree formats \citep{hackbuschandkuhn09,grasedyck10}.

While the OP format has the advantage of parsimony, it does not allow
for the level of recovery to be adjusted through a tensor rank (as
will be illustrated in Section \ref{sec:simulation} and Fig.
\ref{fig:simu}). 
The CP format is a more attractive alternative, since it is the
natural generalization of the low-rank matrix format (a tensor with CP
rank $k$ is the sum of $k$ tensors with CP rank $1$). However, the CP
format can be too restrictive in some scenarios, such as when the
tensor modes have very different sizes. In such cases, the Tucker
format provides a more appealing alternative. Moreover, a
Tucker-formatted tensor has the interpretation of being a core tensor
that is stretched on each mode by a tall matrix. However, its
disadvantage is that the core tensor has the same number of modes as
the original tensor, making it impractical in very high-order tensor
scenarios. In such situations, the TR format is preferred because each
additional tensor-mode requires only one additional tensor factor.  

\subsubsection{Estimation of $\sigma^2$}
In all cases, the estimation of $\Sigma_k$ involves finding the sum of squared errors along the $k$-mode $S_k$, as in equations \eqref{eq:Tucker_MLikel} and \eqref{eq:reg_MLECP}. Given the estimated $\hSigma_k$ and $S_k$, the estimate of  $\sigma^2$ is very cheap and given as
\begin{equation}\label{eq:sig2}
\hat{\sigma}^2
= \dfrac{1}{nm}\tr(\hat{\Sigma}_kS_k).
\end{equation}
Thus $\hat{\sigma}^2$ is obtained alternatively within each iteration. Moreover, the log-likelihood function evaluated at the current estimated values is greatly simplified based on $\hat\sigma^2$ as
\begin{equation*}\label{est_likelihood}
\ell = -\dfrac{nm}{2} 
\left[
1 + \log(2\pi \hat{\sigma}^2) + \sum_{k=1}^p \log|\hSigma_k|/m_k
\right].
\end{equation*}

\subsubsection{Initialization and convergence}\label{subsec:computation_normal}
For local optimality, we need the two conditions that the log-likelihood
$\ell$ is jointly continuous and that the set 
$\{ \boldsymbol{\theta} :
\ell(\boldsymbol{\theta})\geq \ell(\boldsymbol{\theta}^{(0)})
\}$, for a set of initial values $\boldsymbol{\theta}^{(0)}$, is 
compact~\citep{deleeuw94}. These conditions are satisfied because
of the TVN distributional assumption on our errors, as long as the
initial values satisfy the constraints on the parameters. We
initialized $\Sigma_k=I_{m_k}$ and the tensor factor entries in $\mB$
 with draws from the $\mU(0,1)$ distribution. With the TK format,
 $M_k$ has the constraint
$M_k'\Sigma_k^{-1}M_k=I_{d_k}$ for $k=1,\sdots,p$, so we used
$\Sigma_k^{\frac12}U$ as its initializer, with $U$ having the left
singular vectors of a random matrix of the same order as $M_k$.  We also suggest using identity matrices to initialize $\Sigma_1,\dots,\Sigma_p$ and $\sigma^2$ can be initialized with 1.
Our algorithms are declared to converge when we have negligible
changes in the loglikelihood, as simplified in Section~\ref{est_likelihood}. A different criteria is the difference in norm  $||\mB||+||\sigma^2\Sigma||$, where
$||\sigma^2\Sigma|| = \sigma\prod_{k=1}^p||\Sigma_k||$ and
$||\mB||$ simplifies as per format:
\begin{itemize}
\item for $\mB_{TK}$, with $A^q$ as the Q matrix from the LQ decomposition of $A$ \citep{kolda06}, we have $||\mB_{TK}|| = ||[\![  \mV ; L_1^q,\sdots,L_l^q,M_1^q,\sdots,M_p^q  ]\!]||$.
\item for $\mB_{CP}$, $||\mB_{CP}||^2 = \sum_{k,l=1}^{R} \big\{diag(\boldsymbol{\lambda})[*_{i=1}^l (L_i'L_i)]*[*_{i=1}^p(M_i'M_i)]\big\}(k,l)$, where ``$*$'' is the Hadamard, or entry-wise product \citep{kolda06}.
\item for $\mB_{OP}$, $||\mB_{OP}|| = \prod_{i=1}^p ||M_i||$.
\end{itemize}
\subsection{Properties  of our estimators}
\subsubsection{Computational complexity}
We derive the computational complexity of our estimation algorithms. Recall that in all cases the response $\mY_i\in\mathbb{R}^{\times_{k=1}^p m_k}$ and the covariate $\mX_i\in\mathbb{R}^{\times_{k=1}^l h_k}$ where $l$ and $p$ are considered fixed.   WLOG, we assume that $m_1\! =\! \max\{m_1,\dots,m_p\}$ and $h_1\! =\! \max\{h_1,\dots,h_l\}$.
\begin{thm}\label{theo:comp}
The computational complexity of our ToTR algorithms when $\mB$ has 
\begin{enumerate}
\item the TK format of Section \ref{sec:estTuck}, with $d = \prod_{q=1}^pd_q$, $d_{-1} = d/d_1$, $c = \prod_{q=1}^lc_q$,  $d_1 = \max\{d_1,\dots,d_p\}$ and $c_1 = \max\{c_1,\dots,c_l\}$ and implemented in Algorithm \ref{alg:1} is 
$\mO\big(
 nhc_1+
 n^2c +
 n^2m_1d_{-1}+
 nm_1^2d_{-1}
 \big) + \mO\big(
nmd_1 
 \big) + \mO\big(
ncdh_1+
nc_1^2h_1^2d+
h_1^3c_1^3
 \big) + \mO\big(
m_1^3 +nmm_1
\big)$.
\item  the CP format of Section \ref{sec:estCP} and implemented in Algorithm \ref{alg:CP} is 
$\mO\big(
nh_1^2r^2+
nrm+
rm_1^2+
m_1^3+
h_1^3r^3+
  nrh+m_1r^2+
  nmm_1
\big).$
\item  the TR format, as described in Section \ref{sec:estTR} and with $g_0g_1 = \max\{g_{k-1}g_k:k=1,2,\dots p\}$, $s_0s_1 = \max\{s_{k-1}s_k:k=1,2,\dots l\}$, $g = \max\{g_0,g_1\}$,  $g_0g_1\leq m_1$ is 
$\mO\big(
mg_1g_0g_p + hnss_l + mh_1^2s_0^2s_1^2+h_1^3s_0^3s_1^3
  \big) +\mO\big(
  hs_{1}s_0s_l+m_1^3+
  g_0^3g_1^3+nmm_1
\big).$
\end{enumerate}
\end{thm}
\begin{proof}
See Sections \ref{Ssec:computTK} - \ref{Ssec:computTR} for the proofs. 
\end{proof}

In all cases we have the term $\mO(nmm_1)$, which is the complexity of
obtaining the sum of square errors $S_k$ of equation \eqref{eq:Tucker_MLikel} across the largest tensor-response
mode, and it is necessary for obtaining the scale matrices
$\Sigma_1,\sdots,\Sigma_p$. In many cases this term will dominate the
computational complexity. However, $O(nmm_1)$ is considerably smaller
than $O(nm^2)$, which would be the case where our complexity increases
quadratically with the dimensionality of the tensor response. We also
note that in all cases, the cubic terms are WRT the tensor ranks,
which can be considered negligible because such ranks are often chosen
to be small, in the spirit of scientific parsimony. Finally,  for the TK
format, some of the factors can assumed to be identity 
matrices, allowing us to further reduce the complexity. 
(See Section \ref{Ssec:computOP} for the computational complexity of the
OP format under specific conditions.)

\subsubsection{Asymptotic sampling distributions}\label{subsec:sampdist}
We now derive the asymptotic distributions of our model-estimated
parameters, specifically, the linear  component and the covariance 
component (Section~\ref{inference:covariance}).

We first explore the limiting distribution of the estimated linear components $\vecc(\mhB)$, which in all cases is multivariate normal with mean $\vecc(\mB)$ satisfying the same low-rank format of $\vecc(\mhB)$. 
For the Tucker format, we first show that the vectorized core tensor $\vecc(\mhV)$ follows a non-singular multivariate normal distribution, and therefore by Slutsky's theorem $\vecc(\mhB)$ follows a singular multivariate normal distribution, where the singularity of the covariance matrix constraints the limiting distribution to the original low-rank Tucker format.
For the CP, TR and OP cases we first show that the low-rank format factors in $\mB$ are jointly normally distributed. Therefore, by the Delta method the estimated tensors $\mhB$ in vectorized form are asymptotic normally distributed. In this case the resulting multivariate normal distribution is also singular, but these are only approximations to the CP, TR or OP formats and not constraints on the limiting distributions like it is the case for the Tucker format.

For the remainder of this paper, we define $
h \doteq \prod_{i=1}^l h_i$, $
M \doteq \bigotimes_{i=p}^1 M_i$ and $
L \doteq \bigotimes_{i=l}^1 L_i$.
We first assume that \eqref{eq:model_noM} holds without an intercept.
\begin{thm}\label{thm:inference_Tucker}
Let \eqref{eq:model_noM} hold with $\mB\equiv \mB_{TK}$ of Tucker
format as in \eqref{eq:reg_Tuck} and let $X = [
vec(\mX_1)\sdots vec(\mX_n)]$ and $\mhB_{TK}
=[\![  \mhV ; \hL_1,\hL_2, \sdots,\hL_l,\hM_1,\hM_2,\sdots,\hM_p  ]\!]$. Then as $n\rightarrow \infty$
$$
\vecc(\mhB_{TK})
\! \overset{d}{\rightarrow} \!
\N_{mh} \!\!
\big(
\vecc(\mB_{TK}),
\sigma^2 (MM') \otimes (P_L(XX')^{-1}P_L)
\big),
$$
where
$
P_L = 
\bigotimes_{i=l}^1 P_i
$
and
$
P_i= L_i(L_i'L_i)^{-1}L_i'.
$
\end{thm}
\begin{proof}
See Section \ref{proof:inference_Tucker}.
\end{proof} 
The limiting distribution in Theorem~\ref{thm:inference_Tucker} is TVN when $XX'$ has a
Kronecker structure: examples include factorial designs and B-splines \citep{currietal06}. Here we present one such case.
\begin{cor}\label{cor:tensnorm}
When $\mhB_{TK}$ is used to estimate a balanced TANOVA with $q$ units for each factor combination, then for $\bs=[h_1,h_2,\sdots,h_l,m_1,m_2,\sdots,m_p]'$, as $n\rightarrow \infty$
$$
\mhB_{TK} 
\overset{d}{\rightarrow} 
 \N_{\bs}\big( \mB,\dfrac{\sigma^2}{q}P_1,P_2,\sdots,P_l,M_1M_1',
 M_2M_2'\sdots,M_pM_p'\big).
$$
\end{cor}
\begin{proof}
  Here $XX'=qI_{h}$ and so the variance-covariance matrix in the limiting distribution of Theorem \ref{thm:inference_Tucker} is $(\sigma^2/q)(MM')\otimes (\bigotimes_{i=l}^1 P_i)$, which is Kronecker-separable. The result follows from Definition \ref{def:tvn}.
\end{proof}
The CP format case is similar to Theorem
\ref{thm:inference_Tucker}.
\begin{thm}\label{thm:inference_others}  
Consider \eqref{eq:model_noM} with $\mB \equiv\mB_{CP}$ as in
\eqref{eq:reg_CP} and the ML estimator
$
\mhB_{CP} 
=[\![ \hL_1,\hL_2, \sdots,\hL_l,\hM_1,\hM_2,\sdots,\hM_p  ]\!].
$
Then
$$
\vecc(\mhB_{CP})
\overset{d}{\rightarrow} 
\N_{mh} \Big(
\vecc(\mB_{CP})
,
J_{CP}R_{CP}( \I_n\otimes \Sigma)R_{CP}'J_{CP}'
\Big)
$$
as $n\rightarrow \infty$.
Here $J_{CP}$ is a Jacobian matrix and is given along with the block matrix $R_{CP}$ in Section \ref{proof:inference_others}.
\end{thm}
\begin{proof}
See Section \ref{proof:inference_others}.
\end{proof}

The sampling distributions of $\mB$ under the OP or TR formats are
similar to the CP case, and  are in theorem \ref{thm:inference_TRnOP:actual} of Section
\ref{sec:inference_TRnOP:actual}. 
\begin{thm}\label{thm:inference_with_intercept}
For a model with intercept, as in \eqref{eq:model}, Theorems \ref{thm:inference_Tucker} and \ref{thm:inference_others} also hold after centering the covariates.
\end{thm}
\begin{proof}
See Section \ref{proof:inference_with_intercept}.
\end{proof}

Section \ref{inference:covariance} also discusses inference on the scale components $\hSigma_1,\hSigma_2,\sdots,\hSigma_p$. Theorem \ref{thm:inference_covariance}  establishes the asymptotic independence of the scale and the linear components, we find the Fisher information matrix WRT the scale parameters, and establish its singularity. Our 
results on the asymptotic distribution of the scale components are not unique to our regression methodology but also generally hold for the TVN distribution.

\subsection{Model selection or rank determination}\label{subsec:rank_choos}
For the CP, TR and TK formats, we determine optimal ranks using the
Bayesian information criterion~(BIC)~\citep{kashyap77,schwarz78}. This
calculation requires the loglikelihood that we simplified in
Section~\ref{est_likelihood}. Section \ref{supp:tpars}  provides
more details on rank determination (equivalently, model selection) and
on the total number of  calculations needed to obtain the BIC.

\section{Experimental Evaluations}\label{sec:simulation}
We study estimation performance of the scale
parameters and the low-rank  linear component $\mB$ using simulation
experiments on different ToTR models. Section~\ref{simu:mat-on-mat}
assesses the consistency of our estimators, and
Section~\ref{simmulation:camelids} evaluates the amounts of recovery 
that different low-rank formats have of $\mB$, 
and the impact of noise on discrimination in a TANOVA framework. 


\subsection{A TANOVA(2,2)  model with low-rank formats}\label{simu:mat-on-mat}

\begin{figure}[ht]
\centering
\includegraphics[width = \linewidth]{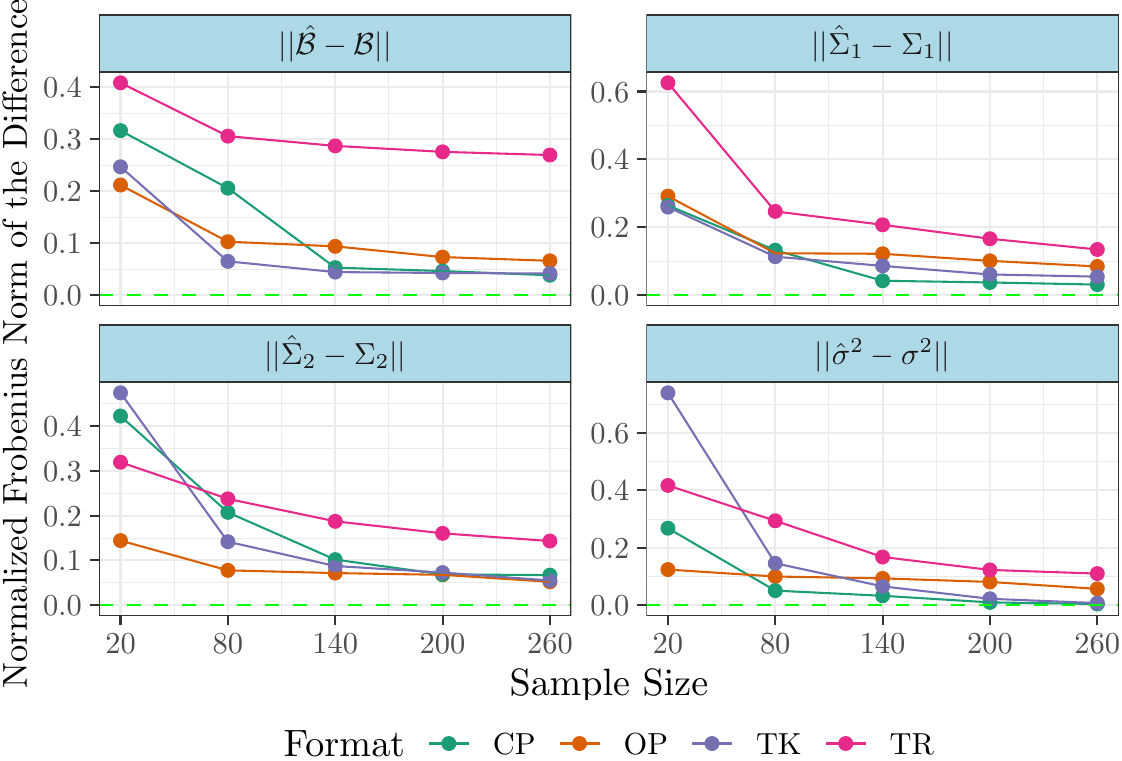}
     \caption{Performance of the four models presented in Section \ref{simu:mat-on-mat}, each corresponding to a different format on $\mB$. Each plot corresponds to the Frobenius norm of the difference between an estimated and  true population parameter, against the sample size. We observe that in all cases, an increase in sample size leads to more accurate estimates. }
     \label{fig:simu}
   \end{figure}

\begin{figure*}[ht]
\subfloat[]{
\includegraphics[width=.57\linewidth]{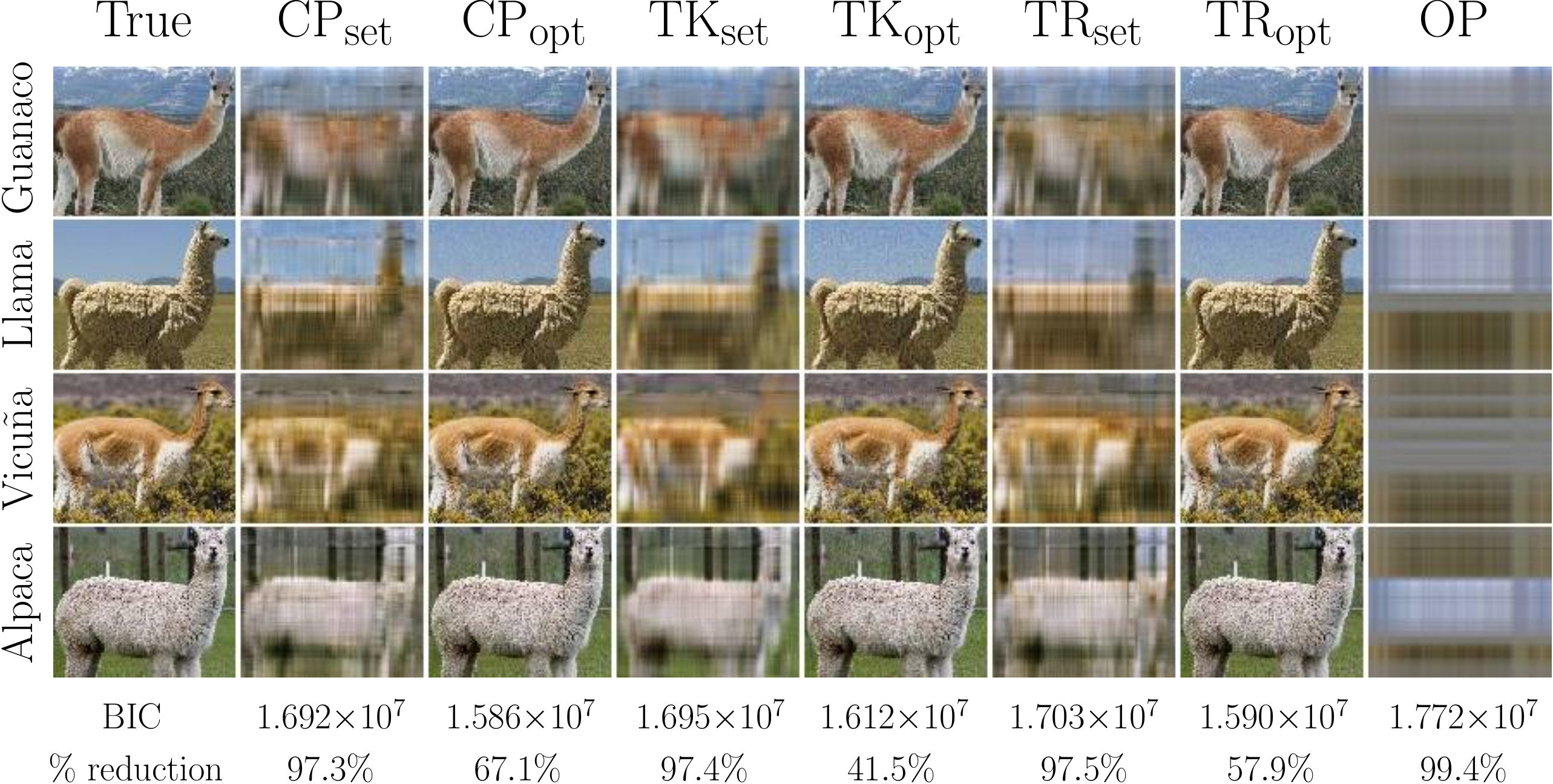}
}
\hspace{.5em}
\subfloat[]{
\includegraphics[width=.39\linewidth]{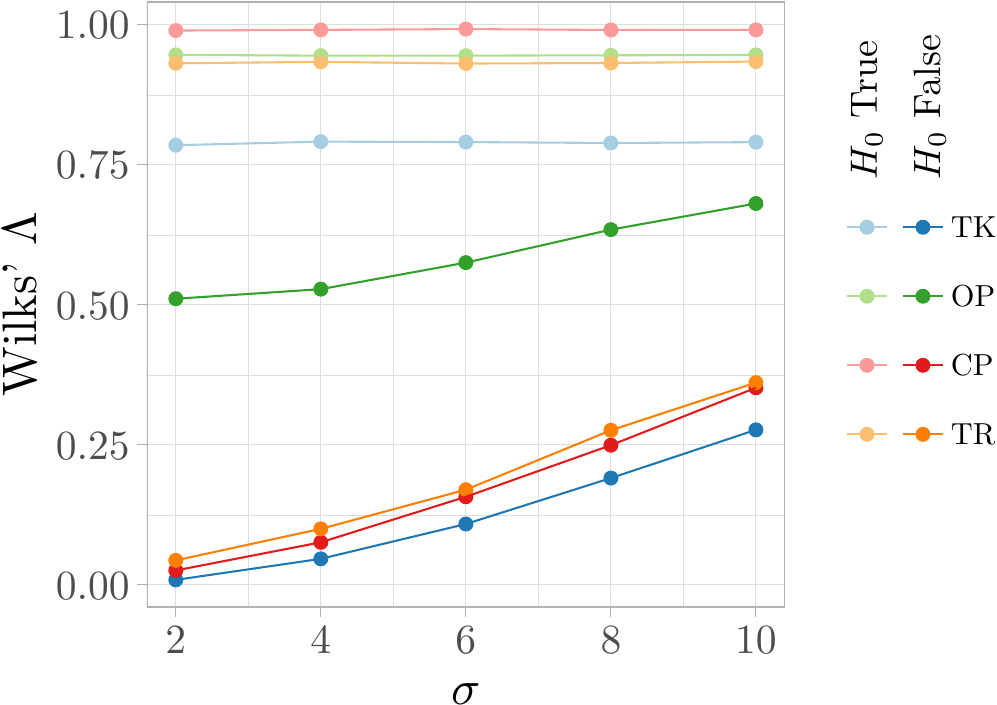}
}
\caption{
(a) 
Results of fitting seven TANOVA(2,2) models on
    data simulated from~\eqref{eq:realsim_camelids}. One model was fit assuming the OP format, and two models were fit for each of the TK, CP and TR formats: 
    one with set ranks and another one with optimal ranks, as chosen by BIC. The two factors are the type of andean camelid
    (Guanaco, Llama, Vicu{\gn}a, Alpaca) and the type of additive
    color RGB (red, green, blue). The Guanaco, Llama and Alpaca images
    are from Wikipedia Commons and the Vicu{\gn}a image is from Encyclopædia
    Britannica. It is evident that while one cannot adjust the OP rank, increased rank for the TR, TK and CP formats result in more image restoration.
(b) Monte-Carlo 95th quantiles of the Wilks' $\Lambda$ statistics that
test the set of hypothesis in equation \eqref{eq:camelids_hypothesis},
for the OP, CP, TR and TK formats, five values of scalar variance
$\sigma^2$ in the x-axis, and for both true and false null
hypotheses. Large variabilities lead to larger test statistics when
$H_0$ is false, leading to weaker evidence against the null
hypothesis. In all cases,  CP, TR and TK formatting leads to more significant statistics when compared to the OP, even when the ranks are not optimal. 
}
  \centering
\label{fig:camelids}
\end{figure*}

We simulated  observations from the matrix-on-matrix regression (MoMR)
model
\begin{equation}\label{eq:sim1}
Y_{ijk} =   \langle X_{ij} | \mB\rangle +E_{ijk}
,\quad
E_{ijk} \overset{iid}{\sim} \N_{[6,7]'}
(0,
\sigma^2
\Sigma_1,\Sigma_2)
\end{equation}
where $i=1, 2,3,4$, $j=1,2,3,4,5$ and
$X_{ij}$ 
is a $4\!\times\! 5$ matrix with 1 at the $(i,j)$th position and zeroes
everywhere else. We set $\sigma^2=1$ and obtained $\Sigma_1$
and $\Sigma_2$ independently from Wishart distributions, that is, we
obtained $\Sigma_1\sim \mW_6(6,I_6)$ and $\Sigma_2\sim \mW_7(7,I_7)$ before
scaling each 
by their $(1,1)$th element.  
We obtained realizations from four MoMR models, one
each for $\mB$ of TK, CP, TR and OP formats, and fit
appropriate models to the data using the ML estimation procedures
described in Section \ref{subsec:parest}. To study consistency 
properties of our estimators, we used $k=1,4,7,10$ and
$13$, meaning that our sample sizes ranged over $n\in\{
20,80,140,200,260\}$. An unstructured $\mB$ in this experiment would
have $4\times 5\times 6\times 7 = 840$ entries, but
$\mB$ has only 59 unconstrained parameters with the OP format and
45 unconstrained parameters when it is of CP format of rank 2. This
number is only 60 when $\mB$ has TK format with rank (2,2,2,2) and
70 when it is of the TR format with rank (2,2,2,2). Thus, our lower-rank simulation
framework had at least $91\%$ fewer unconstrained parameters in $\mB$. We simulated data
from~\eqref{eq:sim1} using $\mB$, $\sigma^2$, $\Sigma_1$ and
$\Sigma_2$ and estimated the parameters for each
replication. Fig. \ref{fig:simu}
displays the Frobenius norm of the difference between the true and
estimated parameters, and shows that as sample size increases,
$(\mhB,\hat{\sigma}^2,\hat{\Sigma}_1,\hat{\Sigma}_2)$ approach
the true parameters $(\mB,\sigma^2,\Sigma_1,\Sigma_2)$, demonstrating
consistency of the estimators. 

\subsection{Evaluating recovery and discrimination}\label{simmulation:camelids}
We simulated 600 observations from the MoMR model
\begin{equation}\label{eq:realsim_camelids}
Y_{ijk} =\langle X_{ij} | \mB\rangle + E_{ijk},\quad E_{ijk}  \overset{iid}{\sim} \N_{[87,106]'}(0,\sigma^2\Sigma_1,\Sigma_2),
\end{equation}
where $i=1,2,3,4$, $j=1,2,3$, $X_{ij}$ is a $4\!\times\! 3$ matrix
with 1 at the $(i,j)$th position  and zero elsewhere, and $\Upsilon_{ij}=\langle X_{ij} | \mB\rangle$ corresponds to the pixel-wise logit transformation of
the $j$th additive color (Red, Green, Blue) of the  $i$th Andean
camelid (Guanaco, Llama, Vicu{\gn}a, Alpaca) images of Fig.
\ref{fig:camelids}(a).  We set
$\Sigma_1$ and $\Sigma_2$ to be AR(1) 
 correlation matrices with coefficients $0.1$ and $-0.1$ respectively,
 and $\sigma^2=1$. Sans constaints, we have $3\times 4 \times 87\times
 106 = 110664$ parameters in $\mB$.  
 
We fit \eqref{eq:realsim_camelids} separately for TR-, TK-, CP- and
OP-formatted $\mB$, with ranks set to have similar number of
unconstrained parameters in $\mB$. The TR format 
with rank (3,3,5,3) had 2958 such parameters, the TK
format with rank  (4,4,9,9) had 3061 parameters, the CP format with
rank 15 had 3001 parameters in $\mB$, while the relatively inflexible
OP format had 666 parameters. In all cases, the dimension of $\mB$ was
reduced by over 97\%. The estimated tensor $\mB$ in each case
corresponds to the estimated color images of the four Andean
camelids. Fig. \ref{fig:camelids}(a) 
shows varying success of these four formats in recovering the
underlying camelid image (true $\mB$). The OP-estimated images are the
least-resolved, with the reduced number of parameters for $\mB$
inadequate for recovery. But the other formats can adjust for the quantum 
of reduction in parameters through their ranks. We illustrate this aspect
by fitting \eqref{eq:realsim_camelids} with $\mB$ having the TK, CP and TR
formats with optimal rank chosen by BIC, following Section
\ref{subsec:rank_choos}. Fig.~\ref{fig:camelids}(a) shows
very good recovery of $\mB$ by these BIC-optimized $\mhB$s, with
unappreciable visual differences in all
cases. In contrast, the model fit 
with unstructured $\mB$ and diagonal $\Sigma$, has a BIC of
$1.64\times 10^{7}$, while fitting a model with a similar $\mB$ but
 Kronecker-separable $\Sigma$ has a BIC of $1.63\times 10^{7}$. The
 CP, TK and TR formats therefore outperform these two alternatives when the ranks are tuned.

The TANOVA(2,2) formulation of  \eqref{eq:realsim_camelids} enables us to test 
\begin{align}\label{eq:camelids_hypothesis}
\begin{split}
&\text{H}_0:\mP_1=\mP_2=\mP_3=\mP_4\hspace{4em}
\mbox{ vs. }
\\&
\text{H}_\text{a}:\mP_i \neq \mP_{i^*}, \mbox{ for some } i\neq i^*\in \{1,2,3,4\}
\end{split}
\end{align}
where $\mP_i$ is a third-order tensor of size $3\times 87\times 106$
that contains the RGB slices of the $i$th Andean camelid image. 
The usual Wilks' $\Lambda$ statistic~\citep{mardiaetal79} is
$
\Lambda = |\hat\Sigma_{R}|/|\hat\Sigma_{T}|,
$
where $\hat\Sigma_{R}$ is the sample covariance matrix of the residuals and $\hat\Sigma_{T}$ is the sample covariance matrix of the simpler model's residuals, which finds a common mean across all camelids. (Section \ref{app:camelidsWilks} details the calculation of $\hat\Sigma_{R}$ and $\hat\Sigma_{T}$.) We illustrate the role of $\sigma^2$ and the low-rank
OP, TR, CP or TK formats in distinguishing the four camelids, as
measured by the Wilks' $\Lambda$ test statistic, in 
Fig.~\ref{fig:camelids}(b), for $\sigma=2,4,6,8,10$. The value of
$\Lambda$ increases with $\sigma^2$, meaning 
that larger variances decrease the power of our test. Further, the
CP, TR and TK formats yield lower-valued (more significant) test
statistics than OP. This finding illustrates the limits of the
less-flexible OP format relative to the others in recovering
$\mB$. Nevertheless, OP joins the other three formats in consistency of
estimation and discrimination, as illustrated by Wilks' $\Lambda$.

\section{Real Data Applications}\label{sec:data_application}
Having evaluated performance of our reduced-rank ToTR methodology, we apply it to the datasets of
Section~\ref{sec:examples}.

\subsection{A TANOVA(1,5) model for cerebral activity}
\label{application:suicide}
Section~\ref{subseq:suicide} laid out a TANOVA model involving 30
fMRI volumes of voxel-wise changes in activation from a baseline, each
volume  corresponding to one of ten words connoting death, positive or 
negative affects, for each of
17 subjects. 
For the $j$th subject
we have a fifth-order   tensor $\mY_{j}$ of order $3\!\times\! 10\!\times
\!43\!\times\! 56\!\times\! 20$, where the first two modes
correspond to the three kinds of word stimulus and the individual
words, and the other modes correspond to the dimensions of
the image volume. The $j$th subject has status given by 
$\boldsymbol{x}_{j}$ that is a 2D unit vector with 1 at position 
$i$ that is 1 for attempter or 2 for ideator. We model these responses and covariates as 
\begin{equation*}
\mY_{j} = \langle \boldsymbol{x}_{j}|\mB\rangle + \mE_{j},\quad  \mE_{j} \overset{iid}{\sim}
\N_{\bbm_1}(0,\sigma^2\Sigma_1,\Sigma_2,\Sigma_3,\Sigma_4,\Sigma_5),
\end{equation*}
where $\bbm_1=[3,10,43,56,20]'$ and $j=1,\sdots,17$. We let 
$\mB $ have the TK format $[\![\mV;L_1,M_1,M_2,M_3,M_4,M_5 ]\!]$
with rank $(2,3,6,15,20,7)$ chosen by BIC from 256 candidate ranks, and where $M_k'\Sigma_k^{-1} M_k\! = \!I_{d_k},k=1,2,3,4,5$.
The 77578 parameters to be estimated in our $\mB$ represent an over 
97.3\% reduction over that of the unconstrained $\mB$ of size
$2\!\times\! 3\!\times\! 10\!\times\! 43\!\times\! 56\!\times\! 20$,
or 2889600 parameters. (Our  use 
of a TK format exploits its nicer distributional
properties for easier inference, and therefore we only use this format here.) We set $\Sigma_1$ (specifying relationships
between word types) to be unconstrained, 
$\Sigma_2$ (covariances between same kinds of words) to have an
equicorrelation structure and $\Sigma_3,\Sigma_4,\Sigma_5$ with AR(1)  
correlations to capture spatial context in the image volume. 
Fitting the model with unstructured $\mB$ and diagonal $\Sigma$
yielded a BIC of $210$ million, while the fitted model with
a similar $\mB$ but Kronecker-separable $\Sigma$ reported a BIC of $164$ million. In contrast, our TK model with Kronecker-separable covariance outperformed these two alternatives with a BIC of $127$ million.

Our primary interest here is to find regions of significant
interaction between word type and subject
suicide attempter/ideator status to determine markers for suicide risk assessment and intervention. The interaction estimate can be expressed as 
$
\mhB_{*}=\mhB\times_1  \boldsymbol{c}_1' \times_2  C_2 \times_3
\boldsymbol{c}_3',
$
where $\boldsymbol{c}_1$ is a contrast
vector that finds differences between suicide attempter/ideation
status, $C_2$ is a contrast matrix for differences across word type
and $\boldsymbol{c}_3$ is a contrast vector that averages the ten
words of each type. These contrast matrices and vectors are given in
\eqref{inter:brain}. From Theorem \ref{thm:inference_Tucker},
\begin{equation*}\label{inter:asympt}
\mhB_{*}
\overset{d}{\rightarrow} 
 \N_{\bbm_2}( \mB_{*},\tau^2C_2M_1M_1'C_2',M_3M_3',M_4M_4',M_5M_5'),
\end{equation*}
where $\bbm_2=[3,20,43,56]'$, 
$\mB_{*}=\mB\times_1  \boldsymbol{c}_1' \times_2  C_2 \times_3
\boldsymbol{c}_3'$ and
$\tau^2$ is as in Section \ref{sec:app:brain_interaction}. Using
the asymptotic distribution of $\mhB_{*}$, we  marginally
standardize it to obtain $\mhZ_*$ as shown in \eqref{stand_statistic}, which also follows the TVN distribution but with correlation matrices as scale parameters. In Section \ref{sec:app:brain_interaction} we detail its derivation, interpretation and asymptotic distribution.
 For the $i$th level interaction, consider the set of hypotheses at the $(k,l,m)$th voxel
\begin{equation*}\label{hypotheses}
H_o: \mB_{*}(i,k,l,m) = 0
\quad
\text{ vs }
\quad
H_a: \mB_{*}(i,k,l,m) \neq 0 .
\end{equation*}
Under the null hypothesis of no $i$th interaction effect at the  $(k,l,m)$th voxel, the marginal distribution of $\mhZ_{*}
(i,k,l,m)$ is asymptotically $N(0,1)$.
\begin{figure}[h]
\vspace{-0.3in}
  \mbox{
    \subfloat{
      \hspace{-0.25in}
    \begin{minipage}[b][][t]{0.99\linewidth}
      \mbox{
    \setcounter{subfigure}{-2}
  \subfloat[death-negative]{
    \begin{minipage}[b][][t]{\linewidth} 
      \mbox{\subfloat{\includegraphics[width=.485\linewidth]{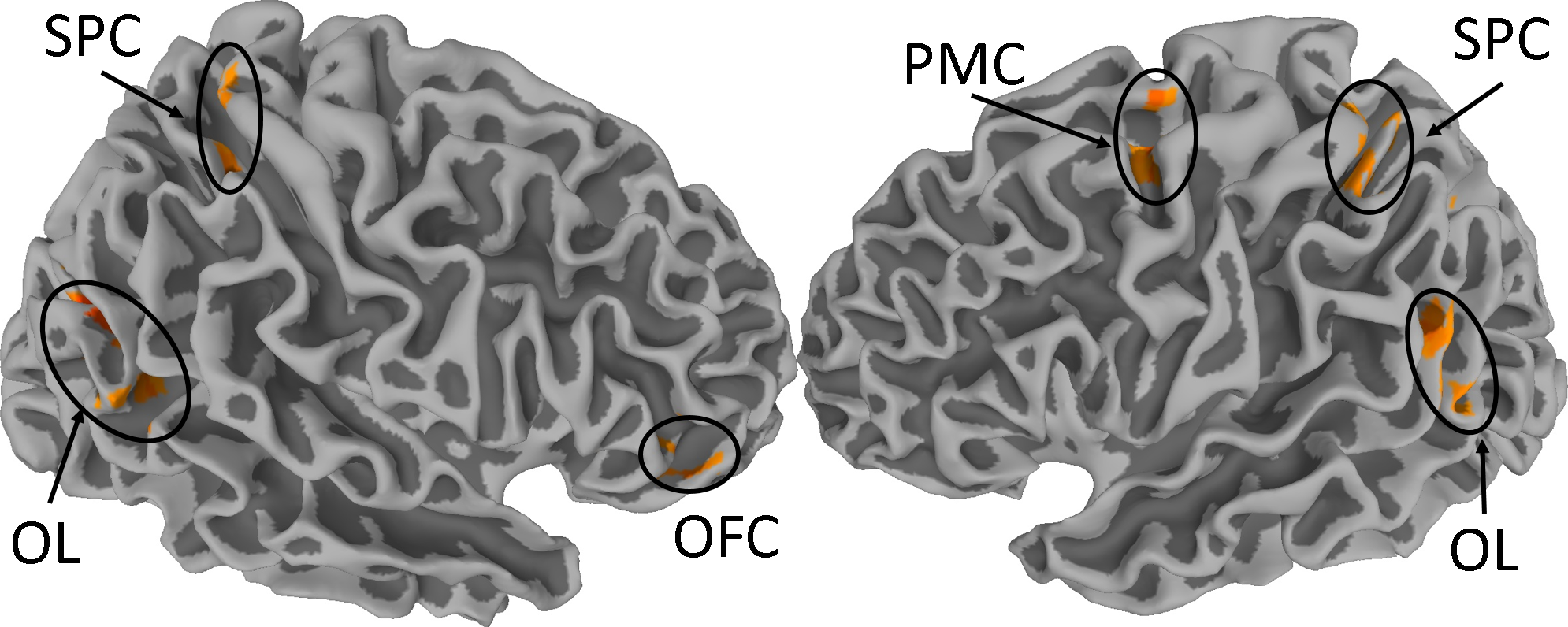}}
        \subfloat{\includegraphics[width=.485\linewidth]{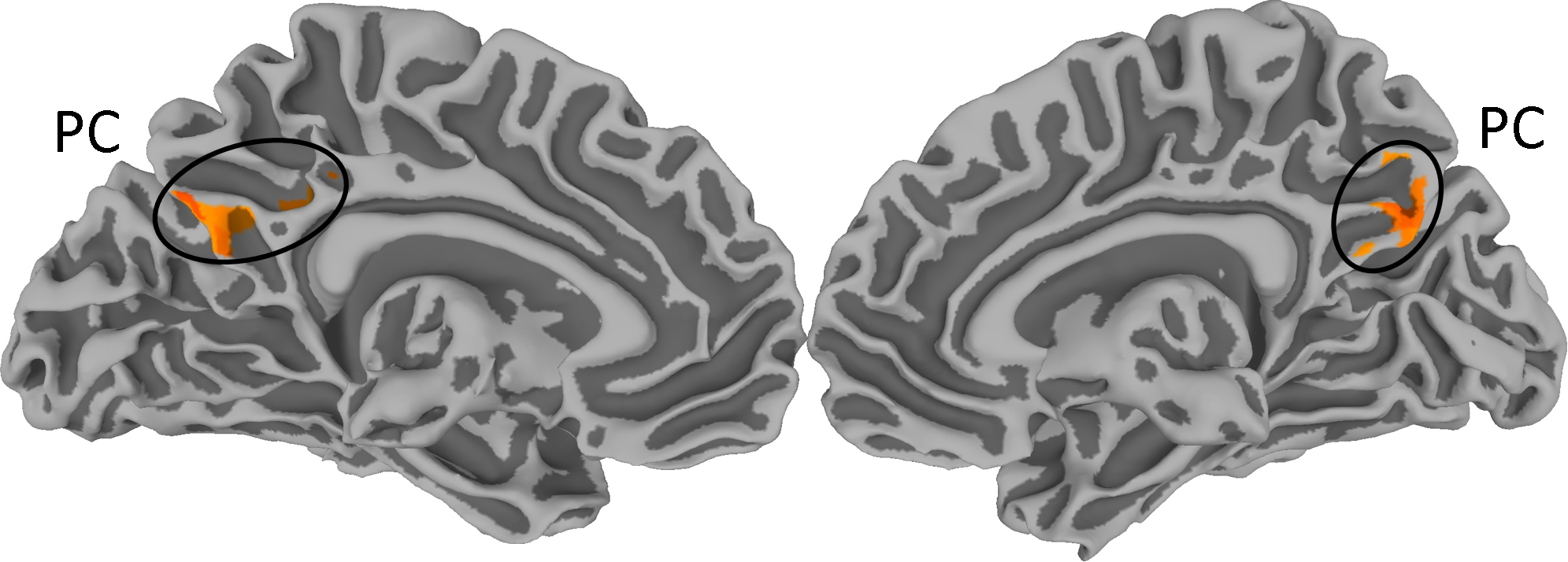}}
      }
\vspace{-0.1in}
      \end{minipage}}%
  }
\mbox{
  \setcounter{subfigure}{-1}
\vspace{-0.1in}
  \subfloat[death-positive]{
    \begin{minipage}[b][][t]{\linewidth} 
      \mbox{\subfloat{\includegraphics[width=.485\linewidth]{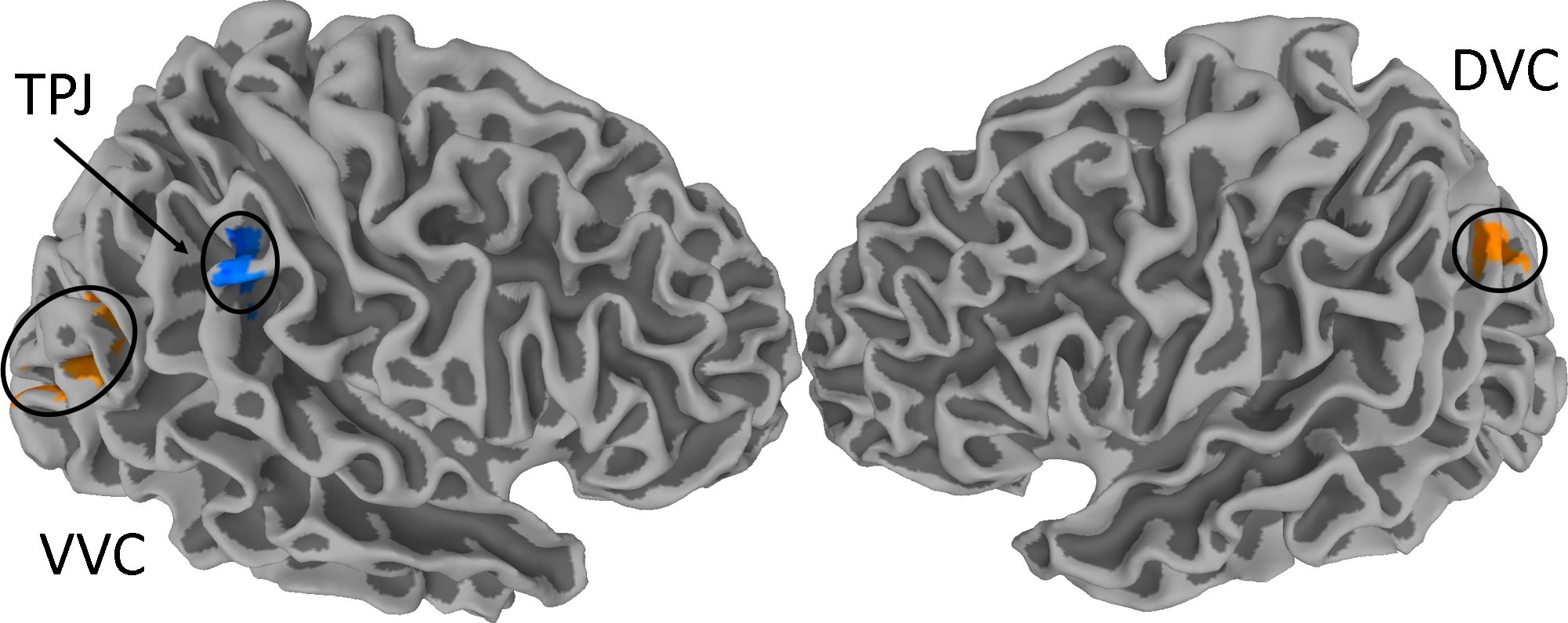}}
        \subfloat{\includegraphics[width=.485\linewidth]{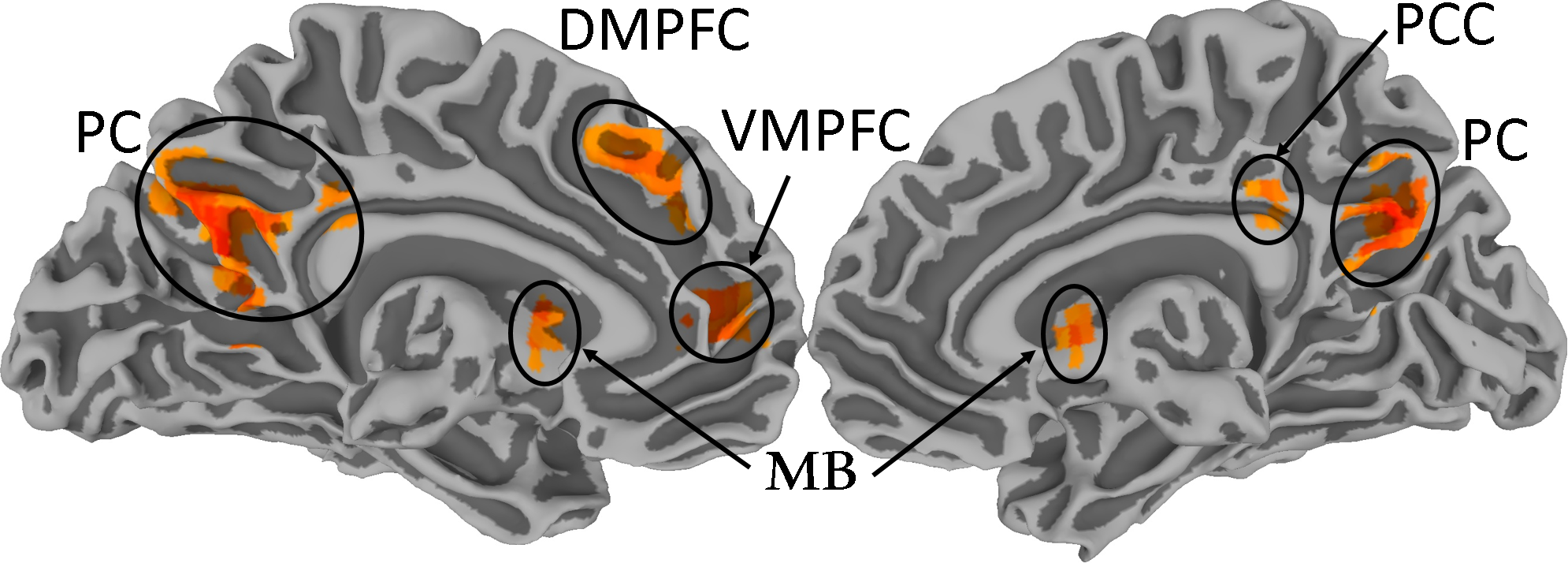}}
      }
\vspace{-0.1in}
      \end{minipage}}%
  }
\mbox{
\vspace{-0.1in}
  \setcounter{subfigure}{0}
  \subfloat[negative-positive]{
    \begin{minipage}[b][][t]{\linewidth} 
      \mbox{\subfloat{\includegraphics[width=.485\linewidth]{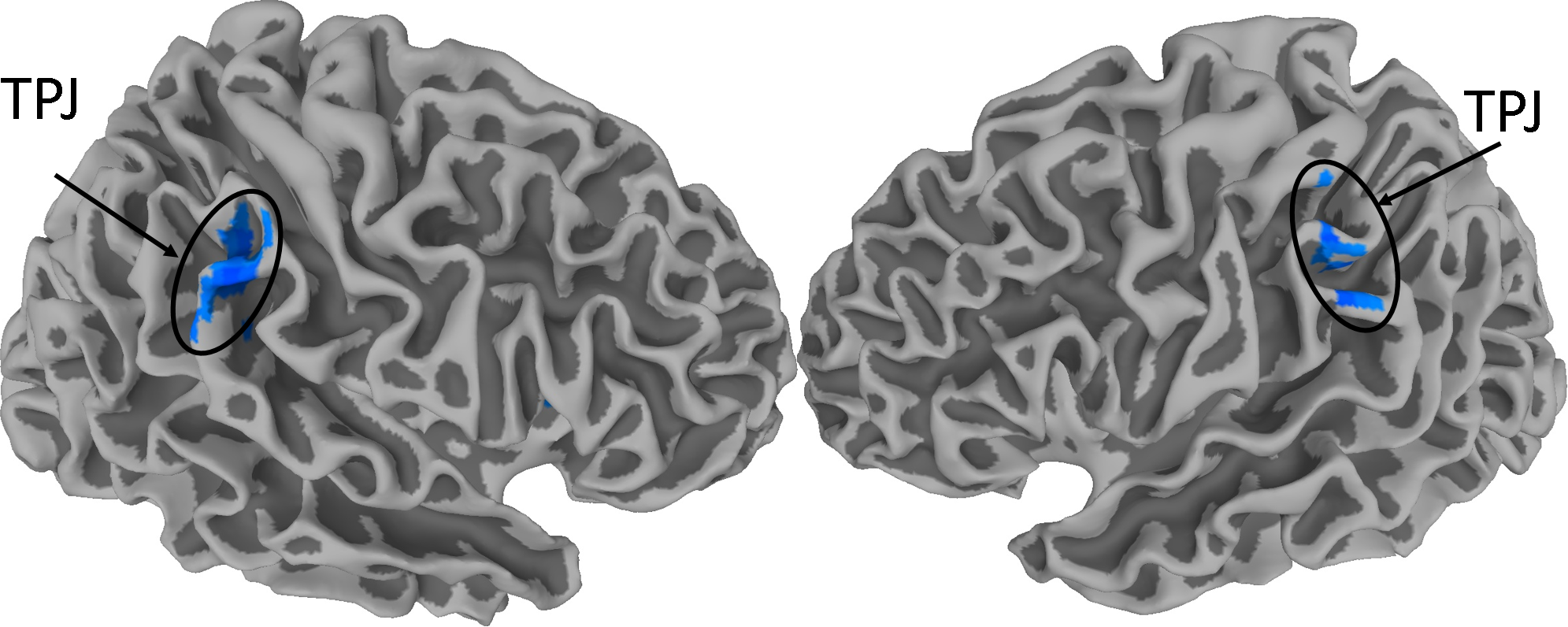}}
        \subfloat{\includegraphics[width=.485\linewidth]{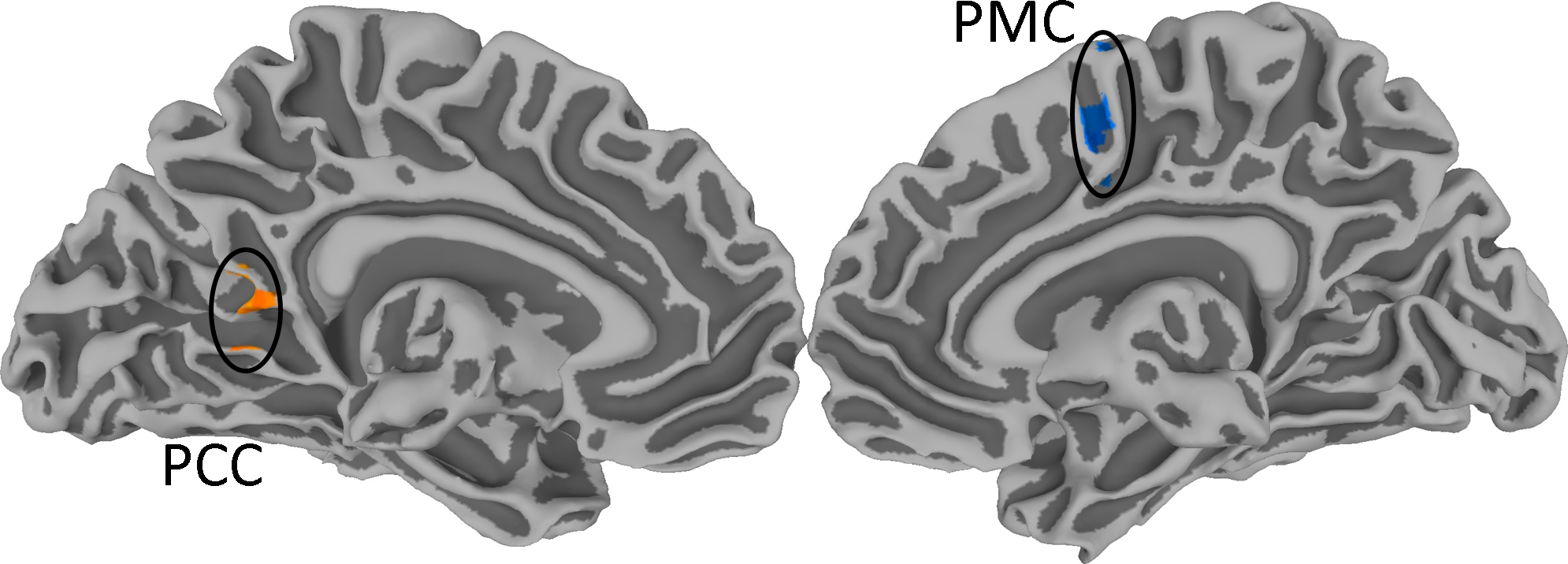}}
      }
\vspace{-0.1in}
      \end{minipage}}%
\vspace{-0.1in}
}
\end{minipage}
}
\hspace{-0.05\linewidth}
 \raisebox{.05\height}{\subfloat{\includegraphics[width=.075\linewidth]{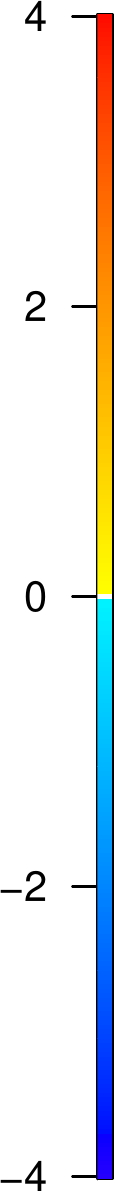}}}
}
\vspace{-0.1in}
  \caption{The test statistic $\mhZ_{*}$ of the interaction between
    subject's attempter/ideator status  and (a) death-negative, (b)
    death-positive and (c) negative-positive words at voxels
    identified significant by cluster thresholding at the 5\%
    level. These voxels are in the precuneus (PC), temporal-parietal
    junction (TPJ), orbital frontal cortex (OFC), premotor cortex (PMC),
    superior parietal cortex (SPC), ventral visual cortex (VVC),
    dorsal visual cortex (DVC),  dorsal medial frontal cortex (DMPFC),
    ventral medial prefrontal cortex (VMPFC),  mamillary bodies (MB),
    posterior cingulate cortex (PCC) and occipital lobe (OL).}
  \label{fig:brain}
\end{figure}
Fig.~\ref{fig:brain} displays 3D maps of the brain with significant
values of  $\mhZ_*$ overlaid for each of the three pairs of
interactions.
Significant voxels were decided using
cluster thresholding~\citep{helleretal06} ($\alpha=0.05$), with
clusters of at least 12 contiguous (under a second-order 
neighborhood specification) voxels, with this minimum cluster size determined by the
Analysis for Neuroimaging (AFNI) software~\citep{cox96,coxandhyde97}.
There are many methods~\citep{genoveseetal02,benjaminiandheller07,smithandfahrmeir07,smithandnichols09,tabelowetal06,polzehletal10,almodovarandmaitra19}
for significance detection in fMRI studies 
but we use cluster thresholding here as an illustration and also because it is the most popular method. We now briefly discuss the results.

Fig.~\ref{fig:brain}(a) identifies significant 
interactions between death- and negative-connoting words on the one hand
and suicide attempters vis-a-vis ideators on the other. All
significant interactions are positive and dominated by the precuneus and the
orbital frontal cortex. The precuneus is associated with depression
and rumination \citep{chengetal18,jacobetal20,zhouetal20}, while the
orbital frontal cortex is associated with the influence that emotions
and feelings have on decision-making \citep{becharaetal00}, as well as
with suicide attempters' reactions to external stimuli
\citep{jollantetal08}. Both regions are also associated with the
Default Mode Network (DMN) that plays a role in representing emotions
\citep{satputeetal19}. These results indicate more differential rumination and
emotions (between attempters and ideators) caused by death-related
words, as compared to negative-connoting words. These findings 
are reinforced by the significance detected in the occipital lobe, the
premotor cortex (PMC) and the superior parietal cortical regions that are 
related to working memory and depression
\citep{simonetal02,koenigsetal09,malleretal14}. Fig.~\ref{fig:brain}(b)
displays significant interactions between the positive and
death-related words and suicide attempters and ideators. The
precuneus is more pronounced here relative to
Fig.~\ref{fig:brain}(a), indicating that death-related words are
more salient than words that have negative and positive connotations
among attempters vis-a-vis ideators. This observation is reinforced
with the detected significance in the dorsal and ventral visual medial
prefrontal 
cortex, the mammillary bodies, and the posterior cingulate cortex (PCC) that
are all involved in processing emotional information
\citep{smithetal18,rolls19}. The PCC is also involved in memory, emotion,
and decision-making \citep{bubbetal17,heilbronner11} and is 
connected to the temporal-parietal junction \citep{zhaoetal16a}
which is involved with emotions and perception~\citep{zaitchiketal10,lettierietal19}. High $\mhZ_*$ values in the
ventral and dorsal visual cortices are commensurate with their
association with working memory tasks~\citep{ungerleideretal98}. Also,
the low values of $\hat\mZ_*$ in the temporal parietal junction point
to needed additional processing of  death-related versus
positive-connoting words 
among attempters relative to ideators.
Fig.~\ref{fig:brain}(c) shows significant interactions between
negative and positive-emoting words and suicide attempters and
ideators. The low values in the left and right temporal-parietal
junctions and the PMC indicate that words conveying negative thoughts
don't need as much processing as do positive-connoting words among
attempters relative to ideators.  The significant assocication of the
PCC in both 
Figs.~\ref{fig:brain}(b) and (c) supports our hypothesis that
death-related words are more salient than negative or positive words
in differentiating attempters from ideators. 
In summary, the two groups of subjects have 
positive- and negative-connoting words result in neurally similar
significant brain regions when compared to death-related words, which
show further significance in areas associated with the processing of
emotional feelings and planning. Our conclusions here are on an
experiment with only 9 attempters and 8 ideators and so are
preliminary, but are intrepretable, providing some confidence in the practical
reductions afforded by TANOVA when coupled with the use of the
Tucker-formatted $\mB$ for this application. 
\subsection{A TANOVA(3,3) model for the LFW face database}\label{application:LFW}

We return to the LFW database of Section~\ref{subseq:lfw} that is a
compendium of over 13,000 face images. Using the steps detailed in Section
\ref{sec:app:lfw}, we selected 605 images with unambiguous genders,
age group and ethnic origin, and such that there are at most 33 images
for each factor combination. 
This dataset was also used by \citet{lock17}  with the goal of
classification, leading to a vector-variate response of attributes
and tensor-valued covariates of color images, for which a CP format
was assumed. In contrast, our objective is to distinguish the
characteristics of different attributes, leading to a 
TANOVA(3,3) model with color images as the response and  
gender, ethnic origin and gender as 
covariates. Our model is as per \eqref{eq:model_noM} and specifically
$$
\mY_{ijkl} = \langle \mX_{ijk}|\mB\rangle + \mE_{ijkl}
,\quad 
\mE_{ijkl} \sim \N_{\bbm_3}(0,\sigma^2\Sigma_1,\Sigma_2,\Sigma_3),
$$
where 
$
i=1,2
,\quad j=1,2,3
,\quad k=1,2,3,4
,\quad l=1,\sdots,n_{ijk},
$
$\bbm_3 = [151,111,3]'$, and the $(i,j,k,l)$th response $\mY_{ijkl}$ is the color image of
 size $151\!\times\! 111\!\times\! 3$ for the $l$th  person
 of the $i$th gender, $j$th ethnic
 origin and $k$th age group.  Here $\mX_{ijk}\!=\!\be_i^2
 \circ\be_j^3 \circ\be_k^4 $ is the 
 tensor-valued covariate for a TANOVA(3,3) model with $(h_1,h_2,h_3)\!=\! (2,3,4)$, as described in Section \ref{subseq:TANOVA}, encoding the
 genders$\times$ethnic-origin$\times$age-group 
 attributes of  $\mY_{ijkl}$.
\begin{figure}[h]
\includegraphics[width=\linewidth]{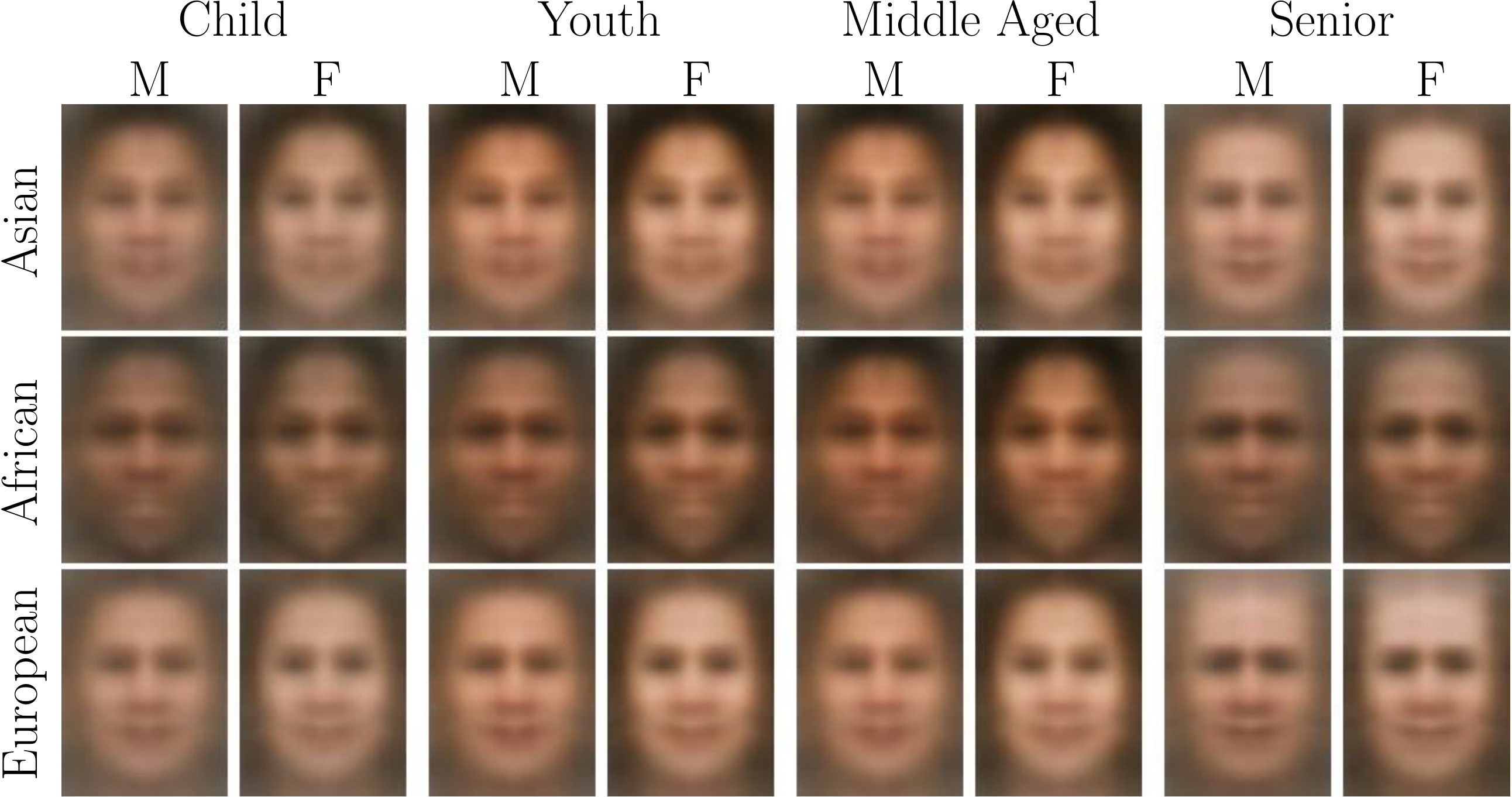}
  \caption{Different slices of the resulting factorized tensor $\mB$ that results from fitting a TANOVA(3,3) model on the LFW dataset using the TT format. The results are compressed mean images across genders (male, female), ethnic origin (Asian, African, European) and age groups (child, youth, middle aged and senior) from 605 central LFW images. We can observe  that the TT format preserved vital information regarding the factor-combination of age group, gender and ethnic-origin.  }
  \label{fig:lfw}
\end{figure}
The corresponding TANOVA parameter $\mB$ is of size
$2\!\times\!3\!\times\!4\!\times\!151\!\times\!111\!\times\!3$ and contains all the group
means. We constrained $\mB$ to have a tensor train (TT) format of TR rank
$(1,3,3,4,10,3)$, chosen using BIC out of a total of 64 candidate
ranks. (In terms of the BIC, the TT format also bested the TK, CP
and OP formats.) 
The number of parameters involved in $\mB$ is 6393 due to the TT
restriction, which is a reduction in the number of unconstrained
parameters of around $99\%$ from the unconstrained $\mB$
that has more than 1.2 million parameters. Fig. \ref{fig:lfw}
displays the estimated  $\mhB$, from where we observe that the TT
format preserved visual information regarding ethnic origin, gender,
and age-group.   
Fitting the model with unstructured $\mB$ and diagonal $\Sigma$
resulted in a BIC of $1.02\times 10^{9}$, while the fitted model with
a similar $\mB$ but Kronecker-separable $\Sigma$ reported a BIC of $-1.22\times 10^{7}$. In contrast, our TT model with Kronecker-separable covariance outperformed these alternatives with a BIC of $-1.87\times 10^{7}$.

\section{Discussion}\label{sec:discussion}
We have provided a multivariate regression and ANOVA framework
that exploits the tensor-valued structure of the explanatory and
response variables using four different low-rank formats on the
regression coefficient and a Kronecker-separable structure on the
covariance matrix. These structures are imposed for context but more so
for practical reasons,
as the number of parameters involved in the classical MVMLR model grows exponentially with the tensor
dimensions. Different stuctures can be compared between each other
using criteria such as BIC. We provided algorithms for ML
estimation, derived their computational complexity, and evaluated them via simulation experiments. We also studied the asymptotic properties of our estimators and applied our
methodology to  identify brain regions associated with suicide attempt
or ideation status and  death- negative- or positive-connoting words. Finally, we 
also used our methods to distinguish facial characteristics in the LFW dataset.

There are several other avenues for further investigation. For
instance, we can perform additional dimension reduction by adding an
$L_1$ penalty on the likelihood optimization. Also, the number of
parameters in the intercept can potentially grow when the tensor response is
high-dimensional, which motivates specifying a low-rank structure on the intercept
term. Similarly, the independent and identically distributed assumption
on the errors is not feasible when
external factors group data-points into units that are similar to one
another. For these cases, a mixed-effects model is more
appropriate.
Also, \citet{papadogeorgouetal21} recently demonstrated that in the
scalar-on-tensor regression case, the use of the CP or TK format can
induce a block structure along the direction of the modes of the
tensor, a phenomenon that we also saw in Fig.~\ref{simu:mat-on-mat}
when the ranks were set, but that went away when BIC was allowed to
tune these ranks. Some of our further 
investigations~(Section~\ref{evaluation-arxiv}) support this latter
finding in a 2D regression framework, however additional 
investigations and palliative measures may be needed.
Further, it would be worth investigating generalization of the Kronecker separable structure of the dispersion matrix or the
normality assumption to  incorporate more general  
distributional forms. Finally, it would be interesting to study the exact distribution of Wilks' $\Lambda$ statistic or other statistic that can be used for testing hypothesis in our TANOVA framework without the need to do simulation.  These are some issues that may benefit from
further attention and that we leave for future work.

\section*{Acknowledgments}
The authors are very grateful to B. Klinedinst and A. Willette of the
Program of Neuroscience and the Department of Food Sciences and Human
Nutrition at Iowa State University for help with the interpretation of
Fig.~\ref{fig:brain}. This research was supported in part by the
      National Institute of Justice (NIJ) under Grants
      No. 2015-DN-BX-K056 and 2018-R2-CX-0034. The research of the
      second author was also supported in part by the National
      Institute of Biomedical Imaging and Bioengineering (NIBIB) 
of the National Institutes of Health (NIH) under Grant R21EB016212,
and the United States Department of Agriculture (USDA) National
Institute of Food and Agriculture (NIFA) Hatch project IOW03617.
The content of this paper is however solely the responsibility of the
authors and does not represent the official views of the NIJ, the NIBIB, the NIH, the NIFA or the USDA.

\section*{\Large Supplementary Appendix}

\renewcommand\thefigure{S\arabic{figure}}\setcounter{figure}{0}
\renewcommand\thetable{S\arabic{table}}\setcounter{table}{0}
\renewcommand\thesection{S\arabic{section}}\setcounter{section}{0}
\renewcommand\theequation{S\arabic{equation}}\setcounter{equation}{0}

\section{Supplement to Section \ref{sec:methodology}}
\label{app:multilinear_statistics}
\subsection{Some Matrix Algebra Properties}
In this section we provide identities and detail notation used in the
main paper, in preparation for the proof of Lemma~\ref{lemma1} and Theorem \ref{thm:mat_outer_product}. While the $\vecc(.)$ operator stacks the columns of a matrix into a vector, the commutation matrix $K_{k,l} \in \mathbb{R}^{kl \times kl}$ matches the elements of $\vecc(A)$ and $\vecc(A')$, 
\begin{equation}\label{supp:vec}
\vecc{(A')} = K_{k,l}\vecc{(A)}, \quad
A \in \mathbb{R}^{k\times l}.
\end{equation}
The commutation matrix plays a critical role in tensor algebra, as it allows us to reshape tensors inside their vectorization. 
If $A$ is a symmetric $n\times n$ matrix, then $\vecc(A)$ has $n^2$
elements of which $n(n-1)/2$ are repetitions, while the half-vectorization $\vech(A)$ contains the $n(n+1)/2$ unique elements of $A$, since it does not take into account the elements above the diagonal. The duplication matrix $D_n$ maps the elements of $\vecc(A)$ and $\vech(A)$ as
$
D_n\vech(A) = \vecc(A),
$
and is a full-column rank matrix of size $n^2\times(n(n+1)/2)$.
The following is from \citet{magnusandneudecker99}:
\begin{property}\label{prop:Com}
Let $A_1 \in \mathbb{R}^{m \times n}$ and $A_2 \in \mathbb{R}^{p
  \times q}$ and $K_{\cdot,\cdot}$ be a commutation matrix. Then
\begin{enumerate}[label=\alph*.]
\item \label{prop:Com:a}
$K_{\cdot,\cdot}$ is orthogonal and flipping the arguments
  results in its transpose. That is, 
$$K_{m,n}' = K_{m,n}^{-1} = K_{n,m}.$$
\item \label{prop:Com:b}
$K_{\cdot,\cdot}$ can be used to change the order of the
  Kronecker product. That is, 
$$K_{p,m}(A_1 \otimes A_2)K_{n,q} = A_2 \otimes A_1.$$
\item \label{prop:Com:c}
$K_{\cdot,\cdot}$ can be used to split the Kronecker product
  inside the vectorization. That is, 
$$\vecc(A_1\otimes A_2) = R_{A_1}\vecc(A_2),\quad R_{A_1} = (I_n \otimes K_{q,m})(\vecc A_1 \otimes I_q).$$
\end{enumerate}
\end{property}
The Kronecker product ($\otimes$) between $A$ and $B$ results in the following block matrix:
\begin{equation}
A\otimes B = 
\begin{bmatrix}
A(1,1) B&\hdots&A(1,m) B\\
\vdots&\ddots&\vdots\\
A(n,1) B&\hdots&A(n,m) B
\end{bmatrix}.
\end{equation}
The following properties follow recursively from the two-matrix cases in \citet{petersenandpedersen12}.
\begin{property}\label{prop:Kron} Let $A_1,A_2,\hdots A_p$ be matrices of any size and $\Sigma_1,\Sigma_2,\hdots,\Sigma_p$ be positive definite matrices of any size. Then
\begin{enumerate}[label=\alph*.]
\item \label{prop:Kron:a}
$
\big| \bigotimes\limits_{i=p}^1 \Sigma_i |  
= \prod\limits_{i=1}^p | \Sigma\_i|^{m_{-i}}
,\quad
m_{-i} = m/m_i
$
\item \label{prop:Kron:b}
$
\big(\bigotimes\limits_{i=p}^1 \Sigma_i \big)^{-1} 
=\big(\bigotimes\limits_{i=p}^1 \Sigma_i^{-1} \big) 
$,\quad
$
\big(\bigotimes\limits_{i=p}^1 A_i \big)' 
=\big(\bigotimes\limits_{i=p}^1 A_i' \big).
$
\item \label{prop:Kron:c}
$
\bigotimes\limits_{i=p}^{1} A_i = \big( \bigotimes\limits_{i=p}^{l+1} A_i )
\otimes \big( \bigotimes\limits_{i=l}^{1} A_{i} )
$,\quad
$l = 2,\hdots, p-2$.
\item \label{prop:Kron:d}
$
\big( \bigotimes\limits_{i=p}^1 A_i\big)\big( \bigotimes\limits_{i=p}^1 B_i \big)
 =  \bigotimes\limits_{i=p}^1 (A_iB_i)$
,\quad where $A_i$ has as many columns as the rows of $B_i$.
\end{enumerate}
\end{property}
\subsection{Some Tensor Algebra Properties}\label{tensoralgebra}
The vectorization, $k$th-mode matricization and the $k$th canonical
matricization of a $p$-th order tensor $\mX$ of size $m_1\times
m_2\times \ldots\times m_p$ as defined in Table \ref{table:1} and Equation~\eqref{eq:mat} are
\begin{equation}\label{defvec}
\vecc (\mathbf{\mX})
= \sum\limits_{i_1=1}^{m_1} \hdots \sum\limits_{i_p=1}^{m_p} \mX(i_1 \hdots i_p)
\big( \bigotimes\limits_{q =p}^1 \be_{i_q}^{m_q} \big),
\end{equation}
\begin{equation}\label{kmode}
\mX_{(k)} 
=\sum\limits_{i_1=1}^{m_1} \hdots \sum\limits_{i_p=1}^{m_p} \mX(i_1 \hdots i_p)
\be_{i_k}^{m_k} \big( \bigotimes\limits_{\overset{q=p}{q\neq k}}^1 \be_{i_q}^{m_q} \big)',
\end{equation}
and
\begin{equation}\label{canonical}
\mX_{<k>}  
= \sum\limits_{i_1=1}^{m_1} \hdots \sum\limits_{i_p=1}^{m_p} \mX(i_1 \hdots i_p)
\big( \bigotimes\limits_{q =k}^1 \be_{i_q}^{m_q} \big)
\big( \bigotimes\limits_{q =p}^{k+1} \be_{i_{q}}^{m_{q}} \big)',
\end{equation}
respectively. To illustrate these reshapings, consider the third-order
tensor $\mX\in \mathbb{R}^{3 \times 4 \times 2}$ illustrated in
Fig.~\ref{fig:tensim} and that
\begin{figure}[h]
\includegraphics[width = 0.35\textwidth]{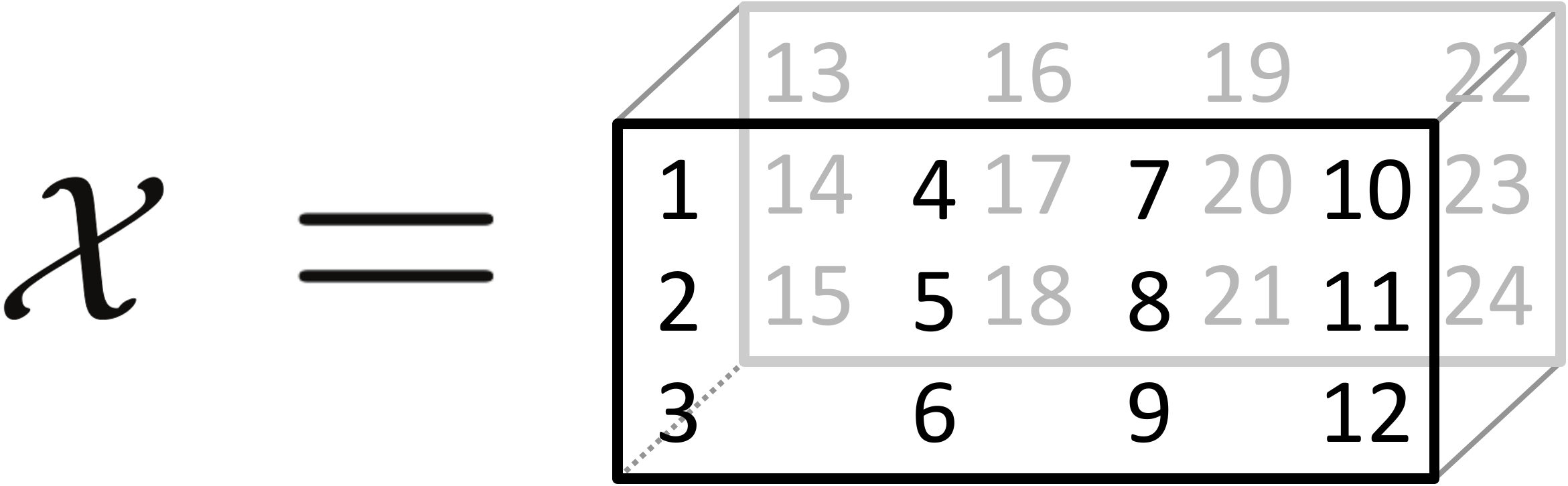}
  \centering
  \caption{A third order tensor $\mX\in \mathbb{R}^{3 \times 4 \times 2}$ with elements 1 through 24. }
  \label{fig:tensim}
\end{figure}
was also used in \cite{koldaandbader09}. 

The $k$th-mode matricizations are
$$ 
\mX_{(1)} = \begin{bmatrix}
1&4&7&10&13&16&19&22\\
2&5&8&11&14&17&20&23\\
3&6&9&12&15&18&21&24
\end{bmatrix},
\hspace{0.5cm}
\mX_{(2)} = \begin{bmatrix}
1&2&3&13&14&15\\
4&5&6&16&17&18\\
7&8&9&19&20&21\\
10&11&12&22&23&24
\end{bmatrix}
\hspace{0.1cm},$$
$$
\mX_{(3)} = 
\left[
\begin{array}{*{12}c}
1  & 2  & 3 & 4  & 5  & 6  & 7  & 8  & 9  & 10 & 11 & 12 \\
13 & 14 & 15 & 16 & 17 & 18 & 19 & 20 & 21 & 22 & 23 & 24  
\end{array}
\right]
$$
and the $k$ canonical matricizations are
$$ 
\mX_{<1>} = \mX_{(1)}
,\quad
\mX_{<2>} = \mX_{(3)}'
,\quad
\mX_{<3>} = \vecc(\mX) = [ 1 \quad 2\quad  \hdots\quad 24]'.
$$

The last-mode with first-mode contraction, denoted as $\times^1$ and
defined in Table~\ref{table:2} and
Equation~\eqref{def:modecontraction}, is used to define the TR format
in equation \eqref{TRform}. One special case of this contraction is
the matrix product between $X$ and $Y$, which contracts the columns of
$X$ with the rows of $Y$, and can be expressed as $X\times^1 Y$. In
Fig.~\ref{fig:tensim2} we further ilustrate this contraction, for
the $\mX$ of Figure \ref{fig:tensim}.
\begin{figure}[h]
\includegraphics[width = 0.5\textwidth]{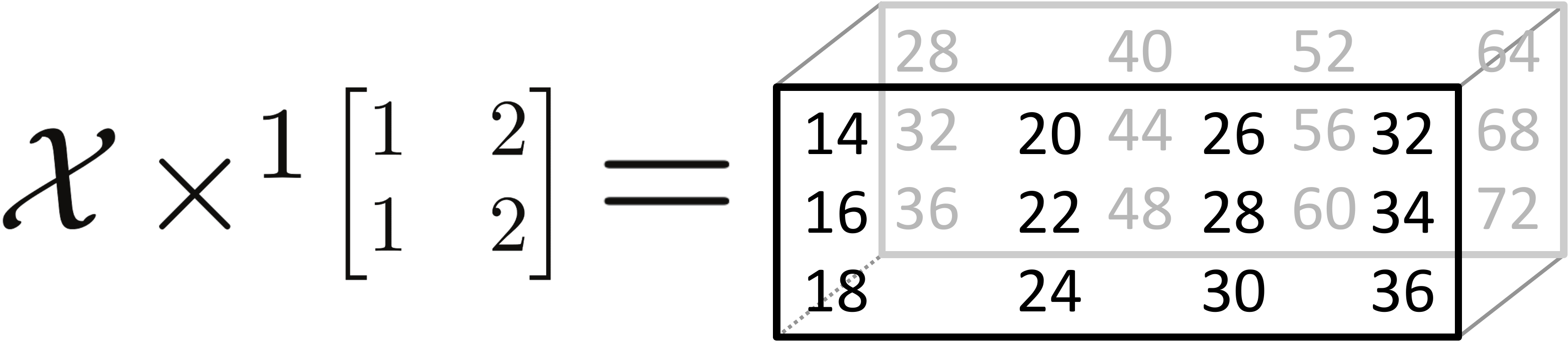}
  \centering
  \caption{The third order tensor that results from applying the
    last-mode with first-mode contraction between $\mX$ and $Y$, where
    $\mX$ is as in Figure \ref{fig:tensim}.}
  \label{fig:tensim2}
\end{figure}

We can reshape tensors by manipulating the vector outer product
($\circ$) applied to unit basis vectors, as in Equation~\eqref{tens}. If $\{\be_k^{m_k}\}_{k=1}^p$ are unit basis vectors, then $\bigotimes_{k=p}^1 \be_k^{m_k} $ is also so, and so
\begin{equation}\label{eq:unitbasis}
\Big( {\bigotimes_{k=p}^1 \be_k^{m_k}} \Big)' \Big( \bigotimes_{k=p}^1 \be_k^{m_k} \Big)=  \bigotimes_{k=p}^1({\be_k^{m_k}}' \be_k^{m_k}) = 1.
\end{equation}
Equation \eqref{eq:unitbasis} is helpful in simplifying matrix products between matricized tensors. 
\subsection{Proof of Lemma \ref{lemma1}}\label{proof:lemma1}
\begin{proof}\hfill
\begin{enumerate}[label=(\alph*)]
\item We have
 \begin{equation*}
    \begin{split}
\mX_{<p-1>} '
&=\sum\limits_{i_1=1}^{m_1} \hdots \sum\limits_{i_p=1}^{m_p} \mX(i_1,i_2,\hdots,i_p)
\Big(
\big( \bigotimes\limits_{q =p-1}^1 \be_{i_q}^{m_q} \big) \big(\be_{i_p}^{m_p}\big)'
\Big)'
\\&=
\sum\limits_{i_1=1}^{m_1} \hdots \sum\limits_{i_p=1}^{m_p} \mX(i_1,i_2,\hdots,i_p)
\Big( \big(\be_{i_p}^{m_p}\big)
\big( \bigotimes\limits_{q =p-1}^1 \be_{i_q}^{m_q} \big)'
\Big) = \mX_{(p)}.
\end{split}
\end{equation*}
\item For any $l=1,2,\ldots,p$, we have
  \begin{equation*}
    \begin{split}
\vecc(\mX_{<l>})
&=\sum\limits_{i_1=1}^{m_1} \hdots \sum\limits_{i_p=1}^{m_p} \mX(i_1,i_2,\hdots,i_p)
 \vecc\Big(\big( \bigotimes\limits_{q =l}^1 \be_{i_q}^{m_q} \big) 
 \big( \bigotimes\limits_{q =p}^{l+1} \be_{i_q}^{m_q} \big)' \Big)
\\&=
\sum\limits_{i_1=1}^{m_1} \hdots \sum\limits_{i_p=1}^{m_p} \mX(i_1,i_2,\hdots,i_p)
\big( \bigotimes\limits_{q =p}^1 \be_{i_q}^{m_q} \big)
=\vecc(\mX).
\end{split}
\end{equation*}
\item Using the definition of vectorization in Table~\ref{table:1}, we have
  \begin{equation*}
    \begin{split}
(\vecc{\mX})' (\vecc{\mY})
& =
\Big[\sum\limits_{i_1=1}^{m_1} \hdots \sum\limits_{i_p=1}^{m_p}        \mX(i_1,i_2,\hdots,i_p) 
\big( \bigotimes\limits_{q =p}^1 \be_{i_q}^{m_q} \big)'\Big]
\Big[\sum\limits_{i_1=1}^{m_1} \hdots \sum\limits_{i_p=1}^{m_p} \mY(i_1,i_2,\hdots,i_p)
\big( \bigotimes\limits_{q =p}^1 \be_{i_q}^{m_q} \big)\Big]
\\& =
\sum\limits_{i_1=1}^{m_1} \hdots \sum\limits_{i_p=1}^{m_p}
\mX(i_1,i_2,\hdots,i_p)\mY(i_1,i_2,\hdots,i_p) \big(
\bigotimes\limits_{q =p}^1 \be_{i_q}^{m_q} \big)'\big(
\bigotimes\limits_{q =p}^1 \be_{i_q}^{m_q} \big) 
\\&=
\sum\limits_{i_1=1}^{m_1} \hdots \sum\limits_{i_p=1}^{m_p} \mX(i_1,i_2,\hdots,i_p)\mY(i_1,i_2,\hdots,i_p)= \langle \mX , \mY \rangle .
\end{split}
\end{equation*}
Similarly if $\bi = [i_1,\dots,i_p]'$ then
  \begin{equation*}
   \begin{split}
\tr(\mX_{(k)}\mY_{(k)}')
&=
\tr \Bigg\{ \Big[\sum\limits_{i_1=1}^{m_1} \hdots \sum\limits_{i_p=1}^{m_p} \mX(\bi)
(\be_{i_k}^{m_k}) \big( \bigotimes\limits_{\overset{q=p}{q\neq k}}^1 \be_{i_q}^{m_q} \big)'\Big]
\Big[\sum\limits_{i_1=1}^{m_1} \hdots \sum\limits_{i_p=1}^{m_p} \mY(\bi)
 \big( \bigotimes\limits_{\overset{q=p}{q\neq k}}^1 \be_{i_q}^{m_q} \big)\big(\be_{i_k}^{m_k}\big)'\Big]\Bigg\}
\\&=
\sum\limits_{i_1=1}^{m_1} \hdots \sum\limits_{i_p=1}^{m_p} \mX(i_1,i_2,\hdots,i_p)\mY(i_1,i_2,\hdots,i_p)
\tr \Bigg\{\be_{i_k}^{m_k}\big( \bigotimes\limits_{\overset{q=p}{q\neq k}}^1 \be_{i_q}^{m_q} \big)'\big( \bigotimes\limits_{\overset{q=p}{q\neq k}}^1 \be_{i_q}^{m_q} \big) \big(\be_{i_k}^{m_k}\big)'\Bigg\}
\\&=
\sum\limits_{i_1=1}^{m_1} \hdots \sum\limits_{i_p=1}^{m_p} \mX(i_1,i_2,\hdots,i_p)\mY(i_1,i_2,\hdots,i_p) = \langle \mX , \mY \rangle .
\end{split} 
\end{equation*}
\item   
\begin{equation*}
    \begin{split}
\vecc[\![ \mX; A_1, \hdots,A_p  ]\!] 
&=
\sum\limits_{i_1=1}^{m_1} \hdots \sum\limits_{i_p=1}^{m_p} \mX(i_1,i_2,\hdots,i_p)
 \vecc\Big( \circs\limits_{q=1}^p A_q(:,i_q) \Big)
\\&=
\sum\limits_{i_1=1}^{m_1} \hdots \sum\limits_{i_p=1}^{m_p} \mX(i_1,i_2,\hdots,i_p)
\Big( \bigotimes\limits_{q=p}^1 A_q(:,i_q) \Big)
=
 \big(\bigotimes\limits_{q=p}^1 A_q\big) \vecc( \mX).
\end{split}
\end{equation*}
\item We have for $\bi = [i_1,\dots,i_p]'$ and $\bj = [j_1,\dots,j_q]'$ that
  \begin{equation*}
    \begin{split}
\mB_{<p>}' \vecc{ \mX}
&=
 \Big[\sum\limits_{i_1=1}^{m_1} \hdots \sum\limits_{i_p=1}^{m_p}\sum\limits_{j_1=1}^{h_1} \hdots \sum\limits_{j_q=1}^{h_q}
\mB([\bi'\bj']')
\Big(
\big( \bigotimes\limits_{k =q}^1 \be_{j_k}^{h_k} \big)
\big( \bigotimes\limits_{k =p}^1 \be_{i_k}^{m_k} \big) '
\Big)\Big] \Big[
\sum\limits_{i_1=1}^{m_1} \hdots \sum\limits_{i_p=1}^{m_p} \mX(\bi)
\big( \bigotimes\limits_{k =p}^1 \be_{i_k}^{m_k} \big)\Big]
\\&=
\sum\limits_{i_1=1}^{m_1} \hdots \sum\limits_{i_p=1}^{m_p}\sum\limits_{j_1=1}^{h_1} \hdots \sum\limits_{j_q=1}^{h_q}
\mB([\bi'\bj']')\mX(\bi)
\big( \bigotimes\limits_{k =q}^{1} \be_{j_k}^{h_k}\big)\big( \bigotimes\limits_{k =p}^1 \be_{i_k}^{m_k} \big) '\big( \bigotimes\limits_{k =p}^1 \be_{i_k}^{m_k} \big) 
\\&=
\sum\limits_{i_1=1}^{m_1} \hdots \sum\limits_{i_p=1}^{m_p}\sum\limits_{j_1=1}^{h_1} \hdots \sum\limits_{j_q=1}^{h_q}
\mB([\bi'\bj']')\mX(\bi)
\big( \bigotimes\limits_{k =q}^{1} \be_{j_k}^{h_k}\big)
=\vecc\langle \mX | \mB \rangle.
\end{split}
\end{equation*}
\item Because of Lemma~\ref{lemma1}\ref{lemma1:z} and Lemma~\ref{lemma1}\ref{lemma1:y} and the commutation matrix, if $\mY \in \mathbb{R}^{m_1\times m_2\times\hdots\times m_k}$, then
\begin{equation}\label{eq:subtens_vec}
\vecc (\mY)  = \vecc (\mY_{<k-1>}) = \vecc(\mY_{(k)}') = K_{m_k,\prod_{i=1}^{k-1}m_i}\vecc(\mY_{(k)}).
\end{equation} Therefore, if we let ${\mX}^{(j_{k+1},\hdots,j_{p})}\in \mathbb{R}^{m_1\times m_2\times\hdots\times m_k}$ be the sub-tensor of $\mX$  defined as 
$
{\mX}^{(j_{k+1},\hdots,j_{p})} = {\mX}{(:,\dots,:,j_{k+1},\hdots,j_{p})}
$, then from~\eqref{eq:subtens_vec},
  \begin{equation*}
    \begin{split}
\vecc{(\mX)}
&=
\begin{bmatrix}
\vecc{{\mX}^{(1,\hdots,1,1)}}\\
\vecc{{\mX}^{(1,\hdots,1,2)}}\\
\vdots\\
\vecc{\mX^{ (m_{k+1},\hdots,m_{p})}}
\end{bmatrix} 
=
\begin{bmatrix}
K_{m_k,\prod_{i=1}^{k-1}m_i}\vecc{(\mX^{(1,\hdots,1,1)}_{(k)})}\\
K_{m_k,\prod_{i=1}^{k-1}m_i}\vecc{(\mX^{(1,\hdots,1,2)}_{(k)})}\\
\vdots\\
K_{m_k,\prod_{i=1}^{k-1}m_i}\vecc{(\mX^{ (m_{k+1},\hdots,m_{p})}_{(k)})}
\end{bmatrix}
\\&=
K_{(k)}'
\vecc(\mX_{(k)}).
\end{split}
\end{equation*}
The result follows by left-multiplying both sides by the orthogonal
matrix $K_{(k)}$.\vspace*{-.5cm} 
\end{enumerate}
\end{proof}
\subsection{Proof of Theorem \ref{thm:mat_outer_product}} \label{proof:mat_outer_product}
\begin{proof}\hfill
\begin{enumerate}[label=(\alph*)]
\item We have based on property \ref{prop:Kron}
\begin{equation*}
    \begin{split}
(\circ [\![ M_1,M_2 \hdots,M_p  ]\!])_{<p>}
&= \sum\limits_{i_1=1}^{m_1} \sum\limits_{j_1=1}^{h_1} 
 \hdots 
 \sum\limits_{i_p=1}^{m_p}  \sum\limits_{j_p=1}^{h_p}
\big(\prod\limits_{q=1}^p M_q(i_q,j_q)\big)
\Big\{
\big( \circs\limits_{q=1}^p \be_{j_q}^{h_q} \big) \circ
\big( \circs\limits_{q=1}^p \be_{i_q}^{m_q} \big)
\Big\}_{<p>}
\\&=
 \sum\limits_{i_1=1}^{m_1} \sum\limits_{j_1=1}^{h_1} 
 \hdots 
 \sum\limits_{i_p=1}^{m_p}  \sum\limits_{j_p=1}^{h_p}
\big(\prod\limits_{q=1}^p M_q(i_q,j_q)\big)
\Big\{
\big( \bigotimes\limits_{q=p}^1 \be_{j_q}^{h_q} \big) 
\big( \bigotimes\limits_{q=p}^1 \be_{i_q}^{m_q} \big)'
\Big\}
\\&=
\bigotimes\limits_{q=p}^1 \left( \sum\limits_{j_q=1}^{h_q}  \sum\limits_{i_q=1}^{m_q}
M_q(i_q,j_q) \be_{j_q}^{h_q} {\be_{i_q}^{m_q}}' 
\right)
=\bigotimes\limits_{q=p}^1 M_q'.\hspace{7.9cm}
\end{split}
\end{equation*}
\item Because of Theorem~\ref{thm:mat_outer_product}\ref{thm:mat_outer_product:z} and Lemma~\ref{lemma1}\ref{lemma1:v},
  \begin{equation*}
    \begin{split}
\vecc\langle \mX | \circ [\![ M_1,M_2 \hdots,M_p  ]\!] \rangle 
&= (\circ [\![ M_1,M_2 \hdots,M_p  ]\!])'_{<p>}\vecc(\mX)
=
\big(\bigotimes\limits_{q=p}^1 M_q \big) \vecc(\mX) 
= \vecc[\![ \mX; M_1, \hdots,M_p  ]\!].
\end{split}
\end{equation*}
Therefore $[\![ \mX; M_1, \hdots,M_p  ]\!]$ and $\langle \mX | \circ [\![ M_1,M_2 \hdots,M_p  ]\!] \rangle$ are of the same size and vectorization, and must be the same.
\item Using definition \ref{def:matdef} and  property \ref{prop:Kron}
\begin{equation*}
    \begin{split}
\circ [\![ M_1,M_2 \hdots,M_p  ]\!]_{(\mathscr{S}\times \mathscr{S}^\mathsf{c})} 
&=
\sum\limits_{i_1=1}^{m_1} \sum\limits_{j_1=1}^{h_1} 
 \hdots 
 \sum\limits_{i_p=1}^{m_p}  \sum\limits_{j_p=1}^{h_p}
\big(\prod\limits_{q=1}^p M_q(i_q,j_q)\big)
\Big\{
\big(  
\be_{j_k}^{h_k} 
\otimes 
\be_{i_k}^{m_k} 
\big) 
\big( 
\bigotimes\limits_{q=p,q\neq k}^1 \be_{j_q}^{h_q} 
\otimes
\bigotimes\limits_{q=p,q\neq k}^1 \be_{i_q}^{m_q} 
\big)'
\Big\}
\\&=
(\vecc M_k')(\vecc \circ [\![ M_1,M_2 \hdots,M_{k-1},M_{k+1},\hdots,M_p  ]\!])'.
\end{split}
\end{equation*}
\end{enumerate}
\end{proof}

\subsection{Some matrices involved in our methodology}\label{app:closedforms}

In this Section we provide and simplify some expressions involved in Sections \ref{subsec:parest} and \ref{subsec:sampdist}. If $\mV$ is a tensor of size $r\times r\times\dots\times r$, then its super-diagonal $\boldsymbol{v} \in \mathbb{R}^r$ contains the elements with coincident indices in each dimension of $\mV$, such that
\begin{equation}\label{eq:super-diagonal}
\boldsymbol{v}(i) = \mV(i,i,\dots,i).
\end{equation}
We assume vectors are their own super-diagonals.
Now we provide simplifications for the expressions involved in the CP format using the Hadamard (or element-wise) matrix-product $(\Ast)$, which is computationally more efficient than the traditional matrix product.

\subsubsection{The CP format}\label{supp:CP}

\textbf{Simplifying Equation \eqref{eq:CPM}}:
First, let $\bw_i$ be the super-diagonal of $[\![\mX_i;L_1',L_2'\hdots,L_l']\!]$, as defined in equation \eqref{eq:super-diagonal}, and let $W_i$ be the $r\times r$ diagonal matrix with diagonal elements in $\bw_i$. Then the matrix $G_{ik}^{CP}$ in \eqref{eq:reg_CPkth} can be written as
\begin{equation}\label{supp:CPsimp1}
G_{ik}^{CP}= W_i
\Big(
\bigodot\limits_{q=p,q\neq k}^1 M_q
\Big)',
\end{equation}
where $\odot$ is the Khatri-Rao matrix product of \citep{petersenandpedersen12}.  This formulation allows us to simplify the matrix 
$
\sum_{i=1}^n G_{ik}^{CP}
\Sigma_{-k}^{-1}   {G_{ik}^{CP}}'$ from Equation~\eqref{eq:CPM} using Proposition 3.2 from \citet{kolda06} as
\begin{equation}\label{supp:CPsimp2}
\sum\limits_{i=1}^n G_{ik}^{CP}
\Sigma_{-k}^{-1}   {G_{ik}^{CP}}'
=
\Big[
\sum\limits_{i=1}^n \bw_i\bw_i'
\Big]
\Ast
\Big[
\Ast\limits_{q=1,q\neq k}^{p} 
(M_k'\Sigma_k^{-1}M_k)
\Big].
\end{equation}
Further, we can simplify equation \eqref{eq:CPM} based on equation \eqref{supp:CPsimp2} and the following identity
 $$ {\bmY_i}_{(k)} \Sigma_{-k}^{-1}
 {G_{ik}^{CP}}'=
 [\by_{1i}\dots\by_{m_ki}]'
W_i,
 $$
 where $\by_{si}$ is the super-diagonal of $
[\![\bmY_{si};M_1'\Sigma_{1}^{-1},\dots,M_{k-1}'\Sigma_{k-1}^{-1},M_{k+1}'\Sigma_{k+1}^{-1},\dots,M_{p}'\Sigma_{p}^{-1}]\!],
$
and $\bmY_{si}$ is the tensor of size $m_1\times\dots\times m_{k-1}\times m_{k+1}\times\dots\times m_p$ obtained upon fixing the $k$-th mode of $\bmY_i$ at the position $s$. After having estimated $\widehat{M}_k$ in equation \eqref{eq:CPM}, we can also simplify $S_k$ as 
$$
S_k = \sum\limits_{i=1}^n \bmY_{i(k)}\Sigma_{-k}^{-1}\bmY_{i(k)}' 
-
\widehat{M}_k 
\left(\sum\limits_{i=1}^n G_{ik}^{CP}
\Sigma_{-k}^{-1}   {G_{ik}^{CP}}'\right)
\widehat{M}_k',
$$
where $\sum\limits_{i=1}^n G_{ik}^{CP}
\Sigma_{-k}^{-1}   {G_{ik}^{CP}}'$ is already simplified in equation \eqref{supp:CPsimp2}.

\textbf{Simplifying Equation \eqref{eq:CPL}}:
The matrix
$
\sum\limits_{i=1}^n 
 H_{ik}^{CP} \Sigma^{-1}  {H_{ik}^{CP}}' 
$
from equation \eqref{eq:CPL} can be written as $\mH^2_{<2>}$, which is the $2$-canonical matricization of the tensor $\mH^2$ of size $h_k\times r\times h_k\times r$, with sub-tensors  
$$
\mH^2(s,:,t,:) = (\Ast\limits_{k=1}^pM_k'\Sigma_k^{-1}M_k) * (\sum_i^n \bw_{si}\bw_{ti}')
\quad \text{ for all }\quad
s,t = 1,2,\dots,h_k, 
$$
where $\bw_{si}\in \mathbb{R}^r$ is the super-diagonal of 
$
[\![\mX_{si};L_1',\dots,L_{k-1}',L_{k+1}',\dots,L_l]\!],
$
and $\mX_{si}$ is the tensor of size $h_1\times\dots\times h_{k-1}\times h_{k+1}\times\dots\times h_l$ obtained upon fixing the $k$-th mode of $\mX_i$ at the position $s$. 
Finally, equation \eqref{eq:CPL} can be greatly simplified using the above simplification along with the following identity
$$
H_{ik}^{CP}\Sigma^{-1} \vecc{(\bmY_i )}
=
[(\bw_{1i}*\by_i)'\dots (\bw_{h_ki}*\by_i)']',
$$
where $\by_{i}$ is the super-diagonal of $
[\![\bmY_{i};M_1'\Sigma_{1}^{-1},\dots,M_{p}'\Sigma_{p}^{-1}]\!]$.
\subsubsection{The TR format}\label{supp:TR}

To define the matrices $G_{ik}^{TR}$ for $k=1,2,\dots,l$, let
\begin{equation}\label{supp:TRnoM}
\mB_{-\mM_k}
=
\mM_{k+1}\times^1\hdots \times^1 \mM_{p}
\times^1 
\mL_{1}\times^1\hdots \times^1 \mL_{l}
\times^1
\mM_{1}\times^1\hdots \times^1 \mM_{k-1},
\end{equation}
be the $(p+1)$th order TT-formatted tensor of dimension $(g_k\times m_{k+1}\times \dots\times m_{p}\times h_1\times \dots \times h_l\times m_1\times \dots\times m_{k-1}\times g_{k-1})$ that is generated similar to the TR tensor in Equation \eqref{eq:reg_TR} , but is missing the TR factor $\mM_{k}$.
The tensor $\mB_{-\mM_{k}}$ in Equation~\eqref{supp:TRnoM} is defined for the cases where $k=2,3,\dots,p-1$, however $\mB_{-\mM_1}$ and $\mB_{-\mM_{p}}$ are similar: the tensor $\mB_{-\mM_p}$ has dimensions $(g_p\times h_1\times\dots\times h_l\times m_1\times \dots\times m_{p-1}\times g_{p-1})$ because $\mL_1$ is the first factor in the train, and the tensor $\mB_{-\mM_1}$ has dimensions $(g_1\times m_2\times \dots\times m_{p}\times h_1\times\dots\times h_l\times g_0)$ because $\mL_l$ is the last factor in the train. We define the matrix $G_{ik}^{TR}$ involved in Equation~\eqref{eq:reg_TRkth} as
$$
G_{ik}^{TR}= \Big\{
\mX_i
\times_{1,2,\dots,l}^{p-k+2,p-k+3,\dots, p-k+l+1}
\mB_{-\mM_k}
\Big\}_{(\mathscr{R}\times \mathscr{C})},
$$
where $\mathscr{R} = \{1,p+1\}$, 
and $
\mathscr{C}$ is the ordered set of size $p-1$ that indexes the modes of size $(m_1,m_2,\dots,m_p)$, but skipping $m_k$ and in reverse order (so that the order of the modes matches with the columns of $\bmY_{i(k)}$ in Equation~\eqref{eq:reg_TRkth}).
Similarly, to define the matrices $H_{ik}^{TR}$ for $k=1,2,\dots,l$, let
\begin{equation}\label{supp:TRnoL}
\mB_{-\mL_k}
=
\mL_{k+1}\times^1\hdots \times^1 \mL_l
\times^1 
\mM_1\times^1\hdots \times^1 \mM_p
\times^1
\mL_1\times^1\hdots \times^1 \mL_{k-1},
\end{equation}
be the $(p+1)$-th order TT-formatted tensor of dimensions $(s_k\times h_{k+1}\times \dots\times h_{l}\times m_1\times \dots \times m_p\times h_1\times \dots\times h_{k-1}\times s_{k-1})$ that is generated similar to the TR tensor in Equation~\eqref{eq:reg_TR} , but is missing the TR factor $\mL_k$.
The tensor $\mB_{-\mL_k}$ in~\eqref{supp:TRnoL} is defined for the cases where $k=2,3,\dots,l-1$, however $\mB_{-\mL_1}$ and $\mB_{-\mL_l}$ are similar: the tensor $\mB_{-\mL_l}$ has dimensions $(s_l\times m_1\times\dots\times m_p\times h_1\times \dots\times h_{l-1}\times s_{l-1})$ because $\mM_1$ is the first factor in the train, and the tensor $\mB_{-\mL_1}$ has dimensions $(s_1\times h_2\times \dots\times h_{l}\times m_1\times\dots\times m_p\times s_0)$ because $\mM_p$ is the last factor in the train. We define the matrix $H_{ik}^{TR}$ involved in Equation~\eqref{eq:reg_TRvec} as
\begin{equation}\label{supp:TR_H}
H_{ik}^{TR}= \Big\{
\mX_i
\times_{1,\dots,k-1,k+1,\dots,l}^{l+p+2-k,\hdots,l+p,2,\hdots,l-k+1}
\mB_{-\mL^{(k)}}
\Big\}_{(\mathscr{R}\times \mathscr{C})},
\end{equation}
where $\mathscr{R} = \{2,1,p+3\}$, 
and $
\mathscr{C} = \{p+2,p+1,\dots,3\}$ is the ordered set of size $p-1$ that indexes the modes of size $(m_1,m_2,\dots,m_p)$. Note that we define the contraction in Equation~\eqref{supp:TR_H} as $(\times_{2,\dots,l}^{2,\hdots,l})$ when $k=1$ and as $(\times_{1,\dots,l-1}^{p+2,\hdots,l+p})$ when $k=l$.

\subsection{Rank determination and number of parameters}\label{supp:tpars}
In Section \ref{subsec:rank_choos} we propose choosing ranks using the BIC, defined as
$$
BIC= K\log(n) - 2\ell,
$$
where $n$ is the sample size and $\ell$ is the loglikelihood that we
simplified in Section~\ref{est_likelihood}. We decompose the total
number of parameters $K$ as $K=K_\Sigma+K_\mB$, where $K_\Sigma$ is
the number of parameters involved in the covariance matrix and $K_\mB$
is the total of parameters involved in the low-rank format of
$\mB$. The value of $K_\Sigma$  depends on the correlation structure
that is imposed on all the scale matrices
$\Sigma_1,\dots,\Sigma_p$. For the general unconstrained model of
\eqref{eq:model_noM}, we have 
$$
K_\Sigma=1+\sum_{k=1}^p \left(
\dfrac{m_k(m_k+1)}{2} -1
\right)
$$ 
to account for $\sigma^2$ as well as for the constraint that $\Sigma_k(1,1) = 1$ for all $k=1,2,\dots,p$.
To determine $K_\mB$ for the constrained cases, we need to consider each low-rank format separately:
\begin{itemize}
\item For $\mB_{TK}$ as in equation \eqref{eq:reg_Tuck} we have
$$
K_\mB = \prod_{j=1}^l c_k\prod_{j=1}^p d_k + 
\sum_{k=1}^l \left(
h_kc_k - \dfrac{c_k(c_k+1)}{2}
\right)+
\sum_{k=1}^p \left(
m_kd_k - \dfrac{d_k(d_k+1)}{2}
\right)
$$
to account for the parameters involved in $\mV$ as well as those involved in each $L_1,\dots,L_l,M_1,\dots,M_p$ and their orthogonality constraints.

\item For $\mB_{CP}$ as in equation \eqref{eq:reg_CP} we have
$$
K_\mB = R\times 
\left(
\sum_{k=1}^l 
h_k
+
\sum_{k=1}^p 
m_k
-l-p+1
\right)
$$
to account for each $L_1,\dots,L_l,M_1,\dots,M_p$ and their column-scale constraints.
\item For $\mB_{TR}$ as in equation \eqref{eq:reg_TR} we have
$$
K_\mB =  
\sum_{k=1}^l s_{k-1}h_ks_k
+
\sum_{k=1}^p g_{k-1}m_kg_k
-l-p+1
$$
to account for each $\mL_1,\dots,\mL_l,\mM_1,\dots,\mM_p$ and their scale indeterminacy.

\item For $\mB_{OP}$ as in equation \eqref{eq:reg_OP} we have
$$
K_\mB =  
\sum_{k=1}^p h_km_k
-p+1
$$
to account for each $M_1,\dots,M_p$ and their scale indeterminacy.
\end{itemize}

\subsection{Computational Complexity and Proof of Theorem~\ref{theo:comp}}\label{Ssubsec:computat}

\subsubsection{Computational complexity for the Tucker format}\label{Ssec:computTK}

\begin{proof}
The complexity is expressed in terms of four terms, which corresponds to lines 4, 6, 8 and 11 in Algorithm \ref{alg:1}, respectively. We now derive the complexity of each of these lines:

\textbf{Line 4:} The computational complexity of line 4 corresponds to implementing equation \eqref{eq:Tuckerest_M} by obtaining:
\begin{enumerate}[4a)]
\item 
$
[\![\bmY_i ; I_{m_1}, M_2'\Sigma_2^{-1},\dots, M_p'\Sigma_p^{-1}]\!]
$
for all $i=1,2,\dots,n$, which will be dominated later by 6a)
\item 
$
\mW_i=[\![\mX_i ; L_1', L_2',\dots, L_p']\!]
$
for all $i=1,2,\dots,n$, which has
$\mO(nhc_1)$ complexity.
\item $W'(WW')^{-1}W$ where $W=[\vecc(\mW_1) ,\cdots,\vecc(\mW_n)]$, with $\mO(n^2c)$ complexity.
\item  $Q_1$ of eq. (24), the product $\Sigma_1^{-1/2}Q_1$ and the truncated singular value decomposition of $\Sigma_1^{-1/2}Q_1$, which have $\mO(n^2m_1d_{-1}+nm_1^2d_{-1})$ complexity.
\end{enumerate}

\textbf{Line 6:} The computational complexity of line 6 corresponds to implementing equation \eqref{eq:Tuckerest_V} by obtaining:
\begin{enumerate}[6a)]
\item 
$\bmZ_i=[\![\bmY_i ;  M_1'\Sigma_2^{-1}, M_2'\Sigma_2^{-1},\dots, M_p'\Sigma_p^{-1}]\!]
$
for all $i=1,2,\dots,n$, which has
$\mO(nmd_1)$ complexity.
\item 
$W^{-}$, which is the generalized inverse of $W$. This computation is dominated by 4c).
\item $[\vecc(\bmZ_1) ,\cdots,\vecc(\bmZ_n)]W^{-'}$, which has $\mO(cnd)$ complexity and will be dominated by 8b).
\end{enumerate}

\textbf{Line 8:} The computational complexity of line 8 corresponds to implementing equation \eqref{eq:Tuckerest_L} by obtaining:
\begin{enumerate}[8a)]
\item $\mG_i=[\![\mX_i ; I_{h_1}, L_2',\dots, L_l']\!]$, which is dominated by 4b).
\item 
$
\mH_i=\mG_i \times_{2,\cdots,l}^{2,\cdots,l} \mV
$ for all $i=1,2,\dots,n$, which has
$\mO(ncdh_1)$ complexity. 
\item 
$
\mH_{i<2>}\mH_{i<2>}'
$ and $\mH_{i<2>}\vecc(\bmZ_i)$ for all $i=1,2,\dots,n$ to form matrices in equation \eqref{eq:Tuckerest_L}, with
$\mO(nc_1^2h_1^2d)$ complexity.
\item inverse and final product in equation \eqref{eq:Tuckerest_L}, with complexity $\mO(h_1^3c_1^3)$.
\end{enumerate}

\textbf{Line 11:} The computational complexity of line 11 corresponds to implementing equation \eqref{eq:Tuckerest_S} by obtaining:
\begin{enumerate}[$\eleven$a)]
\item inverse of $\Sigma_1,\Sigma_2,\dots,\Sigma_p$, with complexity $\mO(m_1^3)$.
\item $\langle \mW_i | \mV \rangle$ for all $i=1,2,\dots,n$ with complexity $\mO(ndc_1h_1)$ given $\mH_i$ of 8b), and is dominated by 8b).
\item $[\![\langle \mW_i | \mV \rangle;M_1,M_2,\dots,M_p]\!]$ for all $i=1,2,\dots,n$, which is dominated by 6a).
\item $\sum_{i=1}^n \bmY_{i(k)}\Sigma_{-k}^{-1}\bmY_{i(k)}'$ with complexity $\mO(nmm_1)$.
\end{enumerate}

\end{proof}

\subsubsection{Computational complexity for the CP format}\label{Ssec:computCP}

\begin{proof}
The complexity is expressed in terms of two terms, which corresponds to lines 4 and 7 in Algorithm \ref{alg:CP}, respectively. We now derive the complexity of each of these lines:

\textbf{Line 4:} The computational complexity of line 4 corresponds to implementing equation \eqref{eq:CPL}, and based on the simplifications detailed in Section \ref{supp:CP}. This involves obtaining:
\begin{enumerate}[4a)]
\item the vectors of super-diagonals $\bw_{si}$, which are computationally dominated by 7b).
\item $\sum_{i=1}^n\bw_{si}\bw_{ti}'$ for all $s,t=1,2,\dots,h_1$, with complexity $\mO(nh_1^2r^2).$
\item 
$\mM_q = \circ_{k=1}^{p}M_k[:,q]$
for all $q=1,2,\dots,r$, with complexity $\mO(mr)$. This is dominated by 4d).
\item $\by_{i}(q) =\langle\bmY_{i},\mM_q\rangle $ for all $i=1,2,\dots,n$, and $q=1,2,\dots,r$, with complexity $\mO(nrm)$. 
\item $\bw_{si}*\by_i'$ for all $s=1,2,\dots,h_1$ and $i=1,2,\dots,n$, with complexity $\mO(nrh_1).$ This is dominated by 4b).
\item $*_{k=1}^pM_k'\Sigma_k^{-1}M_k$ with computational complexity $\mO(rm_1^2+m_1^3)$
\item The final inversion and matrix-vector product of equation \eqref{eq:CPL} with complexity $\mO(h_1^3r^3)$.

\end{enumerate}

\textbf{Line 7:} The computational complexity of line 7 corresponds to implementing equation \eqref{eq:CPM}, which is simplified in Section \ref{supp:CP}. This involves obtaining:
\begin{enumerate}[7a)]
\item 
$\mL_q = \circ_{k=1}^{l}L_k[:,q]$
for all $q=1,2,\dots,r$, with complexity $\mO(hr)$. This is dominated by 7b).
\item $\bw_{i}(q) =\langle\mX_{i},\mL_q\rangle $ for all $i=1,2,\dots,n$ and $q=1,2,\dots,r$, with complexity $\mO(nrh)$. 
\item the vectors of super-diagonals $\by_{si}$, which are computationally dominated by 4d).
\item The inverse and matrix product of equation \eqref{eq:CPM} has complexity $\mO(m_1r^2+r^3)$, and $\mO(r^3)$ is dominated by 4g). 
\item $\sum_{i=1}^n \bmY_{i(k)}\Sigma_{-k}^{-1}\bmY_{i(k)}'$ with complexity $\mO(nmm_1)$.
\end{enumerate}
\end{proof}

\subsubsection{Computational complexity for the OP format}\label{Ssec:computOP}
Unlike the CP and Tucker cases, the complexity of the OP format depends on the order of the tensor response and covariate modes. Therefore, finding the exact complexity will depend on many special cases unless a restriction is imposed. Here we will assume that for all $k=1,2,\dots,p$, $m_k=h_k$. This means that all the matrix factors $M_k$ are square. This way, the order of the modes will not matter anymore and we can assume without generality that $m_1$ is the size of the largest mode.
\begin{thm}
The computational complexity of implementing the OP format is 
$$O\big(
nmm_1 + m_1^3
\big).$$
\end{thm}
\begin{proof}
The computational complexity of line 4 corresponds to implementing equation \eqref{eq:CPM}, but where $G_{ik}^{CP}$ is replaced with the $G_{ik}^{OP}$ of Section \ref{sec:estOP}.

\begin{enumerate}[4a)]
\item Computing $G_{ik}^{OP}$ for each $k=1,2,\dots,p$ and for all $i=1,2,\dots,n$ with $\mO(nmm_1)$ complexity.
\item $\sum_{i=1}^n {\bmY_i}_{(k)} \Sigma_{-1}^{-1}
 {G_{ik}^{OP}}'$, $\sum_{i=1}^n G_{ik}^{OP}
\Sigma_{-k}^{-1}   {G_{ik}^{OP}}'$ and $\sum_{i=1}^n \bmY_{i(k)}\Sigma_{-k}^{-1}\bmY_{i(k)}'$ all have $\mO(nmm_1)$ complexity.
\item The inverse and matrix product of equation \eqref{eq:CPM} has complexity $\mO(m_1^3)$.
\end{enumerate}

\end{proof}

\subsubsection{Computational complexity for the TR format}\label{Ssec:computTR}
\begin{proof}
The complexity is expressed in terms of two terms, which corresponds to lines 4 and 7 in Algorithm \ref{alg:CP}, respectively. We now derive the complexity of each of these lines:

\textbf{Line 4:} The computational complexity of line 4 corresponds to implementing equation \eqref{eq:CPM}, which is implemented in equation \eqref{eq:CPL} when the $H_{ik}^{TR}$ of Section \ref{supp:TR} is used instead of $H_{ik}^{CP}$. This involves obtaining:
\begin{enumerate}[4a)]
\item $\mM=\mM_1\times^1\mM_2\times\dots\times^1\mM_p$ with complexity $\mO(mg_{1}g_0g_p)$.
\item $\mL_{-1} = \mL_2\times^1\mL_2\times\dots\times^1\mL_{l}$, which is dominated by 7a).
\item $\mD_i=\mL_{-1}\times_{2,\dots,l}^{2,\dots,l}\mX_i$ for all $i=1,2,\dots,n$ with complexity $\mO(hns_1s_{l})$, dominated by $\mO(hnss_{l})$.
\item $\mM \times^1 \mD_i$ for all $i=1,2,\dots,n$ with complexity $\mO(nmh_1s_ls_0s_1)$. This is dominated by 4e).
\item The final inversion and matrix-vector product of equation \eqref{eq:CPL} with complexity $\mO(nmh_1^2s_0^2s_1^2+h_1^3s_0^3s_1^3)$.
\end{enumerate}

\textbf{Line 7:} The computational complexity of line 4 corresponds to solving equation (37), which is implemented in equation \eqref{eq:CPL} when the $G_{ik}^{TR}$ of Section \ref{supp:TR} is used instead of $G_{ik}^{CP}$. This involves obtaining:
\begin{enumerate}[7a)]
\item $\mL=\mL_1\times^1\mL_2\times\dots\times^1\mL_l$ with complexity $\mO(hs_{1}s_0s_l)$.
\item $H_i=\mL\times_{2,\dots,h+1}^{1,\dots,h}\mX_i$ for all $i=1,2,\dots,n$ with complexity $\mO(hns_0s_l)$, dominated by 4c).
\item $\mM_{-1} = \mM_2\times^1\mM_3\times\dots\times^1\mM_{p}$, which is dominated by 4a).
\item $\mM_{-1}\times^1H_i$ for all $i=1,2,\dots,n$ with complexity $\mO(nm_{-1}g_0g_{1}g_p)$. This is dominated by 7e).
\item $\sum_{i=1}^n {\bmY_i}_{(k)} \Sigma_{-k}^{-1}
 {G_{ik}^{TR}}'$ and $\sum_{i=1}^n G_{ik}^{TR}
\Sigma_{-k}^{-1}   {G_{ik}^{TR}}'$, with complexity $O\left(nm_{-1}g_0g_1(g_0g_1+m_1)+m_1^3\right)$. The first term in the sum is dominated by 7g).
\item The final inversion and matrix-vector product of equation \eqref{eq:CPM} with complexity $\mO(g_0^3g_1^3)$.
\item $\sum_{i=1}^n \bmY_{i(k)}\Sigma_{-k}^{-1}\bmY_{i(k)}'$ with complexity $\mO(nmm_1)$.
\end{enumerate}

\end{proof}

\subsection{Proof of Theorem \ref{thm:inference_Tucker}}\label{proof:inference_Tucker}
\begin{proof}
Based on Lemma~\ref{lemma1}\ref{lemma1:w} and Equation~\eqref{eq:Tuckerest_V}, we obtain
\begin{equation}\label{eq:V}
\vecc(\mhV) = \big[
M'\Sigma^{-1}\otimes W^{-'}
\big]\vecc(Y'),
\end{equation}
where 
\begin{equation}\label{eq:Yt}
\vecc(Y')\sim \N_{n\times m}\big( (M\otimes W')\vecc(\mV),\sigma^2\Sigma\otimes \I_n\big).
\end{equation}
Furthermore, because $W^- = X'(XX')^{-1}L(L'L)^{-1}$ and both $W^{-'}W'$ and $M'\Sigma^{-1}M$ are identity matrices, we obtain based on \eqref{eq:V} and \eqref{eq:Yt} that
\begin{equation}\label{eq:step1}
\vecc(\mhV)\sim \N_{q} 
\Big(
\vecc(\mV)
,
\sigma^2 I_m \otimes \big((L'L)^{-1}L'(XX')^{-1}L(L'L)^{-1}\big)
\Big),
\end{equation}
where $q$ is the product of the Tucker rank. Now, because $\hM_k$ and $\hL_k$ are MLEs and they are identifiable given the rest, as $n\rightarrow \infty$ $\hM_k\overset{p}{\rightarrow} M_k$ and $\hL_k\overset{p}{\rightarrow} L_k$, and therefore
\begin{equation}\label{eq:step2}
(\bigotimes_{k=p}^1 \hM_k)\otimes (\bigotimes_{k=l}^1 \hL_k)
\overset{p}{\rightarrow}
M\otimes L.
\end{equation}
The rest of the proof follows by left-multiplying the $\vecc(\mhV)$ of \eqref{eq:step1} by the left hand side of \eqref{eq:step2} and applying Slutsky's Theorem.
\end{proof}
\subsection{Proof of Theorem \ref{thm:inference_others}}\label{proof:inference_others}
Before proving Theorem \ref{thm:inference_others} we state some of the matrices involved in it. First, let 
$$
\boldsymbol{\hat{\theta}}_{CP}  = [\vecc(\hL_1)'\vecc(\hL_2)'\ldots \vecc(\hL_l)'\vecc(\hM_1)'\vecc(\hM_2)'\ldots \vecc(\hM_p)']'.
$$
Then the Jacobian matrix 
$
J_{CP} =\frac{\partial}{\partial \boldsymbol{\hat{\theta}}_{CP}} \vecc(\mhB_{CP})
$
is a block matrix with blocks as per Lemma~\ref{lemma:ders}\ref{lemma:ders:z}. Furthermore, we have 
\begin{equation}\label{eq:supp:R_CP}
R_{CP} = [A_1' A_2' \ldots  A_l' B_1' B_2' \ldots   B_p']',
\end{equation}
where for all $k = 1,2,\ldots,l$,
\begin{equation}\label{eq:supp:A_CP}
A_k = [S_{1k}^{-1}H_{1k}^{CP} \Sigma^{-1}\ldots S_{1k}^{-1}H_{nk}^{CP} \Sigma^{-1}]
,\quad
S_{1k} = \sum\limits_{i=1}^n  H_{ik}^{CP} \Sigma^{-1}  {H_{ik}^{CP}}'
\end{equation}
and for all $k = 1,2,\ldots,p$,
\begin{equation}\label{eq:supp:B_CP}
B_k =
\left(
[S_{2k}^{-1}{G_{1k}^{CP}}\Sigma_{-k}^{-1} \ldots S_{2k}^{-1}{G_{nk}^{CP}}\Sigma_{-k}^{-1}] \otimes \I_{m_k}
\right) (I_n\otimes K_{(k)}),
\quad
S_{2k} = \sum\limits_{i=1}^n G_{ik}^{CP}
\Sigma_{-k}^{-1}   {G_{ik}^{CP}}',
\end{equation}
where $K_{(k)}$ is given in Lemma \ref{lemma1}.\ref{lemma1:u} and the matrices $S_{1k}$ and $S_{2k}$ are simplified in Section \ref{supp:CP}. Now we can prove Theorem \ref{thm:inference_others}.

\begin{proof}
From Equations \eqref{eq:CPM} and \eqref{eq:CPL}, we obtain that $\vecc(\hat{L}_k) = A_k \boldsymbol{y}$ and $\vecc(\hat{M}_k) = B_k \boldsymbol{y}$, and thus
\begin{equation}\label{eq:thet_CP}
\boldsymbol{\hat{\theta}}_{CP}
=
R_{CP}
\boldsymbol{y}.
\end{equation}
Further, from our normality and independence assumptions,  
\begin{equation}\label{eq:asymp_c_1}
\boldsymbol{y} = 
\begin{bmatrix}
 \vecc(\bmY_1) \\ \hdots \\   \vecc(\bmY_n)
\end{bmatrix}
\sim
\N_{m\times n}\left(
\begin{bmatrix}
 \vecc\langle\mX_1|\mB\rangle \\ \hdots \\   \vecc\langle\mX_n|\mB\rangle
\end{bmatrix}, \I_n\otimes \Sigma \right),
\end{equation}
and from \eqref{eq:thet_CP} and \eqref{eq:asymp_c_1},
$$
\boldsymbol{\hat{\theta}}_{CP}
\sim \N_{R\times( \sum_{i=1}^pm_i+\sum_{i=1}^lh_i)}
\Big(
\boldsymbol{\theta}_{CP}
,
R_{CP}( \I_n\otimes \Sigma)R_{CP}'
\Big).
$$
Now consider the transformation $g:\mathbb{R}^{R\times( \sum_{i=1}^pm_i+\sum_{i=1}^lh_i)}\rightarrow\mathbb{R}^{m\times h}$ such that $g( \boldsymbol{\hat{\theta}}_{CP})=\vecc(\mhB_{CP})$. The Jacobian matrix of this transformation is $J_{CP}$. Using the multivariate extension of Taylor's Theorem as in Theorem 5.2.3 in \citet{Lehmann99}, we obtain 
$$
g( \boldsymbol{\hat{\theta}}_{CP})
\overset{d}{\rightarrow}
\N_{m\times h}
\Big(
g(\boldsymbol{\theta}_{CP})
,
J_{CP}R_{CP}( \I_n\otimes \Sigma)R_{CP}'J_{CP}'
\Big),
$$
where $g( \boldsymbol{\hat{\theta}}_{CP}) = \vecc(\mhB_{CP}) $ and $g(\boldsymbol{\theta}_{CP})  = \vecc(\mB_{CP})$.
\end{proof}
\subsection{The OP and TR cases}\label{sec:inference_TRnOP:actual}
We now state and prove some results on inference for the OP and TR
format:
\begin{thm}\label{thm:inference_TRnOP:actual} \hfill
\begin{enumerate}[label=(\alph*)]
\item Consider Equation \eqref{eq:model_noM} with $\mB = \mB_{TR}$ as in \eqref{eq:reg_TR} with its analogue
$
\mhB_{TR}
$ and 
$$
\boldsymbol{\hat{\theta}}_{TR} = [\vecc(\hat{\mL}_1)'\vecc(\hat{\mL}_2)'\hdots \vecc(\hat{\mL}_l)'\vecc(\hat{\mM}_{1(2)})'\vecc(\hat{\mM}_{2(2)})'\hdots \vecc(\hat{\mM}_{3(2)})']'.
$$
Then as $n\rightarrow \infty$,
$$
\vecc(\mhB_{TR})
\overset{d}{\rightarrow} 
\N_{m\times h} \Big(
\vecc(\mB_{TR})
,
J_{TR}R_{TR}( \I_n\otimes \Sigma)R_{TR}'J_{TR}'
\Big),
$$
where $J_{TR} =\dfrac{\partial \vecc(\mhB_{TR})}{\partial \boldsymbol{\hat{\theta}}_{TR}}$ is a block matrix with blocks given in Lemma \ref{lemma:ders}\ref{lemma:ders:x} and $R_{TR}$ is equal to $R_{CP}$ in  equation \eqref{eq:supp:R_CP} after replacing $(G_{ik}^{CP},H_{ik}^{CP})$ with $(G_{ik}^{TR},H_{ik}^{TR})$.
\item Now let Equation \eqref{eq:model_noM} have $\mB = \circ [\![ M_1,\sdots,M_p  ]\!]$with its analogue
$
\mhB_{OP}
$ and 
$$
\boldsymbol{\hat{\theta}}_{OP} = [\vecc(\hat{M}_1)'\vecc(\hat{M}_2)'\hdots \vecc(\hat{M}_p)']',
$$
Then as $n\rightarrow \infty$,
$$
\vecc(\mhB_{OP})
\overset{d}{\rightarrow} 
\N_{m\times h} \Big(
\vecc(\mB_{OP})
,
J_{OP}R_{OP}( \I_n\otimes \Sigma)R_{OP}'J_{OP}'
\Big),
$$
where $J_{OP} =\dfrac{\partial \vecc(\mhB_{OP})}{\partial \boldsymbol{\hat{\theta}}_{OP}}$ is a block matrix with blocks given in Lemma \ref{lemma:ders}\ref{lemma:ders:y} and $R_{OP} = [ B_1' B_2' \hdots   B_p'
]'$ , where $B_k$ is given in equation \eqref{eq:supp:B_CP} for all $k=1,\dots,p$ after replacing $G_{ik}^{CP}$ with $G_{ik}^{OP}$.
\end{enumerate}
\end{thm}
\begin{proof} \hfill
The proof of both parts is similar to that of Theorem~\ref{thm:inference_others} and thus omitted.
\end{proof}

\subsection{Jacobian matrices}
The Jacobians needed in Theorems~\ref{thm:inference_others} and \ref{thm:inference_TRnOP:actual} are block matrices, where each block is given in the following lemma
\begin{lemma}\label{lemma:ders}\hfill
\begin{enumerate}[label=(\alph*)]
\item \label{lemma:ders:z} Let 
$
\mB_{CP} = [\![M_1,M_2,\hdots,M_p]\!]
\in  \mathbb{R}^{m_1\times m_2\times \hdots\times m_p}.
$
Then for $k=1,2,\hdots,p$,
$$
\dfrac{\partial \vecc(\mB_{CP})}{\partial \vecc(M_k)}
=K_{(k)}'(T_k^{CP}\otimes I_{m_k})
, \text{ where }
T_k^{CP} = \bigodot\limits_{q=p,q\neq k}^1 M_q.
$$
\item \label{lemma:ders:x} Let 
$
\mB_{TR} = \tr(\mM_1\times^1\mM_2\times^1\hdots\times^1\mM_p)
\in  \mathbb{R}^{m_1\times m_2\times\hdots\times m_p}.
$
Then for $k=1,2,\hdots,p$,
$$
\dfrac{\partial \vecc(\mB_{TR})}{\partial \vecc(\mM_{k(2)})}
=K_{(k)}'(T_k^{TR'}\otimes I_{m_k}),
$$
where
$$
T_k^{TR} = \tr(\mM_{k+1}\times^1\hdots \times^1\mM_p\times^1\mM_1\times^1\hdots\times^1\hdots\times^1\mM_{k-1})_{(\{1,p+1\}\times\{2,3,\hdots,p\} )}
.$$
\item \label{lemma:ders:y} Let $\mB_{OP} = \circ [\![ M_1,M_2 \hdots,M_p  ]\!] \in  \mathbb{R}^{m_1\times  \hdots\times m_p\times n_1\times  \hdots\times n_p}$. Then for $k=1,2,\hdots,p$,
$$
\dfrac{\partial \vecc(\mB_{OP})}{\partial \vecc(M_k')}
=K_{(p+k)}'K_{(k+1)}'
\Big((\vecc \circ [\![ M_1, \hdots,M_{k-1},M_{k+1},\hdots,M_p  ]\!])\otimes I_{m_k n_k} \Big).
$$
\end{enumerate}
\end{lemma}
\begin{proof}\hfill
\begin{enumerate}[label=(\alph*)]
\item First, because $\mB_{CP(k)} = M_kT_k^{CP'}$, we have
$$
\dfrac{\partial \vecc(\mB_{CP(k)})}{\partial \vecc(M_k)}
=T_k^{CP}\otimes I_{m_k}.
$$
The result follows from Lemma \ref{lemma1}\ref{lemma1:u}.
\item The derivation of this results follows a similar path as in part
  (a) because $\mB_{TR(k)} = M_{k(2)} T_k^{TR} $. 
\item Let $\mathscr{S}$ be as in Theorem \ref{thm:mat_outer_product}\ref{thm:mat_outer_product:x}. Then the vectorization of \eqref{eq:mat_outer_product} results in
$$
\dfrac{\partial \vecc(\mB_{OP(\mathscr{S}\times \mathscr{S}^\mathsf{c})})}{\partial \vecc(M_k')}
=\Big((\vecc \circ [\![ M_1, \hdots,M_{k-1},M_{k+1},\hdots,M_p  ]\!])\otimes I_{m_k n_k} \Big).
$$
The result follows from 
$
\vecc(\mB_{OP(\mathscr{R}\times \mathscr{C})}) = K_{(k+1)}K_{(p+k)}\vecc(\mB_{OP}) 
$
, where $K_{(k)}$ is orthogonal as in Lemma \ref{lemma1}\ref{lemma1:u}.
\end{enumerate}
\end{proof}

\subsection{Proof of Theorem \ref{thm:inference_with_intercept}}
\begin{proof}\label{proof:inference_with_intercept}
Based on Section \ref{sec:est_intercept}, the estimation of $\mB$ when the intercept is present can be perfomed by centering the covariates and responses. Let
$
C_n = \I_n - \frac{1}{n} \boldsymbol{1}_n \boldsymbol{1}_n'.
$
For the Tucker case in Theorem \ref{thm:inference_Tucker}, let $Y_c = [(\vecc \bmY_1^c)(\vecc \bmY_2^c)\hdots (\vecc \bmY_n^c)]$, where $\bmY_i^c = \bmY_i - \bar{\bmY}$. Then 
$
\vecc(Y_c') = (\I_m\otimes C_n)\vecc(Y'), 
$
and similar to Equation~\eqref{eq:V}, 
$$
\vecc(\mhV) = 
\big[M'\Sigma^{-1}\otimes W^{+'}
\big](\I_m\otimes C_n)\vecc(Y')
= 
\big[M'\Sigma^{-1}\otimes W^{+'}
\big]\vecc(Y'),
$$
where the last equality follows because $X\boldsymbol{1}_n =\bm{0}$
when the covariates are centered, and therefore
$W^{+'}\boldsymbol{1}_n =\bm{0}$ also. The remaining results follow
using the same steps as in the proof of Theorem
\ref{thm:inference_Tucker}, but with centered covariates. For the CP
case as in Theorem \ref{thm:inference_others}, if $\boldsymbol{y}_c =
[(\vecc \bmY_1^c)'(\vecc \bmY_2^c)'\hdots (\vecc \bmY_n^c)']'$ are the
centered responses, then $\boldsymbol{y}_c = (C_n\otimes
I_m)\boldsymbol{y}$, and similar to Equation~\eqref{eq:thet_CP}, 

$$
\boldsymbol{\hat{\theta}}_{CP}
=
R_{CP}(C_n\otimes I_m)\boldsymbol{y}
=
R_{CP}\boldsymbol{y},
$$
where the last step follows because $R_{CP}(\boldsymbol{1}_n\otimes \I_m)  =\bm{0}$. The rest results by following the same steps as in the proof of Theorem~\ref{thm:inference_others}, but with centered covariates. The proof of the TR and OP cases are similar and thus omitted.
\end{proof}

\subsection{Distribution of the scale components}\label{inference:covariance}
The final inference result is on the distribution of the
estimated scale matrices $\hSigma_1,\hSigma_2,\ldots,\hSigma_p$. We obtain the Fisher information matrix and show that it is singular.
\begin{thm}\label{thm:inference_covariance}
Let
$\boldsymbol{\theta}_{\Sigma} = [ (\vech \Sigma_1)',\vech \Sigma_2)' ,\ldots,(\vech \Sigma_p)']',$
where $\Sigma_1,\Sigma_2,\ldots,\Sigma_p$ are scale matrices with no
restrictions on their proportionality, and $\vech$ is the
\textit{half-vectorization} mapping, as discussed further in
Section \ref{app:multilinear_statistics}. Then 
\begin{enumerate}[label=(\alph*)]

\item The joint linear components
$
\begin{bmatrix}
\vecc(\hat{\Upsilon})  \\
\vecc(\mhB)
\end{bmatrix}$ and $\boldsymbol{\widehat{\theta}}_{\Sigma}$ are asymptotically independent regardless of the format of $\mhB$.

\item The Fisher information with respect to $\boldsymbol{\theta}_{\Sigma}$, denoted as $\mathbb{I}_{\Sigma}$, is a block matrix with $kth$ block diagonal matrices ($k = 1,2,\ldots,p$) given by
\begin{equation} 
\mathbb{E}\Big(-\dfrac{\partial^2 \ell(\Sigma_k) 
}{\partial(\vech \Sigma_k)\partial(\vech \Sigma_k)'}\Big)
=
\dfrac{nm_{-k}}{2} D_{m_k}' (\Sigma_k^{-1} \otimes \Sigma_k^{-1})D_{m_k}
\end{equation}
and $(k,l)$th block matrices for $k \neq l$, $k,l = 1,2,\ldots,p$, given as
\begin{equation}
\mathbb{E}\Big(-\dfrac{\partial^2 \ell(\Sigma_k,\Sigma_l)
}{\partial(\vech \Sigma_k)\partial(\vech \Sigma_l)'} \Big)
=
\dfrac{nm_{-kl}}{2}  D_{m_k}' \big( \vecc (\Sigma_k^{-1}) \vecc (\Sigma_l^{-1})')D_{m_l},
\end{equation}
where $D_{m_k}$ is the duplication matrix of Section \ref{app:multilinear_statistics} and $m_{-kl}= \prod_{i=1,i\neq k, i\neq l}^p m_i$.

\item The Fisher information $\mathbb{I}_{\Sigma}$ is singular.
\end{enumerate}
\end{thm}

\begin{proof}\hfill
\begin{enumerate}[label=(\alph*)]
\item This statement follows upon expressing the
  Equation~\eqref{eq:model} as the MVMLR regression model $
\vecc(\bmY_i)  = \vecc(\Upsilon) + \mB_{<l>}'\vecc(\mX_i)+ \boldsymbol{e}_i,
$, and then applying Theorem~15.4 in
  \citet{magnusandneudecker99}. 
\item 
Equation~\eqref{eq:model} with $\sigma^2=1$ can be expressed as
\begin{equation}\label{eq:tensontensMLE2}
\mathcal{Z}_i = \bmY_i - \langle  \mX_i|  \mB \rangle \overset{iid}{\sim}N_{m_1,m_2,\hdots,m_p}
(0, 
\Sigma_1,\Sigma_2,\hdots,\Sigma_p)
,\quad
i = 1,2\hdots,n.
\end{equation}
We will start by finding the diagonal of the block Fisher information matrix. We know from Theorem \ref{thm:tensnorm_reshap}\ref{tensnorm_reshap:x} that
\begin{equation}\label{eq:centermat_model}
\mathcal{Z}_{i(k)} \sim \N_{m_k,m_{-k}} \big(0, \Sigma_k,\Sigma_{-k}\big)
\end{equation}
 for $k=1,\hdots,p$. We find the second differential of the loglikelihood of $\mathcal{Z}_{i(k)}$ in \eqref{eq:centermat_model} with respect to $\Sigma_k$ only as 
 
\begin{equation}\label{eq:diff_likel}
\begin{split}
\ell(\Sigma_k) =  -\dfrac{m_{-k}}{2} \log |\Sigma_k| - \dfrac{1}{2} \tr(\Sigma_k^{-1}\mathcal{Z}_{i(k)}\Sigma_{-k}^{-1}\mathcal{Z}_{i(k)}'), \hspace*{2cm}
&\\
\partial \ell(\Sigma_k) = -\dfrac{m_{-k}}{2} \tr(\Sigma_k^{-1}\partial \Sigma_k )+ \dfrac{1}{2} \tr(\Sigma_k^{-1}\partial \Sigma_k\Sigma_k^{-1}\mathcal{Z}_{i(k)}\Sigma_{-k}^{-1}\mathcal{Z}_{i(k)}'),\hspace*{1cm}
&\\
\partial^2 \ell(\Sigma_k) = \dfrac{m_{-k}}{2} \tr(\Sigma_k^{-1}\partial \Sigma_k \Sigma_k^{-1}\partial \Sigma_k)-  \tr(\Sigma_k^{-1}\partial \Sigma_k\Sigma_k^{-1}\partial \Sigma_k\Sigma_k^{-1}\mathcal{Z}_{i(k)}\Sigma_{-k}^{-1}\mathcal{Z}_{i(k)}').
\end{split}
\end{equation} 

Now, based on Equation~\eqref{eq:centermat_model}, $
\mathcal{Z}_{i(k)}
\Sigma_{-k}^{-1}
\mathcal{Z}_{i(k)}' \sim \mathcal{W}_{m_k}(m_{-k},\Sigma_k),$
where $\mathcal{W}$ denotes the Wishart distribution on $m_k\times
m_k$-dimension random matrices and $m_{-k}$ degrees of freedom. Therefore
$
\mathbb{E}( \mathcal{Z}_{i(k)}\Sigma_{-k}^{-1}\mathcal{Z}_{i(k)}') = m_{-k}\Sigma_k
$. Thus using the duplication matrix and expressing the matrix trace as a vector inner product in Lemmas \ref{lemma1}\ref{lemma1:x} and \ref{lemma1}\ref{lemma1:w}
\begin{equation*}
    \begin{split}
\mathbb{E}(-\partial^2 \ell(\Sigma_k) )
&= \dfrac{m_{-k}}{2} \tr(\Sigma_k^{-1}\partial \Sigma_k \Sigma_k^{-1}\partial \Sigma_k) =
\partial(\vech \Sigma_k)'
\Big( \dfrac{m_{-k}}{2}  D_{m_{-k}}' (\Sigma_k \otimes \Sigma_k) D_{m_{-k}}\Big)
\partial(\vech \Sigma_k).
\end{split}
\end{equation*}
Now we find the element in the (2,1)th position of the block-Fisher
information matrix. Finding this element is enough to find the rest of
the Fisher information matrix because the order of the scale matrices
in the TVN distribution can be permuted in concordance with the modes of
the random tensor. Let $\Sigma_{-12} = \bigotimes_{i=p}^3 \Sigma_i$ and $m_{-12} =\prod_{i=3}^p m_i $. Then based on the distribution of $\mathcal{Z}_{i(1)}$ in \eqref{eq:centermat_model} we can write the log likelihood as 
\begin{equation}\label{eq:loglikel_S12}
\begin{split}
\ell(\Sigma_1,\Sigma_2)
&= -\dfrac{1}{2} \tr\big[(\Sigma_{-12}^{-1} \otimes
\Sigma_2^{-1})\mathcal{Z}_{i(1)} '
\Sigma_1^{-1}\mathcal{Z}_{i(1)}\big] \quad \text{(by Property \ref{prop:Kron}\ref{prop:Kron:c})}
\\&=
-\dfrac{1}{2}\vecc(\Sigma_{-12}^{-1} \otimes \Sigma_2^{-1})'
(\mathcal{Z}_{i(1)} \otimes \mathcal{Z}_{i(1)})'
\vecc{(\Sigma_1^{-1})}\quad \text{(by Lemmas \ref{lemma1}\ref{lemma1:x} and \ref{lemma1}\ref{lemma1:w})}
\\&=
-\dfrac{1}{2} \vecc{(\Sigma_2^{-1})}' 
R_{\Sigma_{-12}^{-1}}'
 (\mathcal{Z}_{i(1)} \otimes \mathcal{Z}_{i(1)})' 
\vecc{(\Sigma_1^{-1})},\quad \text{(by Property \ref{prop:Com}\ref{prop:Com:c})}
\end{split}
\end{equation}
Furthermore, if $X\sim N_{m_1,m_2}(0,\Sigma_1,\Sigma_2)$,  $J^{k,l}$ is a single-entry matrix (with 1 at the position $(k,l)$ and 0 everywhere else) and $\Sigma_1^{1/2}[k,i] = \sigma^1_{ki}$, $\Sigma_2^{1/2}[j,l] = \sigma^2_{jl}$, then for  $Z \sim N_{p,r} (0,I_p,I_r)$,
\begin{equation}\label{proof:kronmatnorm}
\begin{split}
\mathbb{E} (X\otimes X) 
&= (\Sigma_1^{1/2}\otimes \Sigma_1^{1/2})\mathbb{E} (Z\otimes Z)(\Sigma_2^{1/2}\otimes \Sigma_2^{1/2})
=
(\Sigma_1^{1/2}\otimes \Sigma_1^{1/2})
\{ J^{k,l} \}_{k,l}
(\Sigma_2^{1/2}\otimes \Sigma_2^{1/2})
\\&=
\Big\{ 
\sum\limits_{i=1}^p \sum\limits_{j=1}^r
\sigma^1_{ki}\Sigma_1^{1/2} J^{i,j} \Sigma_2^{1/2} \sigma^2_{jl}
\Big\}_{k,l}
\\&
=\left\{ 
\sum\limits_{i=1}^p \sum\limits_{j=1}^r
\sigma^1_{ki}\Sigma_1^{1/2}[:,i]\Sigma_2^{1/2}[j,:] \sigma^2_{jl}
\right\}_{k,l}
\\&=
\Big\{ 
\big(
\sum\limits_{i=1}^p 
\sigma^1_{ki}\Sigma_1^{1/2}[:,i]\big)
\big(\sum\limits_{j=1}^r\Sigma_2^{1/2}[j,:] \sigma^2_{jl}\big)
\Big\}_{k,l}
\\&=
\Big\{ 
\big[\sigma^1_{k1} I_p \hdots \sigma^1_{kp} I_p] \vecc(\Sigma_1^{1/2})
\vecc(\Sigma_2^{1/2})'
\begin{bmatrix}
\sigma^2_{1l} I_r\\
\vdots\\
\sigma^2_{rl} I_r
\end{bmatrix}
\Big\}_{k,l}
\\&=
(\Sigma_1^{1/2} \otimes I_p)
\vecc(\Sigma_1^{1/2})
\vecc(\Sigma_2^{1/2})'
(\Sigma_2^{1/2} \otimes I_r)
=\vecc(\Sigma_1)\vecc(\Sigma_2)'. 
\end{split}
\end{equation}
To differentiate the negative loglikelihood of~\eqref{eq:loglikel_S12}, we use that if $\Sigma$ is non-singular then, $
 \dfrac{\partial \vecc\Sigma^{-1}}{\partial\vecc \Sigma}=-(\Sigma^{-1} \otimes \Sigma^{-1})$, 
which follows from vectorizing both sides of  $\partial \Sigma^{-1} =
-\Sigma^{-1} \partial \Sigma\Sigma^{-1} $. From application of this
operation to both sides of~\eqref{eq:loglikel_S12} and taking expectation of \eqref{proof:kronmatnorm} results in
  \begin{equation*}
    \begin{split}
-\mathbb{E}\Big(
\dfrac{\partial^2 \ell(\Sigma_1,\Sigma_2) 
}{\partial(\vecc \Sigma_2)\partial(\vecc \Sigma_1)'}
\Big)
&=
 \dfrac{1}{2}(\Sigma_2^{-1} \otimes \Sigma_2^{-1})
R_{\Sigma_{-12}^{-1}}'  \Big\{\mathbb{E}(\mathcal{Z}_{i(1)} \otimes \mathcal{Z}_{i(1)})\Big\}' 
(\Sigma_1^{-1} \otimes \Sigma_1^{-1}) 
\\&=
\dfrac{1}{2}(\Sigma_2^{-1} \otimes \Sigma_2^{-1})
R_{\Sigma_{-12}^{-1}}'  
\Big\{
(\vecc \Sigma_1)(\vecc (\Sigma_{-12}\otimes\Sigma_2))'
\Big\}' 
(\Sigma_1^{-1} \otimes \Sigma_1^{-1}) 
\\&=
\dfrac{1}{2}(\Sigma_2^{-1} \otimes \Sigma_2^{-1})
R_{\Sigma_{-12}^{-1}}'  
\vecc (\Sigma_{-12}\otimes\Sigma_2)
(\vecc \Sigma_1^{-1})'
\\&=
\dfrac{1}{2}(\Sigma_2^{-1} \otimes \Sigma_2^{-1})
R_{\Sigma_{-12}^{-1}}'  
R_{\Sigma_{-12}}
(\vecc \Sigma_2)
(\vecc \Sigma_1^{-1})' 
\quad  \text{(Lemma \ref{prop:Com}\ref{prop:Com:c})}
\\&=
\dfrac{m_{-12}}{2}(\Sigma_2^{-1} \otimes \Sigma_2^{-1})
(\vecc \Sigma_2)
(\vecc \Sigma_1^{-1})'
\quad (\text{see below})
\\&=
\dfrac{m_{-12}}{2}
(\vecc \Sigma_2^{-1})
(\vecc \Sigma_1^{-1})'.
\end{split}
\end{equation*}
The rest follows by multiplying $D_{m_1}'$ on the left and $D_{m_{2}}$
on the right. The following term simplifies greatly using
Property~\ref{prop:Com} of the commutation matrix and
Property~\ref{prop:Kron} of the Kronecker product: 
\begin{equation*}
    \begin{split}
R_{\Sigma_{-12}^{-1}}'  
R_{\Sigma_{-12}}
&=
\Big(
((\vecc\Sigma_{-12}^{-1})'\otimes I_{m_2})(I_{m_{-12}}\otimes K_{m_{-12},m_2})
(I_{m_{-12}}\otimes K_{m_2,m_{-12}})(\vecc\Sigma_{-12}\otimes I_{m_2})
\Big)\otimes I_{m_2} 
\\&=
\Big(
((\vecc\Sigma_{-12}^{-1})'\otimes I_{m_2})(I_{m_2 \times m_{-12}^2})(\vecc\Sigma_{-12}\otimes I_{m_2})
\Big)\otimes I_{m_2} 
\\&=
\tr(\Sigma_{-12}^{-1}\Sigma_{-12})\otimes I_{m_2^2} 
=m_{-12}I_{m_2^2}.
\end{split}
\end{equation*}
\item Let $\mathbb{I}_{\Sigma_{12}}$ denote the first four block matrices of the Fisher information,
$$
\mathbb{I}_{\Sigma_{12}} = 
\dfrac{nm_{-12}}{2}\begin{bmatrix}
m_2D_{m_1}'(\Sigma_1^{-1}\otimes\Sigma_1^{-1})D_{m_1} &
  D_{m_1}' \big(\vecc(\Sigma_1^{-1})\vecc(\Sigma_2^{-1})'\big)D_{m_2}
  \\
D_{m_2}' \big(\vecc(\Sigma_2^{-1})\vecc(\Sigma_1^{-1})'\big)D_{m_1}   &
 m_1D_{m_2}'(\Sigma_2^{-1}\otimes\Sigma_2^{-1})D_{m_2} 
\end{bmatrix}
=
\begin{bmatrix}
A &B \\ B' &C
\end{bmatrix}.
$$
 Using the result of Section 3.8 of \citet{magnusandneudecker99}, where $E^-$ denotes the Moore-Penrose inverse of $E$, we obtain:
$$
A^{-1} = \dfrac{2}{nm_{-1}}D_{m_1}^-(\Sigma_1\otimes\Sigma_1)D_{m_1}^{-'}
,\quad
C^{-1} = \dfrac{2}{nm_{-2}}D_{m_2}^-(\Sigma_2\otimes\Sigma_2)D_{m_2}^{-'},
$$
and$$
BC^{-1}B' =\dfrac{nm_{-1}}{2m_1}D_{m_1}' (\vecc\Sigma_1^{-1})(\vecc\Sigma_1^{-1})'D_{m_1}.
$$
Using these results we obtain 
\begin{equation*}
    \begin{split}
\tr(A^{-1}BC^{-1}B') 
&=
\dfrac{1}{m_1}\tr\Big(D_{m_1}^-(\Sigma_1\otimes\Sigma_1)D_{m_1}^{+'}D_{m_1}' (\vecc\Sigma_1^{-1})(\vecc\Sigma_1^{-1})'D_{m_1}\Big)
\\&=
\dfrac{1}{m_1}\tr\Big(D_{m_1}^-(\vecc\Sigma_1)(\vecc\Sigma_1^{-1})'D_{m_1}\Big)=\dfrac{m_1}{m_1} = 1.
\end{split}
\end{equation*}
This last result implies that for the Schur complement of $\mathbb{I}_{\Sigma_{12}}$, namely $S = (A-BC^{-1}B')$, it follows that $
S\big[A^{-1}D_{m_1}'$ $ (\vecc\Sigma_1^{-1})\big] = \bm{0}, 
$
which means that the  $S$ has a non-trivial kernel, and therefore is
singular. So $\mathbb{I}_{\Sigma_{12}}$ is singular
\citep{luandshiou02}, and consequently so is $\mathbb{I}_\Sigma$.
\end{enumerate}
\vspace*{-.5cm}
\end{proof}

\section{Supplement to Section \ref{sec:simulation}}\label{app:camelidsWilks}
To calculate Wilks' Lambda statistic $\Lambda = |\hat\Sigma_{R}|/|\hat\Sigma_{T}|$ we need $\hat\Sigma_{R}$, which is the sample covariance matrix of the residuals
after fitting~\eqref{eq:realsim_camelids}, and is formally expressed as
$$
\hat\Sigma_{R} = \sum_{i=1}^{4}\sum_{j=1}^{3}\sum_{k=1}^{50}\vecc(Y_{ijk}-\langle X_{ij}|\mhB\rangle)\vecc(Y_{ijk}-\langle X_{ij}|\mhB\rangle)',
$$
with appropriate $\mB$ and $\mhB$ is its estimate. The matrix $\hat\Sigma_{T}$ is the sample covariance matrix of the simpler model's residuals, which finds a common mean across all camelids, and is formally written as  
$$
\hat\Sigma_{T} = \sum_{i=1}^{4}\sum_{j=1}^{3}\sum_{k=1}^{50}\vecc(Y_{ijk}-\langle \bx_{j}|\mhB^*\rangle)\vecc(Y_{ijk}-\langle \bx_{j}|\mhB^*\rangle)',
$$
where $\bx_{j}^*\in \mathbb{R}^3$ has 1 in position $i$ and zeroes elsewhere. The tensor $\mhB^*$ is the low-rank estimator of $\mB^*$, both
of size $3\times 87\times 106$, and contains the grand-camelid mean image. 
The tensor $\mB$ is of order 4 while $\mB^*$ is of order 3, and therefore considerations must be considered when choosing the ranks of $\mB^*$. We chose the same CP rank for both $\mB$ and $\mB^*$, and the rank of $\mB^*$ was chosen to concord to the corresponding dimensions of $\mB$ in the TK and TR cases.
The matrices $\hat\Sigma_{R}$ and $\hat\Sigma_{T}$ are
non-negative definite of size 
$9222\times 9222$ and rank $600$, so the 
Wilks' Lambda statistic calculation involves the generalized determinant , which is the product of the non-zero eigenvalues.

\section{Supplement to Section \ref{sec:data_application}}
This section will need the following lemma, which identifies the distribution of the Tucker product of a random TVN  tensor.
\begin{lemma}\label{lemma:tvn_product}
If $\bY \sim \N_{\bbm}(0,\Sigma_1,\Sigma_2,\hdots,\Sigma_p)$ and  $M_k$ $\in$ $\mathbb{R}^{n_k\times m_k}$ for all $k=1,2,\hdots,p$, then $\bX=[\![\bY;M_1,M_2,\hdots,M_p]\!]\sim \N_{[n_1,n_2,\hdots,n_p]'}(0,M_1\Sigma_1M_1',M_2\Sigma_2M_2',\hdots,$ $M_p\Sigma_pM_p')$.
\end{lemma}
\begin{proof}
  From Lemma \ref{lemma1}\ref{lemma1:w} $\vecc(\bX) = (\bigotimes_{k=p}^1M_k)\vecc(\bY)$, and  
  therefore based on the properties of the MVN distribution and Property \ref{prop:Kron}(d),
$$
\vecc(\bX)  \sim \N_{n_1\times \hdots\times n_p}\big(0,\otimes_{k=p}^1(M_k\Sigma_kM_k')\big).
$$
The remainder of the proof follows from 
and Definition \ref{def:tvn}.
\end{proof}

\subsection{Supplement to Section \ref{application:suicide}}\label{sec:app:brain_interaction}

In Section \ref{application:suicide} we are interested in doing inference on the interaction between word type and subject
suicide attempter/ideator status to determine markers for suicide risk assessment and intervention. The $3$-level interaction is given as 
$\mhB_{*}=\mhB\times_1  \boldsymbol{c}_1' \times_2  C_2 \times_3
\boldsymbol{c}_3'$, where $\boldsymbol{c}_1=(1,-1)'$ is a contrast
vector that finds differences between suicide attempter/ideation status, $\boldsymbol{c}_3=
(1,1,\ldots,1)'/10$ a contrast vector that averages the 10 words of each word type, and 
\begin{equation}\label{inter:brain}
C_2 = 
\begin{bmatrix}
1 & 0 & -1 \\
1 & -1 & 0 \\
0 & -1 & 1 
\end{bmatrix}
\end{equation}
is a contrast matrix for differences across word type. The tensor
$\mhB_{*}$ contains all the interactions; for instance, the first row
of $C_2$ finds differences between the first and third word types, so the first level interaction $\mhB_{*}(1,:,:,:)$ 
$$
\dfrac{1}{10}\sum_{k=1}^{10}
\left(
\mhB(1,1,k,:,:,:)-
\mhB(1,3,k,:,:,:)-
\mhB(2,1,k,:,:,:)+
\mhB(2,3,k,:,:,:)
 \right),
$$
identifies regions of significant differences in how death- and
negative-connoted words affect suicide
ideators and attempters. Similarly, the second row of $C_2$ finds
differences between the first two word types, so the second level
interaction $\mhB_{*}(2,:,:,:)$ is $$
\dfrac{1}{10}\sum_{k=1}^{10}
\left(
\mhB(1,1,k,:,:,:)-
\mhB(1,2,k,:,:,:)-
\mhB(2,1,k,:,:,:)+
\mhB(2,2,k,:,:,:)
 \right),
$$
and its significant regions correspond to locations with a
significant difference in how the contrast
between death-related and positive-connoting words affect suicide
ideators and attempters. Finally, the third row of $C_2$ finds
differences between the first and third word types and so the third-level interaction $\mhB_{*}(3,:,:,:)$ is 
$$
\dfrac{1}{10}\sum_{k=1}^{10}
\left(
\mhB(1,3,k,:,:,:)-
\mhB(1,2,k,:,:,:)-
\mhB(2,3,k,:,:,:)+
\mhB(2,2,k,:,:,:)
 \right),
$$
and its significant regions correspond to locations with a
significant difference in how the contrast
between negative- and positive-connoting words affect suicide
ideators and attempters. 

To obtain the asymptotic distribution of the estimate $\mhB$ we use Theorem \ref{thm:inference_Tucker} and Definition \ref{def:tvn}, to obtain that as $n\rightarrow\infty$,
\begin{align}\label{eq:tensnorm_brain}
\mhB
\overset{d}{\rightarrow} 
 \N_{2,3,10,43,56,20}( \mB,&\sigma^2 P_1(XX')^{-1}P_1 ,M_1M_1',
 M_2M_2',M_3M_3',M_4M_4',M_5M_5'),
\end{align}
where $XX'$ is a diagonal matrix with two non-zero entries $(9,8)$. From \eqref{eq:tensnorm_brain}, \eqref{inter:brain} and Lemma \ref{lemma:tvn_product}, the asymptotic
distribution of $\mhB_{*}$ is
\begin{equation}
\mhB_{*}
\overset{d}{\rightarrow} 
 \N_{3,20,43,56}( \mB_{*},\tau^2C_2M_1M_1'C_2',M_3M_3',M_4M_4',M_5M_5'),
\end{equation}
where $\tau^2 =
\sigma^2 \times (\boldsymbol{c}_1' P_1(XX')^{-1}P_1\boldsymbol{c}_1)\times(\boldsymbol{c}_3'M_2M_2'\boldsymbol{c}_3)$ and
$\mB_{*}=\mB\times_1  \boldsymbol{c}_1' \times_2  C_2 \times_3
\boldsymbol{c}_3'$. Using~\eqref{inter:asympt}, we marginally
standardize  $\mhB_{*}$ using the Tucker 
 product and Lemma \ref{lemma:tvn_product} to obtain
\begin{equation}\label{stand_statistic}
\mhZ_*\!\!=[\![\mhB_{*} ; \text{d}_2(\tau^2C_2M_1M_1'C_2'),\text{d}_2(M_3M_3'),\text{d}_2(M_4M_4'),\text{d}_2(M_5M_5')]\!] ,
\end{equation}
where $\text{d}_2(A)$ is a diagonal matrix of the inverse square roots of the diagonal entries of $A$. 
Then $\mhZ_* (i,\cdot,\cdot,\cdot)$ has the TVN distribution, with correlation matrices as scale parameters.

\subsection{Supplement to Section \ref{application:LFW}}\label{sec:app:lfw}
 The ethnic origin and age-group assignments of
these images are from the classifier of \citet{Kumaretal09} that
provided values of attributes that are positive for its presence and
negative otherwise, along with magnitudes that 
describe the degree to which the attribute is present or absent (these
magnitudes cannot be compared between different attributes). Since
these attributes are for all age groups and ethnic origins, we have
ambiguous classification for a large majority of the images.
We removed ambiguities in age group and ethnic origin by selecting
images where (for each factor) the 
maximum attribute is positive and the other attributes are
negative.  The genders were as per~\citet{AfifiandAbdelhamed19}'s 
manual assignments of each image. We
considered images only of male or female genders, and discarded those
that have no or both gender assignments.  Although these two filters
reduced the number of images to 5,472, they were very imbalanced, as 37\% of them corresponded to only one of the $2\times 3\times 4=24$
factor-combinations (senior males of European ethnicity). While
our TANOVA model can handle imbalance, it can provide inaccurate
information for the underrepresented factor-combinations. Therefore,
we randomly selected at most 33 images for each factor combination,
reducing the sample size from 5,472 to 605.

We used the inner $151\times 111\times 3$
 portion of the original $250\times 250\times 3$ voxels where the
 third dimension corresponds to the color intensities in RGB format,
 and the logit transformation was applied to every value to match the
 statespace of the normal distribution.
\begin{figure*}[h]
\centering
\hspace{.9in}
\subfloat{
  \includegraphics[width=.15\linewidth]{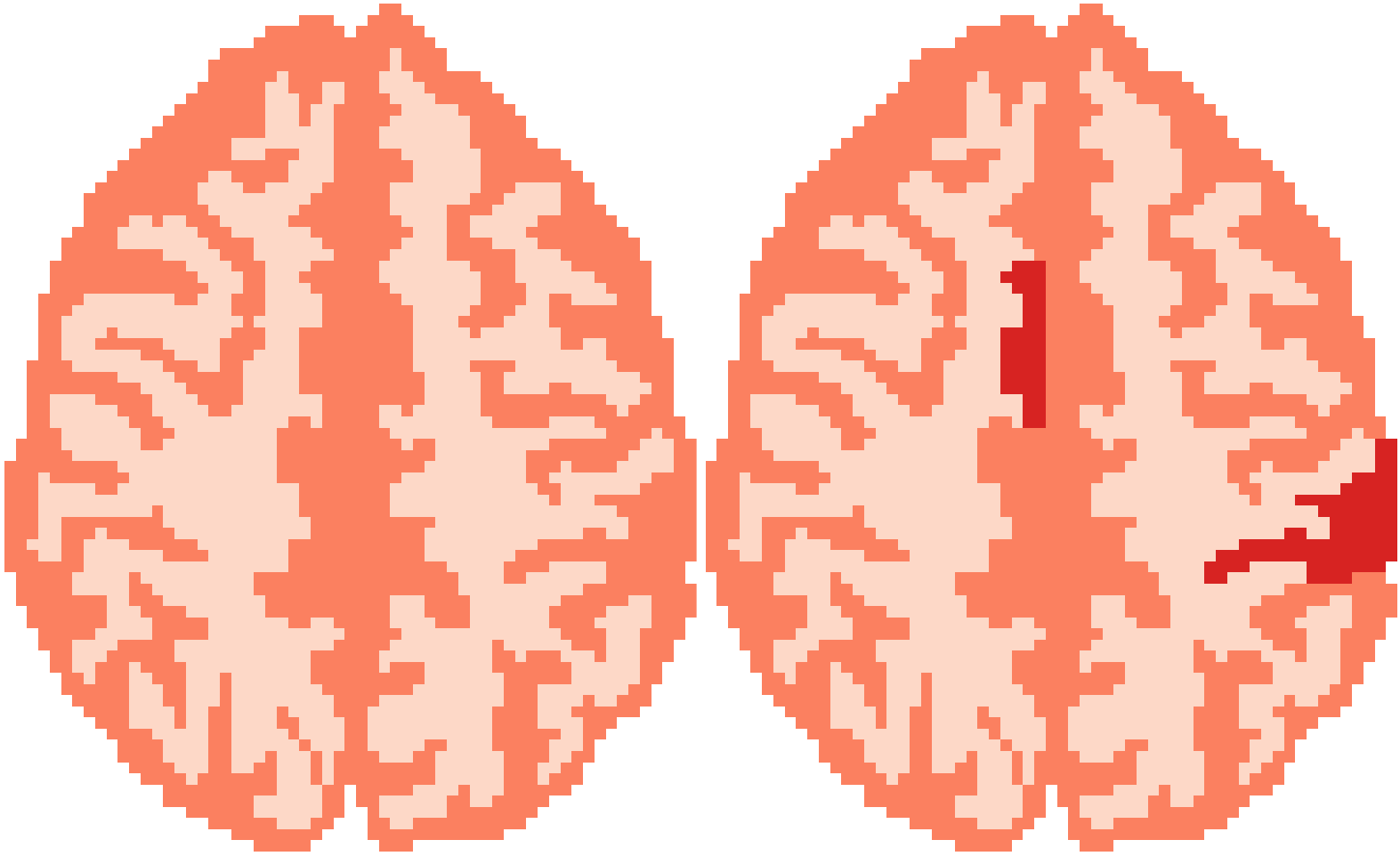}
  }\\
\subfloat{
  \includegraphics[width=.85\linewidth,page=2]{figures/brainsim/brainsim-crop.pdf}
  }\\
  \subfloat{
    \includegraphics[width=.85\linewidth,page=1]{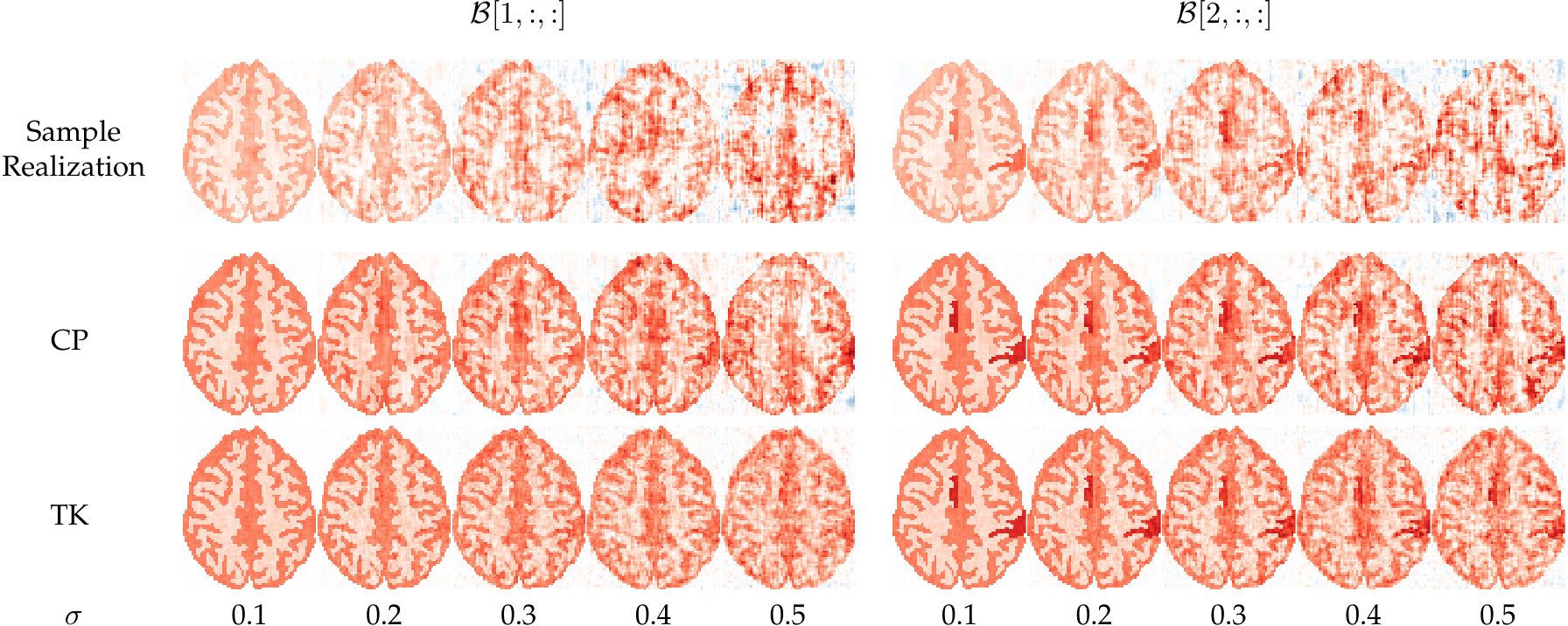}
  }
\caption{The top block displays the modified Hoffman's
  phantom~\citep{almodovarandmaitra19} that formed our true $\mB$ in
  the experiments of Section~\ref{evaluation-arxiv}. The second block
  displays a sample realization (first row) and the $\mhB$ (left
  block corresponds to the first layer of $\mB$, right block
  corresponds to the second later of of $\mB$) using the CP
  (second row) and TK (third row) formats, with optimal ranks
  estimated by BIC, for $\sigma=0.1,0.2,0.3,0.4,0.5$, with data
  generated having modest AR(1) correlations: $\rho_1=0.25,\rho_2=0.35$. The
  bottom block is similar to the middle block, but it displays
  performance in the high AR(1) correlations case: here,
  $\rho_1=0.6,\rho_2=0.8.$ We see that although increasing variance
  and AR(1) correlation lead to noisier sample realizations and
  parameter estimates, the anatomy of brain structure, and the
  activation is largely recovered by both formats, even when we have
  only three realizations under each setting.}
\label{fig:simbrains}
\end{figure*}
\section{Supplement to Section~\ref{sec:discussion}}
\label{evaluation-arxiv}
\citet{papadogeorgouetal21} recently reported blocking behavior in
the direction of the modes of $\mB$ when estimating it via low-rank CP
and TK formats and in the context of a scalar-on-tensor regression
framework. We provide a limited evaluation to see if this
behavior is also exhibited in ToTR, for
instance, in imaging applications, when we use BIC to tune the ranks.
Our evaluation is in 
the context of the balanced TANOVA(1,2) model
\begin{equation*}
Y_{i} =   \langle\bx_{i} | \mB\rangle +E_{i}
,\quad
E_{i} \overset{iid}{\sim} \N_{[61,76]'}(0,\sigma^2,
\Sigma_1,\Sigma_2),
\end{equation*}
with  $\mB$ being a 3-way tensor of size $2\times 61\times 76$, which we
chose from the inner $61\times76$ portion (eliminating all rows and columns
consisting only of background pixels) of the $128\times128$ digitized
Hoffman phantom of~\citep{hoffmanetal90} and recently
adapted by~\citep{almodovarandmaitra19} for use in fMRI simulation experiments. 
The original phantom~\citep{hoffmanetal90} is a digitized
representation of the brain, and the inner ($61\times 76$) portion of this
phantom formed the first layer ({\em i.e.}, $\mB[1,:,:]$) of our
$\mB$, with the background pixels having value 0, the light foreground
pixels having value 1 and the darker foreground pixels having value
2. The second layer $\mB[2,:,:]$ had the same values as $\mB[1,:,:]$,
except for the 138 ``truly activated pixels'' in two disjoint regions
that were introduced by~\citet{almodovarandmaitra19} to evaluate
activation detection methods in simulation settings. At these pixels, $\mB[2,:,:]$ had the
value 3. (Fig.~\ref{fig:simbrains}, top block, displays the two layers of
$\mB$.) The vector $\bx_i$ (where $i=1,2,\dots, 6$)  is a 2D binary
vector that indicates which of the two phantom images corresponds to
$\mathbb{E}(Y_i)$. We let $(\Sigma_1,\Sigma_2)$ be AR(1) correlation
matrices with coefficients $(\rho_1,\rho_2) \in \{ (0.25,0.35)
,(0.6,0.8)\} $ to represent cases with low and high correlations among
the errors and set $\sigma\in\{0.2,0.4,0.6,0.8,1\}$, for five
different noise settings. 

Fig.~\ref{fig:simbrains} (top row of middle and bottom blocks) shows sample realizations  of the model under lower (middle block) and high correlations (bottom block) cases. We estimated $\mB$ assuming CP and TK formats,
with ranks tuned using BIC, as per Section \ref{subsec:rank_choos},
and displayed these $\mhB$s in the second and third rows of the
middle  and lower blocks of Fig.~\ref{fig:simbrains}. 
These estimates show good recovery in $\mB$ in all cases, of both the
putative anatomical and activation regions, even though it is
understandable that this performance depreciates with higher noise and
intra-brain correlation coefficients. 
Lower $\sigma$s yield larger optimal ranks for both TK and CP, while
increasing the AR(1) coefficients increases the model complexity. The
figures also indicate that there is not much appreciable blocking effect in
the direction of the modes of the tensor. 
We however, observed blocking behavior, earlier with the CP and TK
(and also TR) formats when the ranks were set (see Fig.~\ref{fig:camelids}), but not when we allowed BIC to choose these
ranks. We also did not observe such patterns with the TR format and
BIC-tuned ranks in Fig. \ref{fig:lfw}, even though we achieved a
parameter reduction of over 99\%.   Therefore, while further
investigations are needed, our initial evaluations here show that
our BIC-optimized reduced-rank (CP and TK) format ToTR and TANOVA do
not induce unwanted directional effects, therefore showing promise in
imaging and other applications. 

\bibliographystyle{IEEEtran}
\bibliography{tensorontensor-new}

\end{document}